\documentclass[
  aps,
  prx,
  10pt,
  onecolumn,
  tightenlines,
  superscriptaddress,
  nofootinbib,
  longbibliography
]{revtex4-2}

\usepackage{array}[=2016-10-06]
\usepackage{mathtools}
\usepackage{amsmath}
\usepackage[shortlabels]{enumitem}
\usepackage{graphicx,epic,eepic,epsfig,latexsym,verbatim,color}
 
\usepackage{amsfonts}       % blackboard math symbols
\usepackage{nicefrac}       % compact symbols for 1/2, etc.
\usepackage[linesnumbered,ruled,procnumbered]{algorithm2e}
\usepackage{xcolor}

\usepackage{framed}
\definecolor{shadecolor}{rgb}{0.9,0.90,0.9}
\usepackage{bbm}
 
\usepackage{tikz}
\usetikzlibrary{chains}
\usetikzlibrary{fit}

\usepackage{epsfig}
\usetikzlibrary{shapes.symbols,patterns} % for source symbols
\usepackage{pgfplots}
\pgfplotsset{compat=1.18}

\usepackage[strict]{changepage}
\usepackage{hyperref}
\hypersetup{colorlinks=true,citecolor=blue,linkcolor=blue,filecolor=blue,urlcolor=blue,breaklinks=true}

\usepackage{url}
\usepackage{theorem}

\newtheorem{definition}{Definition}
\newtheorem{proposition}{Proposition}
\newtheorem{lemma}[proposition]{Lemma}

\newtheorem{theorem}[proposition]{Theorem}

\newtheorem{corollary}[proposition]{Corollary}

\def\squareforqed{\hbox{\rlap{$\sqcap$}$\sqcup$}}
\def\qed{\ifmmode\squareforqed\else{\unskip\nobreak\hfil
\penalty50\hskip1em\null\nobreak\hfil\squareforqed
\parfillskip=0pt\finalhyphendemerits=0\endgraf}\fi}
\def\endenv{\ifmmode\;\else{\unskip\nobreak\hfil
\penalty50\hskip1em\null\nobreak\hfil\;
\parfillskip=0pt\finalhyphendemerits=0\endgraf}\fi}
\newenvironment{proof}{\noindent \textbf{{Proof~} }}{\hfill $QED$}

\newcounter{remark}

\newcounter{example}

\mathchardef\ordinarycolon\mathcode`\:
\mathcode`\:=\string"8000
\def\vcentcolon{\mathrel{\mathop\ordinarycolon}}
\begingroup \catcode`\:=\active
  \lowercase{\endgroup
  \let :\vcentcolon
  }

\usepackage{cleveref}
\usepackage{graphicx}

\RequirePackage[framemethod=default]{mdframed}
\newmdenv[skipabove=7pt,
skipbelow=7pt,
backgroundcolor=darkblue!15,
innerleftmargin=5pt,
innerrightmargin=5pt,
innertopmargin=5pt,
leftmargin=0cm,
rightmargin=0cm,
innerbottommargin=5pt,
linewidth=1pt]{tBox}

\newmdenv[skipabove=7pt,
skipbelow=7pt,
backgroundcolor=red!15,
innerleftmargin=5pt,
innerrightmargin=5pt,
innertopmargin=5pt,
leftmargin=0cm,
rightmargin=0cm,
innerbottommargin=5pt,
linewidth=1pt]{rBox}

\newmdenv[skipabove=7pt,
skipbelow=7pt,
backgroundcolor=blue2!25,
innerleftmargin=5pt,
innerrightmargin=5pt,
innertopmargin=5pt,
leftmargin=0cm,
rightmargin=0cm,
innerbottommargin=5pt,
linewidth=1pt]{dBox}
\newmdenv[skipabove=7pt,
skipbelow=7pt,
backgroundcolor=darkkblue!15,
innerleftmargin=5pt,
innerrightmargin=5pt,
innertopmargin=5pt,
leftmargin=0cm,
rightmargin=0cm,
innerbottommargin=5pt,
linewidth=1pt]{sBox}
\definecolor{darkblue}{RGB}{0,76,156}
\definecolor{darkkblue}{RGB}{0,0,153}
\definecolor{blue2}{RGB}{102,178,255}
\definecolor{darkred}{RGB}{195,0,0}
\newcommand{\nc}{\newcommand}
\nc{\rnc}{\renewcommand}
\nc{\lbar}[1]{\overline{#1}}
\nc{\bra}[1]{\langle#1|}
\nc{\ket}[1]{|#1\rangle}
\nc{\dketbra}[2]{\vert #1 \rangle \hspace{-.8mm} \rangle \hspace{-.4mm} \langle\hspace{-.8mm}\langle #2 \vert}
\nc{\dbra}[1]{\langle\hspace{-.8mm}\langle #1\vert}
\nc{\dket}[1]{\vert#1\rangle\hspace{-.8mm}\rangle}
\nc{\ketbra}[2]{|#1\rangle\!\langle#2|}
\nc{\braket}[2]{\langle#1|#2\rangle}

\nc{\proj}[1]{| #1\rangle\!\langle #1 |}
\nc{\avg}[1]{\langle#1\rangle}
\nc{\rank}{\operatorname{Rank}}
\nc{\smfrac}[2]{\mbox{$\frac{#1}{#2}$}}
\nc{\tr}{\operatorname{Tr}}
\nc{\ox}{\otimes}
\nc{\dg}{\dagger}
\nc{\dn}{\downarrow}
\nc{\cA}{{\cal A}}
\nc{\cB}{{\cal B}}
\nc{\cC}{{\cal C}}
\nc{\cD}{{\cal D}}
\nc{\cE}{{\cal E}}
\nc{\cF}{{\cal F}}
\nc{\cG}{{\cal G}}
\nc{\cH}{{\cal H}}
\nc{\cI}{{\cal I}}
\nc{\cJ}{{\cal J}}
\nc{\cK}{{\cal K}}
\nc{\cL}{{\cal L}}
\nc{\cM}{{\cal M}}
\nc{\cN}{{\cal N}}
\nc{\cO}{{\cal O}}
\nc{\cP}{{\cal P}}
\nc{\cQ}{{\cal Q}}
\nc{\cR}{{\cal R}}
\nc{\cS}{{\cal S}}
\nc{\cT}{{\cal T}}
\nc{\cU}{{\cal U}}
\nc{\cV}{{\cal V}}
\nc{\cX}{{\cal X}}
\nc{\cY}{{\cal Y}}
\nc{\cZ}{{\cal Z}}
\nc{\cW}{{\cal W}}
\nc{\csupp}{{\operatorname{csupp}}}
\nc{\qsupp}{{\operatorname{qsupp}}}
\nc{\var}{{\operatorname{var}}}
\nc{\rar}{\rightarrow}
\nc{\lrar}{\longrightarrow}
\nc{\polylog}{{\operatorname{polylog}}}
\nc{\wt}{{\operatorname{wt}}}
\nc{\av}[1]{{\left\langle {#1} \right\rangle}}
\nc{\supp}{{\operatorname{supp}}}
\nc{\VComb}{{\widetilde{\cal C}}}
\nc{\VChoi}{{\widetilde{C}}}

\nc{\argmin}{{\operatorname{argmin}}}

\def\x{\xi}

\def\w{\omega}

\nc{\RR}{{{\mathbb R}}}
\nc{\CC}{{{\mathbb C}}}
\nc{\FF}{{{\mathbb F}}}
\nc{\NN}{{{\mathbb N}}}
\nc{\ZZ}{{{\mathbb Z}}}
\nc{\PP}{{{\mathbb P}}}
\nc{\QQ}{{{\mathbb Q}}}
\nc{\UU}{{{\mathbb U}}}
\nc{\EE}{{{\mathbb E}}}
\nc{\id}{{\operatorname{id}}}

\nc{\CHSH}{{\operatorname{CHSH}}}

\nc{\<}{\langle}
\rnc{\>}{\rangle}
\nc{\rU}{\mbox{U}}

\nc{\ob}[1]{#1}

\usepackage{tikz}
\usepackage{hyperref}
\hypersetup{colorlinks=true,citecolor=blue,linkcolor=blue,filecolor=blue,urlcolor=blue,breaklinks=true}

\makeatletter
\def\grd@save@target#1{%
  \def\grd@target{#1}}
\def\grd@save@start#1{%
  \def\grd@start{#1}}
\tikzset{
  grid with coordinates/.style={
    to path={%
      \pgfextra{%
        \edef\grd@@target{(\tikztotarget)}%
        \tikz@scan@one@point\grd@save@target\grd@@target\relax
        \edef\grd@@start{(\tikztostart)}%
        \tikz@scan@one@point\grd@save@start\grd@@start\relax
        \draw[minor help lines,magenta] (\tikztostart) grid (\tikztotarget);
        \draw[major help lines] (\tikztostart) grid (\tikztotarget);
        \grd@start
        \pgfmathsetmacro{\grd@xa}{\the\pgf@x/1cm}
        \pgfmathsetmacro{\grd@ya}{\the\pgf@y/1cm}
        \grd@target
        \pgfmathsetmacro{\grd@xb}{\the\pgf@x/1cm}
        \pgfmathsetmacro{\grd@yb}{\the\pgf@y/1cm}
        \pgfmathsetmacro{\grd@xc}{\grd@xa + \pgfkeysvalueof{/tikz/grid with coordinates/major step}}
        \pgfmathsetmacro{\grd@yc}{\grd@ya + \pgfkeysvalueof{/tikz/grid with coordinates/major step}}
        \foreach \x in {\grd@xa,\grd@xc,...,\grd@xb}
        \node[anchor=north] at (\x,\grd@ya) {\pgfmathprintnumber{\x}};
        \foreach \y in {\grd@ya,\grd@yc,...,\grd@yb}
        \node[anchor=east] at (\grd@xa,\y) {\pgfmathprintnumber{\y}};
      }
    }
  },
  minor help lines/.style={
    help lines,
    step=\pgfkeysvalueof{/tikz/grid with coordinates/minor step}
  },
  major help lines/.style={
    help lines,
    line width=\pgfkeysvalueof{/tikz/grid with coordinates/major line width},
    step=\pgfkeysvalueof{/tikz/grid with coordinates/major step}
  },
  grid with coordinates/.cd,
  minor step/.initial=.2,
  major step/.initial=1,
  major line width/.initial=2pt,
}
\makeatother

\usepackage{thmtools}
\usepackage{thm-restate}
\usepackage{etoolbox}
\makeatletter
\def\problem@s{}
\newcounter{problems@cnt}

\newcommand{\allproblems}{\problem@s}
\makeatother

\usepackage{tikz}
\usetikzlibrary{positioning}
\usetikzlibrary{shapes.geometric}
\usetikzlibrary{calc}
\definecolor{tensorblue}{rgb}{0.8,0.9,1}
\tikzset{ten/.style={fill=tensorblue}}
\usepackage{amssymb,tcolorbox,tikzit}
\usepackage{booktabs,tabularx}
\usepackage[utf8]{inputenc}
\usepackage[T1]{fontenc}
\usetikzlibrary{quantikz2}
\hypersetup{colorlinks=true,citecolor=blue,linkcolor=blue,filecolor=blue,urlcolor=blue,breaklinks=true}

\renewcommand{\braket}[1]{\ensuremath{\left\langle #1 \right\rangle}}

\begin{document}

\title{Entanglement Cost of Optimal Distributed Quantum State Purification}
\author{Jiayi Zhao}
\affiliation{Thrust of Artificial Intelligence, Information Hub, The Hong Kong University of Science and Technology (Guangzhou), Guangdong 511453, China}
\author{Chengkai Zhu}
\affiliation{Thrust of Artificial Intelligence, Information Hub, The Hong Kong University of Science and Technology (Guangzhou), Guangdong 511453, China}
\affiliation{QudeLeap Research, Shanghai 200030, China}
\author{Xin Wang}
\email{felixxinwang@hkust-gz.edu.cn}
\affiliation{Thrust of Artificial Intelligence, Information Hub, The Hong Kong University of Science and Technology (Guangzhou), Guangdong 511453, China}
\author{Ge Bai}
\email{gebai@hkust-gz.edu.cn}
\affiliation{Thrust of Artificial Intelligence, Information Hub, The Hong Kong University of Science and Technology (Guangzhou), Guangdong 511453, China}
\date{\today}

\begin{abstract}

We determine the preshared entanglement required for spatially separated parties, restricted to local operations and classical communication, to attain globally optimal probabilistic two-copy purification of arbitrary bipartite pure states under depolarizing noise. In every local dimension $d\ge 2$, one shared maximally entangled qubit pair suffices: local controlled-SWAP operations exactly reproduce the globally optimal successful transformation. Conversely, any finite-dimensional preshared resource state that attains the same benchmark, even under a positive-partial-transpose relaxation, must have entanglement of formation at least one ebit. For pure resources with one ebit, or for two-qubit resources including mixed states, attaining the benchmark is possible only for states with exactly two nonzero Schmidt weights, both equal to $1/2$, namely those locally unitarily equivalent to a maximally entangled qubit pair. These findings provide a benchmark for evaluating the entanglement demands of noise management in quantum networks and modular quantum computers.%, helping distinguish fundamental resource requirements from implementation overhead.

\end{abstract}  

\maketitle

\section{Introduction}

Reliable quantum information processing requires high-fidelity preparation,
transmission, storage, and manipulation of quantum states.  Imperfect control
introduces errors in state preparation and quantum operations, while residual
coupling to the environment causes decoherence
\cite{zurek2003decoherence,preskill2018quantum}.  Both effects limit the
reliability of quantum computation and communication.  Fault-tolerant quantum
error correction protects quantum information by maintaining a logical state
in encoded form throughout noisy storage or processing, but requires redundant
encoding, repeated syndrome extraction, and high-fidelity control
\cite{terhal2015quantum}.  Quantum state purification
addresses a complementary situation: several copies of the same pure target
state have already been affected by noise, and the goal is to obtain fewer
outputs with higher fidelity to the target
\cite{barenco1997stabilization,cirac1999optimal,keyl2001rate,
fiuravsek2004optimal,yao2025protocols}.
This raises the questions of how much the target fidelity can be improved
from a fixed number of noisy preparations and what operations and resources
are required.

Formally, a noisy process is described by a completely positive
trace-preserving (CPTP) map \(\mathcal N\), also called a quantum channel,
which transforms a pure target \(\psi=\ket\psi\!\bra\psi\) into the noisy state
\(\mathcal N(\psi)\).  An \(n\to1\) probabilistic purification protocol is
described by a completely positive trace-nonincreasing (CPTN) map
\(\mathcal E\), also called a quantum operation. Given
\(\mathcal N(\psi)^{\otimes n}\), the protocol succeeds with probability
\(p_\psi\) and, conditioned on success, produces the output state
\(\sigma_\psi\), where
\[
p_\psi:=\operatorname{Tr}\!\left[
\mathcal E\bigl(\mathcal N(\psi)^{\otimes n}\bigr)\right],
\qquad
\sigma_\psi:=
\frac{\mathcal E\bigl(\mathcal N(\psi)^{\otimes n}\bigr)}{p_\psi}.
\]
The output state \(\sigma_\psi\) is defined when \(p_\psi>0\).  Here
\(\mathrm F(\rho,\tau):=\|\sqrt\rho\sqrt\tau\|_1^2\) denotes quantum fidelity.
We use \(\mathrm F(\psi,\mathcal N(\psi))\), the fidelity of one noisy copy
with the target, as the baseline fidelity.  When \(\mathcal E\) may be chosen
using a classical description of \(\psi\), we call the task
\textit{targeted state purification}.  In \textit{blind state purification},
only the target set \(\mathcal S\) is specified, and \(\mathcal E\) must be
fixed independently of the particular \(\psi\in\mathcal S\)
\cite{zhao2026power}.  The map \(\mathcal E\) purifies \(\mathcal S\)
if \(p_\psi>0\) and
\[
\mathrm F(\psi,\sigma_\psi)
\ge
\mathrm F\bigl(\psi,\mathcal N(\psi)\bigr)
\]
for every \(\psi\in\mathcal S\).  The purification is strict if the inequality
is strict for at least one \(\psi\in\mathcal S\).

The performance of purification depends on the permitted operations on the noisy copies.  In quantum networks and modular quantum computers,
each party has direct access only to its local quantum registers, and arbitrary
joint operations across parties are not freely available
\cite{caleffi2024distributed,main2025distributed}.
Without a shared resource, the operations are thus
restricted to local quantum operations coordinated through classical
communication (LOCC) \cite{chitambar2014everything}. 
%Global state purification permits arbitrary joint operations on all noisy copies.   Such joint 
%This motivates studying whether multiple noisy copies of a distributed state can yield a higher-fidelity output under these locality constraints \cite{zhao2026power}.
For global state purification, where arbitrary joint operations
on the noisy copies are allowed, previous studies have shown
that probabilistic processing of copies affected by depolarizing
noise can produce an output with fidelity above the baseline
\cite{cirac1999optimal,fiuravsek2004optimal,yao2025protocols}.
Depolarizing noise can arise as an effective noise model when
randomized unitary control removes the directional dependence
of noise through twirling over the full unitary
group~\cite{horodecki1999general}.
Under the same noise model, however,
Zhao et al.~\cite{zhao2026power} proved that no \(2\to1\)
LOCC protocol can achieve strict blind purification, even
with postselection, for the set of all pure two-qubit targets.
This separation motivates a quantitative question: how much
preshared entanglement is required for an LOCC protocol to
match the optimal purification performance achievable by
global CPTN maps?

In this work, we answer this question for the \(2\to1\) blind purification of
arbitrary pure states on \(\mathbb C^d\otimes\mathbb C^d\).  Following earlier
studies of quantum state purification
\cite{cirac1999optimal,keyl2001rate,yao2025protocols}, the two available copies
are independently affected by global depolarizing noise of known strength
\(\gamma\). %We allow postselection because the global protocols whose performance we seek to match may themselves be probabilistic. Deterministic protocols are included as the special case of unit success probability. 
As the figure of merit, we compare protocols using the success-probability-weighted
fidelity gain over the baseline, averaged with respect to the Haar measure over
all pure target states. % The global benchmark is the maximum of this average fidelity gain over all CPTN maps acting jointly on the two noisy copies.

For every local dimension \(d\ge2\) and every \(0<\gamma<1\), we provide a probabilistic \(2\to1\) LOCC protocol assisted by one Einstein--Podolsky--Rosen (EPR) pair, a maximally entangled
state of two qubits. The protocol matches the performance of the optimal global CPTN map.  Conversely, any finite-dimensional preshared resource state that enables a finite-round LOCC
protocol to achieve the global CPTN optimum, even under positive-partial-transpose relaxation, must have entanglement of formation
at least one ebit.  For pure resources with one ebit, or for two-qubit resources including
mixed states, the global optimum is attainable if and only if the resource
is locally unitarily equivalent to an EPR pair.

\section{Problem Setting}\label{sec:task_operational_model}

An \(n\to1\) distributed purification protocol takes \(n\) bipartite inputs
shared between Alice and Bob and produces one bipartite output.  We denote the
\(i\)th input by \(A_iB_i\), with Alice holding \(A_i\) and Bob holding \(B_i\).
Alice's full input is \(A^n:=A_1\cdots A_n\), Bob's is
\(B^n:=B_1\cdots B_n\), and the protocol maps \(A^nB^n\) to \(A'B'\), with
Alice holding \(A'\) and Bob holding \(B'\).  Each input pair \(A_iB_i\) and
the output \(A'B'\) has a Hilbert space isomorphic to
\(\mathcal H_A\otimes\mathcal H_B\), where
\(\mathcal H_A\simeq\mathcal H_B\simeq\mathbb C^d\).

Let \(\mathcal S\) be a set of bipartite target states on \(AB\), and let
\(\psi\in\mathcal S\) denote an arbitrary target.  A fixed CPTP noise channel
\(\mathcal N\) acts independently on each copy of \(\psi\), so
the input to the protocol is \(\mathcal N(\psi)^{\otimes n}\). To describe probabilistic transformations between finite-dimensional systems,
define
\[
\mathrm{CPTN}
:=
\left\{
\mathcal E\ \middle|\
\begin{array}{l}
\mathcal E\text{ is completely positive, and}\\
\operatorname{Tr}[\mathcal E(Q)]\le\operatorname{Tr}[Q]
\text{ for every }Q\succeq0
\end{array}
\right\}.
\]
We write \(\mathrm{CPTP}\subset\mathrm{CPTN}\) for the set of completely positive
trace-preserving maps, obtained by replacing the trace inequality above with
equality for every \(Q\succeq0\).  When the systems need to be explicit,
\(\mathcal E_{X\to Y}\) denotes a map whose input is an operator on \(X\) and
whose output is an operator on \(Y\).  We denote by \(\mathrm{LOCC}\) the set of maps implementable
in finitely many rounds, where each quantum operation acts only on the registers
held by one party and later operations may depend on communicated classical
outcomes \cite{chitambar2014everything}. Thus, throughout this work, a protocol is a quantum operation in
\(\mathrm{LOCC}\cap\mathrm{CPTN}\), i.e., a trace-nonincreasing quantum operation implementable by finite-round LOCC, acting on $A^n$, $B^n$ and possible preshared resource.

\begin{definition}[Probabilistic distributed purification~{\cite{zhao2026power}}]
\label{def:distributed-purification}
For a set \(\mathcal S\) of target states on \(AB\) and a CPTP noise channel
\(\mathcal N\), an \(n\to1\) map
\(\mathcal E_{A^nB^n\to A'B'}\in
\mathrm{LOCC}\cap\mathrm{CPTN}\) is a probabilistic distributed purification
protocol for \(\mathcal S\) under \(\mathcal N\) if, for every
\(\psi\in\mathcal S\), the unnormalized output
\(\widehat{\sigma}_\psi:=\mathcal E(\mathcal N(\psi)^{\otimes n})\) has
positive trace \(p_\psi:=\operatorname{Tr}\widehat{\sigma}_\psi\), and the
conditional output \(\sigma_\psi:=\widehat{\sigma}_\psi/p_\psi\) satisfies
\[
\mathrm F(\psi,\sigma_\psi)
\ge
\mathrm F\!\left(\psi,\mathcal N(\psi)\right).
\]
The protocol is \emph{strict} on \(\mathcal S\) if the displayed
inequality is strict for at least one \(\psi\in\mathcal S\).
\end{definition}

In order to quantify preshared entanglement in distributed purification, % is needed for distributed purification to match unrestricted global processing, 
we extend Definition~\ref{def:distributed-purification}
to allow Alice and Bob to share a finite-dimensional bipartite resource state, as
illustrated in Fig.~\ref{fig:purification}.

\begin{definition}[Probabilistic distributed purification with a preshared
resource state]
\label{def:resource-assisted-purification}
Let \(\omega_{A_{\mathrm r}B_{\mathrm r}}\) be a finite-dimensional bipartite
resource state preshared by Alice and Bob.  For a set \(\mathcal S\) of target
states on \(AB\) and a CPTP noise channel \(\mathcal N\), an \((n+1)\to1\) map
\[
\mathcal E_{A^nA_{\mathrm r}B^nB_{\mathrm r}\to A'B'}
\in\mathrm{LOCC}\cap\mathrm{CPTN}
\]
is a probabilistic distributed purification protocol for \(\mathcal S\) under
\(\mathcal N\) using \(\omega_{A_{\mathrm r}B_{\mathrm r}}\) if, for every
\(\psi\in\mathcal S\), the unnormalized output
\(\widehat{\sigma}_\psi:=\mathcal E(\mathcal N(\psi)^{\otimes n}\otimes
\omega_{A_{\mathrm r}B_{\mathrm r}})\) has positive trace
\(p_\psi:=\operatorname{Tr}\widehat{\sigma}_\psi\), and the conditional output
\(\sigma_\psi:=\widehat{\sigma}_\psi/p_\psi\) satisfies
\[
\mathrm F(\psi,\sigma_\psi)
\ge
\mathrm F\!\left(\psi,\mathcal N(\psi)\right).
\]
The protocol is \emph{strict} on \(\mathcal S\) if the displayed
inequality is strict for at least one \(\psi\in\mathcal S\).
\end{definition}

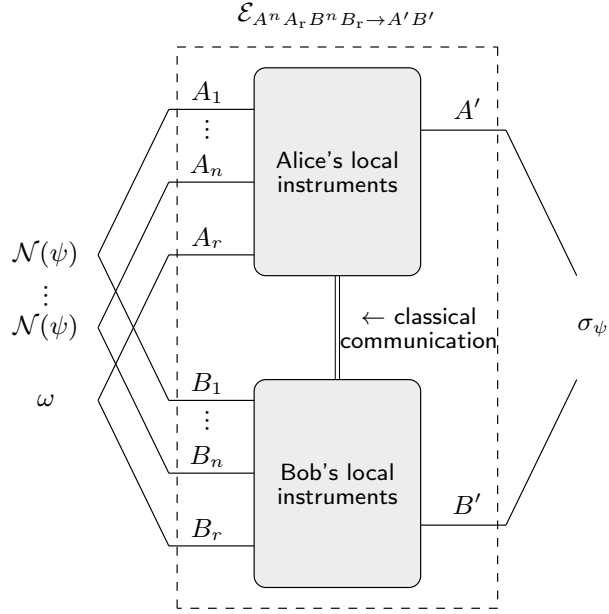
\begin{figure}[htbp]
\centering
{\sffamily
\fontsize{12pt}{15pt}\selectfont
\scalebox{1.1}{\begin{tikzpicture}
\fontsize{9}{11}
	\begin{pgfonlayer}{nodelayer}
		\node [style=none, transform shape, draw=black, rounded corners=4pt, fill=gray, fill opacity=0.14, minimum width=4cm, minimum height=5cm, inner sep=0pt] (72) at (8.25, 10.25) {};
		\node [style=none, transform shape, draw=black, rounded corners=4pt, fill=gray, fill opacity=0.14, minimum width=4cm, minimum height=5cm, inner sep=0pt] (73) at (8.25, 2.75) {};
		\node [style=none] (0) at (6.25, 12.75) {};
		\node [style=none] (1) at (6.25, 7.75) {};
		\node [style=none] (2) at (10.25, 7.75) {};
		\node [style=none] (3) at (10.25, 12.75) {};
		\node [style=none] (4) at (6.25, 5.25) {};
		\node [style=none] (5) at (6.25, 0.25) {};
		\node [style=none] (6) at (10.25, 0.25) {};
		\node [style=none] (7) at (10.25, 5.25) {};
		\node [style=none] (8) at (8.25, 7.75) {};
		\node [style=none] (9) at (8.375, 7.75) {};
		\node [style=none] (10) at (8.25, 5.25) {};
		\node [style=none] (11) at (8.375, 5.25) {};
		\node [style=none] (12) at (6.25, 11.75) {};
		\node [style=none] (13) at (6.25, 10) {};
		\node [style=none] (22) at (6.25, 4.75) {};
		\node [style=none] (23) at (6.25, 3) {};
		\node [style=none] (24) at (4.2, 11.75) {};
		\node [style=none] (25) at (4.2, 10) {};
		\node [style=none] (26) at (4.2, 4.75) {};
		\node [style=none] (27) at (4.2, 3) {};
		\node [style=none] (28) at (2.5, 8.25) {};
		\node [style=none] (29) at (2.5, 6.5) {};
		\node [style=none] (30) at (10.25, 11.25) {};
		\node [style=none] (31) at (10.25, 1.75) {};
		\node [style=none] (32) at (12.3, 11.25) {};
		\node [style=none] (33) at (14, 7.75) {};
		\node [style=none] (34) at (12.3, 1.75) {};
		\node [style=none] (35) at (14, 5.25) {};
		\node [style=none] (36) at (6.25, 8.25) {};
		\node [style=none] (37) at (6.25, 1.25) {};
		\node [style=none] (38) at (4.2, 8.25) {};
		\node [style=none] (39) at (4.2, 1.25) {};
		\node [style=none] (40) at (2.5, 4.75) {};
		\node [style=none] (41) at (1.25, 8.25) {$\mathcal{N}(\psi)$};
		\node [style=none] (42) at (1.25, 6.5) {$\mathcal{N}(\psi)$};
		\node [style=none] (43) at (1.25, 4.75) {$\omega$};
		\node [style=none] (44) at (14.4, 6.5) {$\sigma_\psi$};
		\node [style=none] (45) at (4.4, 13.25) {};
		\node [style=none] (46) at (12.1, 13.25) {};
		\node [style=none] (47) at (12.1, -0.25) {};
		\node [style=none] (48) at (4.4, -0.25) {};
		\node [style=none] (49) at (5.75, 6.5) {};
		\node [style=none] (50) at (10.75, 6.5) {};
		\node [style=none] (51) at (8.25, 14) {$\mathcal E_{A^nA_{\mathrm r}B^nB_{\mathrm r}\to A'B'}$};
		\node [style=none, anchor=south, yshift=0.7mm] (52) at (5.125, 11.75) {$A_1$};
		\node [style=none, anchor=south, yshift=0.7mm] (53) at (5.125, 10) {$A_n$};
		\node [style=none, anchor=south, yshift=0.7mm] (54) at (5.125, 8.25) {$A_r$};
		\node [style=none, anchor=south, yshift=0.7mm] (55) at (5.125, 4.75) {$B_1$};
		\node [style=none, anchor=south, yshift=0.7mm] (56) at (5.125, 3) {$B_n$};
		\node [style=none, anchor=south, yshift=0.7mm] (57) at (5.125, 1.25) {$B_r$};
		\node [style=none, anchor=south] (58) at (11.375, 11.5) {$A^{\prime}$};
		\node [style=none, anchor=south] (59) at (11.375, 2) {$B^{\prime}$};
		\node [style=none, inner sep=0pt] (60) at (5.125, 11.25) {$\cdot$};
		\node [style=none, inner sep=0pt] (61) at (5.125, 4.25) {$\cdot$};
		\node [style=none, inner sep=0pt] (62) at (5.125, 11.45) {$\cdot$};
		\node [style=none, inner sep=0pt] (63) at (5.125, 11.05) {$\cdot$};
		\node [style=none, inner sep=0pt] (64) at (5.125, 4.45) {$\cdot$};
		\node [style=none, inner sep=0pt] (65) at (5.125, 4.05) {$\cdot$};
		\node [style=none, inner sep=0pt] (66) at (1.25, 7.25) {$\cdot$};
		\node [style=none, inner sep=0pt] (67) at (1.25, 7.449999999999999) {$\cdot$};
		\node [style=none, inner sep=0pt] (68) at (1.25, 7.050000000000001) {$\cdot$};
		\node [style=none] (69) at (8.25, 10.25) {{\fontsize{8}{10}\selectfont\shortstack{Alice's local\\instruments}}};
		\node [style=none] (70) at (8.25, 2.75) {{\fontsize{8}{10}\selectfont\shortstack{Bob's local\\instruments}}};
		\node [style=none, anchor=west] (71) at (8.3, 6.5) {{\fontsize{8}{10}\selectfont\shortstack{$\leftarrow$ classical\\communication}}};
	\end{pgfonlayer}
	\begin{pgfonlayer}{edgelayer}
		\draw [double, double distance=1.1pt, line width=0.35pt] (8.center) to (10.center);
		\draw (24.center) to (12.center);
		\draw (25.center) to (13.center);
		\draw (22.center) to (26.center);
		\draw (23.center) to (27.center);
		\draw (24.center) to (28.center);
		\draw (28.center) to (26.center);
		\draw (25.center) to (29.center);
		\draw (29.center) to (27.center);
		\draw (30.center) to (32.center);
		\draw (32.center) to (33.center);
		\draw (31.center) to (34.center);
		\draw (34.center) to (35.center);
		\draw (38.center) to (36.center);
		\draw (38.center) to (40.center);
		\draw (40.center) to (39.center);
		\draw (39.center) to (37.center);
		\draw [dashed] (45.center) to (46.center);
		\draw [dashed] (46.center) to (47.center);
		\draw [dashed] (47.center) to (48.center);
		\draw [dashed] (48.center) to (45.center);
	\end{pgfonlayer}
\end{tikzpicture}}
}
\caption{Distributed \(n\to1\) purification with a preshared resource state.
Alice and Bob process \(n\) noisy copies together with a finite-dimensional
preshared state \(\omega_{A_{\mathrm r}B_{\mathrm r}}\) by finite-round LOCC.
The trace of the unnormalized output
\(\widehat{\sigma}_\psi\) is the success probability.}
\label{fig:purification}
\end{figure}

If \(\omega\) is separable, Alice and Bob can prepare it by LOCC, so allowing
it does not enlarge the class of protocols in
Definition~\ref{def:distributed-purification}.  When \(\omega\) is entangled,
we call the protocol entanglement-assisted. We quantify the entanglement in the preshared resource state by its entanglement
of formation~\cite{bennett1996mixed,horodecki2009quantum}, defined as follows.

\begin{definition}[Entanglement of formation]
The entanglement of formation of a bipartite resource state
\(\omega_{A_{\mathrm r}B_{\mathrm r}}\) is
\[
E_{\mathrm F}(\omega_{A_{\mathrm r}B_{\mathrm r}})
:=
\inf_{\substack{
\omega_{A_{\mathrm r}B_{\mathrm r}}
=\sum_a q_a\ket{\eta_a}\!\bra{\eta_a}
}}
\sum_a q_a
S\!\left(\operatorname{Tr}_{B_{\mathrm r}}
\ket{\eta_a}\!\bra{\eta_a}\right),
\]
where the infimum is over all pure-state ensembles of
\(\omega_{A_{\mathrm r}B_{\mathrm r}}\), and
\(S(\rho):=-\operatorname{Tr}(\rho\log_2\rho)\).  For a pure resource \(\ket\eta\), let
\(p_0\ge p_1\ge\cdots\ge0\) be the eigenvalues of
\(\operatorname{Tr}_{B_{\mathrm r}}\ket\eta\!\bra\eta\),
with zeros appended as needed.
These eigenvalues are called its Schmidt weights, and their
ordered list is called its Schmidt spectrum.
The entanglement of formation then reduces to
\[
E_{\mathrm F}(\ket\eta\!\bra\eta)
=:E(\eta)
=-\sum_jp_j\log_2p_j.
\]
\end{definition}

We take \(E_{\mathrm F}(\omega_{A_{\mathrm r}B_{\mathrm r}})\) as the input
resource cost, independently of the protocol outcome.  The resource registers
may have arbitrary finite dimensions, with no assumed relation to \(d\).

We now specialize to the setting of our main results.  We take \(n=2\), and the
target set is the full set of bipartite pure states on \(AB\),
\[
\mathcal S_P
:=
\left\{
\ket\phi\!\bra\phi:
\ket\phi\in\mathcal H_A\otimes\mathcal H_B,\ 
\langle\phi|\phi\rangle=1
\right\}.
\]
We write \(\psi:=\ket\psi\!\bra\psi\in\mathcal S_P\) for an arbitrary target
and set \(D:=d^2\), the dimension of one bipartite copy.  Each copy is
independently affected by the global depolarizing channel
\[
\mathcal N^\gamma(\rho)
:=
(1-\gamma)\rho+\frac{\gamma}{D}I_D,
\qquad
0<\gamma<1.
\]
The resulting two-copy input is
\[
\mathcal N^\gamma(\psi)_{A_1B_1}\otimes
\mathcal N^\gamma(\psi)_{A_2B_2}.
\]
Alice holds \(A_1A_2\), and Bob holds \(B_1B_2\).
Let \(\mu_P\) denote the normalized Haar probability measure on \(\mathcal S_P\).

To quantify the purification improvement for each target \(\psi\), define
\[
g_\psi(\mathcal E):=\langle\psi|\widehat{\sigma}_\psi|\psi\rangle
-p_\psi\,\mathrm F\!\left(\psi,\mathcal N^\gamma(\psi)\right).
\]
For \(p_\psi>0\), this equals the product of success probability and fidelity gain over the baseline
\[
p_\psi\!\left[
\mathrm F(\psi,\sigma_\psi)
-\mathrm F(\psi,\mathcal N^\gamma(\psi))
\right].
\]
If \(p_\psi=0\), positivity of \(\widehat{\sigma}_\psi\) implies
\(\widehat{\sigma}_\psi=0\), and hence \(g_\psi(\mathcal E)=0\).
We compare protocols using the Haar average of
\(g_\psi(\mathcal E)\).

\begin{definition}[Optimal global CPTN benchmark]
\label{def:global_cptn_benchmark}
For the two-copy purification task on \(\mathcal S_P\) under
the depolarizing channel \(\mathcal N^\gamma\), the optimal
global CPTN benchmark is
\[
G_\star(D,\gamma)
:=
\sup_{\mathcal E_{A_1B_1A_2B_2\to A'B'}\in\mathrm{CPTN}}
\mathbb E_{\psi\sim\mu_P}\!\left[g_\psi(\mathcal E)\right].
\]
\end{definition}

For any fixed resource state, a probabilistic LOCC protocol induces a CPTN map on the two noisy copies.  Its performance is therefore upper-bounded by
\(G_\star(D,\gamma)\).  We seek the smallest value of
\(E_{\mathrm F}(\omega_{A_{\mathrm r}B_{\mathrm r}})\) among such resource
states that enable an LOCC protocol to attain \(G_\star(D,\gamma)\).

\section{Main Results}\label{sec:main_epr_threshold}

In this section, we determine the preshared entanglement needed for finite-round
LOCC to attain the optimal global CPTN benchmark.  The minimum entanglement of
formation of the resource state is one ebit for every local dimension \(d\ge2\)
and every depolarizing parameter \(0<\gamma<1\).  First, we evaluate
\(G_\star(D,\gamma)\) and show that a finite-round LOCC protocol assisted by one
EPR pair attains it.  Second, we prove that every finite-dimensional preshared
resource state that enables LOCC to attain
\(G_\star(D,\gamma)\) has entanglement of formation at least one ebit.  Third,
we characterize the pure one-ebit resources that attain the benchmark: their
Schmidt spectrum is \((1/2,1/2,0,\ldots)\), or equivalently, they are locally
unitarily equivalent to an EPR pair.

The following LOCC algorithm, assisted by one EPR pair, attains the optimal global CPTN benchmark, as shown in Theorem~\ref{thm:main_global_cptn}. Its circuit is shown in
Figure~\ref{fig:full_epr_locc}.

%The construction uses one resource qubit at each party and is specified in Algorithm~\ref{alg:one_epr_purification}; 

\begin{algorithm}[htbp]
\caption{One-EPR distributed purification algorithm}
\label{alg:one_epr_purification}
\KwIn{The two-copy noisy input
\(\mathcal N^\gamma(\psi)_{A_1B_1}\otimes
\mathcal N^\gamma(\psi)_{A_2B_2}\), and a preshared EPR pair
\(\ket{\Phi_2}_{A_{\mathrm r}B_{\mathrm r}}
:=\bigl(\ket{00}+\ket{11}\bigr)/\sqrt2\).}
\KwOut{A success or failure flag and, on success, the retained state relabeled
as the output on \(A'B'\).}
Alice applies the controlled swap of \(A_1\) and \(A_2\), controlled by
\(A_{\mathrm r}\)\;
Bob applies the controlled swap of \(B_1\) and \(B_2\), controlled by
\(B_{\mathrm r}\)\;
Alice and Bob apply Hadamard gates to \(A_{\mathrm r}\) and
\(B_{\mathrm r}\), respectively\;
They measure \(A_{\mathrm r}\) and \(B_{\mathrm r}\) in the computational
basis, obtain \(c_A,c_B\in\{0,1\}\), and communicate their outcomes\;
\eIf{\(c_A=c_B\)}{
Discard \(A_2B_2\), relabel the retained pair \(A_1B_1\) as \(A'B'\), and
return it\;
}{
Return failure\;
}
\end{algorithm}

\begin{theorem}[Optimality of the one-EPR algorithm]
\label{thm:main_global_cptn}
For \(d\ge2\) and \(0<\gamma<1\),
Algorithm~\ref{alg:one_epr_purification} is an entanglement-assisted LOCC
protocol that attains the optimal global CPTN benchmark for the
\(2\to1\) purification task on \(\mathcal S_P\) under the depolarizing channel
\(\mathcal N^\gamma\):
\[
G_\star(D,\gamma)
=
\frac12
\left(1-\gamma+\frac{\gamma}{D}\right)
\gamma(1-\gamma)
\left(1-\frac1D\right).
\]
Here \(D=d^2\) is the dimension of one bipartite copy.
\end{theorem}

The direct calculation in
Appendix~\ref{app:one_epr_protocol_calculation} gives the displayed gain for
Algorithm~\ref{alg:one_epr_purification}.
The global SDP analysis in
Appendix~\ref{app:high_rank_entropy_lower_bound} shows that the same value
upper-bounds every global CPTN map.  The algorithm therefore attains the global
optimum.

\begin{figure}[htbp]
\centering
\begin{quantikz}[
    row sep=0.42cm, column sep=0.65cm, thin lines,
    execute at end picture={
        \coordinate (outcome-a) at (\tikzcdmatrixname-3-5.center);
        \coordinate (outcome-b) at (\tikzcdmatrixname-6-5.center);
        \node[anchor=east, inner sep=3pt] (epr-source)
            at ($(\tikzcdmatrixname-3-1.west)!0.5!
                  (\tikzcdmatrixname-6-1.west)+(-1.20,0)$)
            {$\ket{\Phi_2}_{A_{\mathrm r}B_{\mathrm r}}$};
        \draw[line width=0.3pt]
            (epr-source.north east) -- (\tikzcdmatrixname-3-1.west);
        \draw[line width=0.3pt]
            (epr-source.south east) -- (\tikzcdmatrixname-6-1.west);
        \node[anchor=south west, inner sep=1pt]
            at ([xshift=1pt,yshift=2pt]\tikzcdmatrixname-3-1.west)
            {$A_{\mathrm r}$};
        \node[anchor=south west, inner sep=1pt]
            at ([xshift=1pt,yshift=2pt]\tikzcdmatrixname-6-1.west)
            {$B_{\mathrm r}$};
        \tikzset{
            classical bit/.style={draw=black, double, double distance=0.8pt,
                line width=0.3pt},
            classical arrow/.style={classical bit,
                -{Stealth[length=4pt,width=3.2pt,line width=0.3pt]}},
            message arrow/.style={draw=black, double,double distance=0.8pt,
                line width=0.3pt,
                -{Stealth[length=4pt,width=3.2pt]}}
        }
        \coordinate (fork-a) at ($(outcome-a)+(0.80,0)$);
        \coordinate (fork-b) at ($(outcome-b)+(0.80,0)$);
        \node[draw, rounded corners=2pt, align=center, font=\small,
              anchor=west, inner sep=5pt] (alice-check)
            at ($(outcome-a)+(3.0,-0.25)$)
            {\textbf{Alice accepts}\\ iff $c_A=c_B$};
        \node[draw, rounded corners=2pt, align=center, font=\small,
              anchor=west, inner sep=5pt] (bob-check)
            at ($(outcome-b)+(3.0,0.25)$)
            {\textbf{Bob accepts}\\ iff $c_A=c_B$};
        \draw[classical bit]
            (\tikzcdmatrixname-3-4.east) -- (fork-a)
            node[midway, below=2pt] {$c_A$};
        \draw[classical bit]
            (\tikzcdmatrixname-6-4.east) -- (fork-b)
            node[midway, below=2pt] {$c_B$};
        \draw[classical arrow]
            (fork-a) -- ([yshift=7pt]alice-check.west);
        \draw[classical arrow]
            (fork-b) -- ([yshift=-7pt]bob-check.west);
        \draw[message arrow]
            (fork-a) to[out=-40,in=170]
            node[pos=0.42, above=3pt] {$c_A$}
            ([yshift=7pt]bob-check.west);
        \draw[message arrow]
            (fork-b) to[out=40,in=190]
            node[pos=0.42, below=3pt] {$c_B$}
            ([yshift=-7pt]alice-check.west);
        \fill[fill=black] (fork-a) circle[radius=0.9pt];
        \fill[fill=black] (fork-b) circle[radius=0.9pt];
        % \node[font=\scriptsize, anchor=north]
        %     at ($(fork-b)!0.5!(bob-check.west)+(0,-0.78)$)
        %     {Dashed arrows: classical messages};
    }
]
\lstick{$A_1$} & \swap{1} & \qw & \qw & \qw \rstick{$A'$}\\
\lstick{$A_2$} & \targX{} & \qw & \qw & \qwbundle{} \\
\qw & \ctrl{-2} & \gate{H} & \meter{} & \wireoverride{n} \\[0.30cm]
\lstick{$B_1$} & \swap{1} & \qw & \qw & \qw \rstick{$B'$}\\
\lstick{$B_2$} & \targX{} & \qw & \qw & \qwbundle{} \\
\qw & \ctrl{-2} & \gate{H} & \meter{} & \wireoverride{n}
\end{quantikz}
\caption{Circuit implementation of Algorithm~\ref{alg:one_epr_purification}.
The input consists of two noisy bipartite copies and a preshared resource
\(\ket{\Phi_2}_{A_{\mathrm r}B_{\mathrm r}}
=(\ket{00}+\ket{11})/\sqrt2\).
The resource qubits control Alice's and Bob's local swap gates.
After a Hadamard gate and a computational-basis measurement on each resource
qubit, Alice sends \(c_A\) to Bob and Bob sends \(c_B\) to Alice. 
%Dashed arrows indicate transmitted classical bits. 
Each party compares the received bit with its own outcome.
The protocol accepts if \(c_A=c_B\) (\(00\) or \(11\)) and returns failure
otherwise. On acceptance, \(A_1B_1\) is retained as \(A'B'\), and
\(A_2B_2\) is discarded.}
\label{fig:full_epr_locc}
\end{figure}
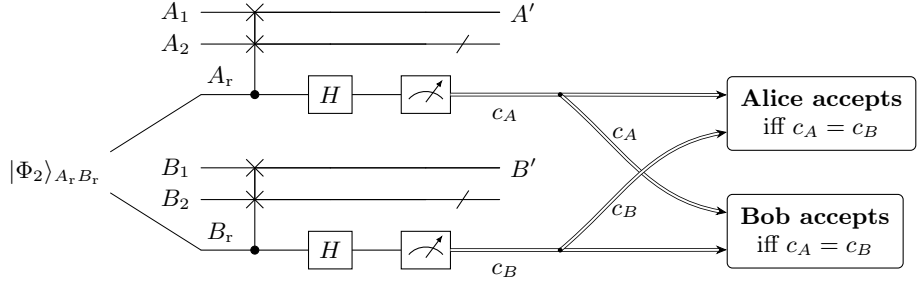

The circuit in Fig.~\ref{fig:full_epr_locc}
can be understood as follows. Alice's and Bob's local swaps together exchange the two complete
copies \(A_1B_1\) and \(A_2B_2\).
Using the shared EPR pair to coherently control these swaps,
Algorithm~\ref{alg:one_epr_purification} realizes a distributed
version of the swap test~\cite{buhrman2001quantum}, which compares
two quantum states using an ancillary qubit and a controlled-SWAP gate.
Accepting matching local measurement outcomes selects the
symmetric branch of this test, and discarding one copy then
produces the purified output.
The protocol thus implements symmetric projection followed by
discarding one copy using LOCC and one EPR pair.

The success probability and conditional fidelity attained by
Algorithm~\ref{alg:one_epr_purification} constitute the two-copy golden point
of the fidelity--success-probability tradeoff characterized by Yao et
al.~\cite{yao2025protocols}: the conditional fidelity is maximal, and the
success probability is maximal among protocols attaining that fidelity.  For
two copies, operating at this golden point also maximizes the success-weighted
gain considered here, as established independently by the global SDP in
Appendix~\ref{app:high_rank_entropy_lower_bound}.

The one-EPR construction raises the converse question of whether the global
optimum can be attained using less than one ebit.  The next theorem shows that
it cannot, even with an arbitrary finite-dimensional preshared resource state.

\begin{theorem}[One-ebit input-resource requirement]
\label{thm:main_one_ebit_necessity}
Fix \(d\ge2\) and \(0<\gamma<1\), and consider the \(2\to1\) purification task
on \(\mathcal S_P\) under the depolarizing channel \(\mathcal N^\gamma\).  Let
\(\omega_{A_{\mathrm r}B_{\mathrm r}}\) be any finite-dimensional preshared
resource state.  If \(\omega_{A_{\mathrm r}B_{\mathrm r}}\) enables a finite-round
LOCC protocol to attain \(G_\star(D,\gamma)\), then
\(E_{\mathrm F}(\omega_{A_{\mathrm r}B_{\mathrm r}})\ge1\).

Conversely, one EPR pair suffices by
Algorithm~\ref{alg:one_epr_purification}.  Hence the minimum input-resource
entanglement required for finite-round LOCC attainment is exactly one ebit.
\end{theorem}

The lower bound is proved for positive-partial-transpose (PPT) instruments, a
relaxation in which the Choi operators associated with both the success and
failure outcomes are required to remain positive under partial transpose.
Every finite-round LOCC instrument belongs to this larger class, so a lower
bound proved for PPT instruments also applies to LOCC.
Appendix~\ref{app:high_rank_entropy_lower_bound} proves that the global
benchmark is attained by a unique CPTN map.  Consequently, once the resource
state is supplied, the effective CPTN map implemented by any attaining
protocol must be this unique optimizer.  For a pure resource, the PPT
constraints then impose a condition on its Schmidt weights that yields
\(E(\eta)\ge1\).

For a mixed resource, linearity of the gain in the resource state and the
convex-roof definition of entanglement of formation extend the same bound to
\(E_{\mathrm F}(\omega)\ge1\).  The details are given in
Appendices~\ref{app:high_rank_entropy_lower_bound},
\ref{app:abstract_effect_schmidt}, and \ref{app:purification_application}. We further show that, among pure resource states with exactly one ebit, attaining the global optimum fixes the Schmidt spectrum.

\begin{theorem}[Uniqueness of the pure one-ebit resource]
\label{thm:main_pure_resource_equality}
Fix \(d\ge2\) and \(0<\gamma<1\), and consider the \(2\to1\) purification task
on \(\mathcal S_P\) under the depolarizing channel \(\mathcal N^\gamma\).  Let
\(\ket\eta\) be a finite-Schmidt-rank pure resource state with
\(E(\eta)=1\).  Then
\(\ket\eta\) enables finite-round LOCC attainment of \(G_\star(D,\gamma)\) if
and only if its Schmidt spectrum is \((1/2,1/2,0,\ldots)\).  Equivalently, the
pure one-ebit resources that attain the benchmark are precisely those locally
unitarily equivalent to an EPR pair.
\end{theorem}

This is to say, for LOCC purification protocols with a pure resource containing exactly one ebit of entanglement, we single out, up to local unitaries, the EPR pair as the sufficient and necessary resource to attain the global optimum. Necessity follows from Corollary~\ref{cor:pure_entropy_obstruction};
sufficiency follows by converting the resource to an EPR pair via local unitaries and applying the one-EPR protocol above. For two-qubit resources, Theorem~\ref{thm:main_one_ebit_necessity} yields a
complete characterization covering both pure and mixed states, as stated
in the following corollary.

% \begin{proofsketch}
% The equality conditions in the resource lower bound force the Schmidt spectrum
% to be \((1/2,1/2,0,\ldots)\).  Conversely, any pure state with this spectrum
% can be transformed by unitaries acting separately on the two resource registers
% into the standard resource used in
% Algorithm~\ref{alg:one_epr_purification}.
% \end{proofsketch}

\begin{corollary}[Uniqueness of the two-qubit resource]
\label{cor:two_qubit_resource_uniqueness}
Fix \(d\ge2\) and \(0<\gamma<1\), and consider the \(2\to1\) purification task
on \(\mathcal S_P\) under the depolarizing channel \(\mathcal N^\gamma\).
Let \(\omega_{A_{\mathrm r}B_{\mathrm r}}\) be an arbitrary two-qubit resource
state.  Then \(\omega\) enables finite-round
LOCC attainment of \(G_\star(D,\gamma)\) if and only if it is a pure
maximally entangled state, equivalently, an EPR pair up to local unitaries.
\end{corollary}

The proof uses the one-ebit lower bound and the fact that every mixed
two-qubit state has entanglement of formation strictly below one; see
Corollary~\ref{cor:two_qubit_resource_characterization} for the full proof.

\section{Discussion}

We have determined the preshared entanglement that finite-round LOCC requires to match the optimal global CPTN performance in the \(2\to1\) blind purification of arbitrary bipartite pure states under depolarizing noise: one ebit. One EPR pair attains the benchmark (Theorem~\ref{thm:main_global_cptn}); no finite-dimensional resource with less than one ebit of entanglement of formation attains it (Theorem~\ref{thm:main_one_ebit_necessity}); and among pure one-ebit resources, only those locally unitarily equivalent to an EPR pair attain it (Theorem~\ref{thm:main_pure_resource_equality}). For two-qubit resources, the purity and one-ebit assumptions can be dropped:
the benchmark is attainable if and only if the resource is locally unitarily
equivalent to an EPR pair (Corollary~\ref{cor:two_qubit_resource_uniqueness}). Notably, the one-ebit
threshold is uniform in the local dimension and the noise strength: it holds for every \(d\ge2\) and every \(0<\gamma<1\), although the benchmark \(G_\star(D,\gamma)\) itself depends on both. Although one ebit is sufficient, Theorem~\ref{thm:main_pure_resource_equality} adds that the form
of the resource matters as well, and Corollary~\ref{cor:entropy_not_sufficient} exhibits pure states with more than one ebit that still do not attain the optimum. This shows that the
amount of entanglement alone does not capture a resource's operational usefulness.

% One EPR pair enables Alice and Bob to implement the symmetric projection on the
% two complete copies by LOCC and thereby attain the global CPTN benchmark.  The
% minimum resource entanglement required is one ebit, independently of \(d\).
% Although the one-ebit lower bound is tight, more than one ebit of entanglement
% entropy does not by itself guarantee attainment of the global benchmark:
% Corollary~\ref{cor:entropy_not_sufficient} gives an explicit pure-state
% counterexample.  

Several extensions remain open. First, below the one-ebit threshold, what a weaker resource can still deliver remains unknown: less than one ebit cannot attain the global optimum but may still enable
intermediate gains, and the attainable gain as a function of the preshared entanglement is unknown.  This entanglement--gain tradeoff would quantify what sub-threshold resources provide where high-quality entanglement is scarce.

% Resources below one ebit cannot attain the global optimum, but
% may still yield intermediate gains; the optimal performance below one ebit
% remains open.

Second, for \(n>2\)
copies, projection onto the fully symmetric subspace gives a natural purification protocol for depolarizing noise
\cite{barenco1997stabilization,yao2025protocols}.
Optimal qubit purification has been characterized for arbitrary copy number \cite{cirac1999optimal}, whereas the conditional-fidelity optimality of symmetric projection for general \(n\) and \(D\) remains conjectural \cite{yao2025protocols}. %For general \(n>2\) and \(D=d^2\) is of direct interest for the many-copy regime relevant to purification rates \cite{keyl2001rate}.
Moreover, maximizing conditional fidelity alone
would not establish optimality for the success-probability-weighted fidelity gain considered here. 
%Maximizing conditional fidelity does not, by itself, establish optimality for the success-probability-weighted fidelity gain considered here.
A distributed implementation exists by generalizing
the coherent control of Algorithm~\ref{alg:one_epr_purification}: since each permutation of the complete bipartite copies factors into local permutations, coherently correlating the \(n!\) permutation labels in a shared maximally entangled control register provides a LOCC implementation of this particular projection. 
This construction uses \(\log_2(n!)\) ebits, and whether less entanglement suffices is unknown. 
Determining the global optimum and its minimal entanglement cost
for general \(n>2\) and \(D=d^2\) is of direct interest for the many-copy regime relevant to purification rates \cite{keyl2001rate}.

% Since each permutation of the complete bipartite copies factors into local
% permutations, coherently correlating the \(n!\) local permutation labels
% in 
% Determining the corresponding optimal global benchmark for general
% \(n>2\) and \(D=d^2\), together with the minimum preshared entanglement
% required to attain it, remains an open problem.

Third, the threshold is tied to the noise model. Depolarizing noise is an isotropic benchmark, while dephasing instead selects a preferred basis and makes the baseline fidelity target-dependent.  One EPR pair
still implements the distributed swap test, but that protocol may not remain
globally optimal---it is close to optimal for Pauli noise but
substantially suboptimal for amplitude damping \cite{yao2025protocols}. 
% , and local dephasing may require less assistance.  Existing
% optimizations likewise find the swap test close to optimal for Pauli noise but
% substantially suboptimal for amplitude damping \cite{yao2025protocols}.
% Consequently, the minimum resource cost generally depends on the noise model.

Our results illustrate how low-dimensional entangled control systems can enable collective noise reduction in distributed systems with arbitrarily large local state dimensions. As a benchmark, the one-ebit threshold helps distinguish fundamental entanglement requirements from the additional entanglement used by particular implementations. Such benchmarks could guide the allocation of shared entanglement to nonlocal operations and noise management in quantum networks and modular quantum computers.

\noindent\textbf{Formal verification and code availability.}
In the setting of Theorem~\ref{thm:main_one_ebit_necessity},
we have formally verified in Lean~4~\cite{moura2021lean} that any
finite-dimensional pure resource \(\ket\eta\) enabling a PPT
instrument to attain \(G_\star(D,\gamma)\) satisfies \(E(\eta)\ge1\).
This verification follows the same core argument as the appendix proof.
We report the verification for pure resources only, since the
mixed-resource bound follows by a short convex-roof argument.
The formalization uses Mathlib~\cite{mathlib2020} and
Lean-QIT~\cite{zhu2026lean}.
Source code is available at
\url{https://github.com/jyzhau/State-Purification}.

\textbf{Acknowledgment.} This work was supported by the National Natural Science Foundation of China (Grant. No.~92576114, 12447107), the Guangdong Provincial Quantum Science Strategic Initiative (Grant Nos.~GDZX2403008, GDZX2503001), the Guangdong Provincial Key Lab of Integrated Communication, Sensing and Computation for Ubiquitous Internet of Things (Grant No.~2023B1212010007), and the Guangdong Basic and Applied Basic Research Foundation (Grant Nos.~2026A1515030035 and 2025A1515110223).

\bibliography{ref}

\clearpage
\appendix

\section{Direct Evaluation of the One-EPR Algorithm}
\label{app:one_epr_protocol_calculation}

This appendix evaluates Algorithm~\ref{alg:one_epr_purification} directly from
its two local controlled-swap gates and its two accepted measurement outcomes.
Let \(X\) be an arbitrary density operator on
\(A_1B_1A_2B_2\); in the purification task it will be the mixed state
\(\mathcal N^\gamma(\psi)^{\otimes2}\).  Before any gate is applied, the EPR
resource and the two noisy copies are in the joint state
\[
\ket{\Phi_2}\!\bra{\Phi_2}_{A_{\mathrm r}B_{\mathrm r}}\otimes X
=
\frac{1}{2}
\Bigl(
\ket{00}+\ket{11}
\Bigr)
\Bigl(
\bra{00}+\bra{11}
\Bigr)_{A_{\mathrm r}B_{\mathrm r}}
\otimes X.
\]
Alice and Bob then apply, respectively, the local controlled-swap gates
\[
\begin{aligned}
\mathsf C_A
&=
\ket0\!\bra0_{A_{\mathrm r}}\otimes I_{A_1A_2}
+\ket1\!\bra1_{A_{\mathrm r}}\otimes\mathsf S_{A_1A_2},\\
\mathsf C_B
&=
\ket0\!\bra0_{B_{\mathrm r}}\otimes I_{B_1B_2}
+\ket1\!\bra1_{B_{\mathrm r}}\otimes\mathsf S_{B_1B_2}.
\end{aligned}
\]
The two gates act on disjoint registers and hence commute.  The \(\ket{00}\)
component of the EPR state activates neither swap, whereas its \(\ket{11}\)
component activates both.  Therefore, after the controlled swaps, the joint
state is
\[
\frac12
\Bigl(
\ket{00}\otimes I_{A_1B_1A_2B_2}
+
\ket{11}\otimes
\mathsf{S}_{A_1A_2}\otimes\mathsf{S}_{B_1B_2}
\Bigr)
X
\Bigl(
\bra{00}\otimes I_{A_1B_1A_2B_2}
+
\bra{11}\otimes
\mathsf{S}_{A_1A_2}\otimes\mathsf{S}_{B_1B_2}
\Bigr).
\]

The two Hadamard gates act on the resource qubits as
\[
\begin{aligned}
(H\otimes H)\ket{00}
&=
\frac12
\bigl(\ket{00}+\ket{01}+\ket{10}+\ket{11}\bigr),\\
(H\otimes H)\ket{11}
&=
\frac12
\bigl(\ket{00}-\ket{01}-\ket{10}+\ket{11}\bigr).
\end{aligned}
\]
Applying these identities on both sides of the controlled-swap state gives the
following joint density operator after the Hadamard gates:
\[
\begin{aligned}
\Omega_H
=\frac18
\biggl[
&\bigl(\ket{00}+\ket{11}\bigr)_{A_{\mathrm r}B_{\mathrm r}}
\otimes
\bigl(I+\mathsf S_{A_1A_2}\otimes\mathsf S_{B_1B_2}\bigr)\\
&+
\bigl(\ket{01}+\ket{10}\bigr)_{A_{\mathrm r}B_{\mathrm r}}
\otimes
\bigl(I-\mathsf S_{A_1A_2}\otimes\mathsf S_{B_1B_2}\bigr)
\biggr]X\\
\times\biggl[
&\bigl(\bra{00}+\bra{11}\bigr)_{A_{\mathrm r}B_{\mathrm r}}
\otimes
\bigl(I+\mathsf S_{A_1A_2}\otimes\mathsf S_{B_1B_2}\bigr)\\
&+
\bigl(\bra{01}+\bra{10}\bigr)_{A_{\mathrm r}B_{\mathrm r}}
\otimes
\bigl(I-\mathsf S_{A_1A_2}\otimes\mathsf S_{B_1B_2}\bigr)
\biggr].
\end{aligned}
\]
Alice and Bob now measure \(A_{\mathrm r}\) and \(B_{\mathrm r}\) in the
computational basis.  The joint projector associated with outcome
\((c_A,c_B)\) is
\[
\Pi_{c_Ac_B}
:=
\ket{c_Ac_B}\!\bra{c_Ac_B}_{A_{\mathrm r}B_{\mathrm r}}
\otimes I_{A_1B_1A_2B_2}.
\]
Reading the \(00\) and \(11\) coefficients directly from \(\Omega_H\) and
projecting on both sides gives
\[
\begin{aligned}
\Pi_{00}\Omega_H\Pi_{00}
&=
\frac18\ket{00}\!\bra{00}_{A_{\mathrm r}B_{\mathrm r}}\otimes
\bigl(I+\mathsf S_{A_1A_2}\otimes\mathsf S_{B_1B_2}\bigr)X
\bigl(I+\mathsf S_{A_1A_2}\otimes\mathsf S_{B_1B_2}\bigr),\\
\Pi_{11}\Omega_H\Pi_{11}
&=
\frac18\ket{11}\!\bra{11}_{A_{\mathrm r}B_{\mathrm r}}\otimes
\bigl(I+\mathsf S_{A_1A_2}\otimes\mathsf S_{B_1B_2}\bigr)X
\bigl(I+\mathsf S_{A_1A_2}\otimes\mathsf S_{B_1B_2}\bigr).
\end{aligned}
\]
The two outcomes are measured and recorded separately.  Hence the unnormalized
branch associated with the event \(c_A=c_B\) is the classical coarse-graining
of these two branches:
\[
\begin{aligned}
&\Pi_{00}\Omega_H\Pi_{00}
+
\Pi_{11}\Omega_H\Pi_{11}\\
&\quad=
\frac18
\bigl(\ket{00}\!\bra{00}+\ket{11}\!\bra{11}\bigr)_{A_{\mathrm r}B_{\mathrm r}}
\otimes
\bigl(I+\mathsf S_{A_1A_2}\otimes\mathsf S_{B_1B_2}\bigr)X
\bigl(I+\mathsf S_{A_1A_2}\otimes\mathsf S_{B_1B_2}\bigr).
\end{aligned}
\]
After discarding the measured resource qubits, the unnormalized accepted state
of the two noisy copies is therefore
\[
\begin{aligned}
&\operatorname{Tr}_{A_{\mathrm r}B_{\mathrm r}}\!\left[
\Pi_{00}\Omega_H\Pi_{00}
+
\Pi_{11}\Omega_H\Pi_{11}
\right]\\
&\quad=
\frac14
\bigl(I+\mathsf S_{A_1A_2}\otimes\mathsf S_{B_1B_2}\bigr)X
\bigl(I+\mathsf S_{A_1A_2}\otimes\mathsf S_{B_1B_2}\bigr).
\end{aligned}
\]
Define
\[
P_{\rm sym}
:=
\frac{
I_{A_1B_1A_2B_2}
+
\mathsf{S}_{A_1A_2}\otimes\mathsf{S}_{B_1B_2}
}{2},
\]
which is a projector because
\((\mathsf S_{A_1A_2}\otimes\mathsf S_{B_1B_2})^2=I\).  The preceding
unnormalized state is exactly
\[
\begin{aligned}
&\frac14
\bigl(I+\mathsf S_{A_1A_2}\otimes\mathsf S_{B_1B_2}\bigr)X
\bigl(I+\mathsf S_{A_1A_2}\otimes\mathsf S_{B_1B_2}\bigr)\\
&\qquad=
P_{\rm sym}XP_{\rm sym}.
\end{aligned}
\]
The outcomes \(01\) and \(10\) are rejected and therefore do not contribute to
this accepted state.
After discarding the second copy, the unnormalized accepted output on
\(A_1B_1\) is therefore
\[
\operatorname{Tr}_{A_2B_2}
\left[P_{\rm sym}XP_{\rm sym}\right].
\]
Since \(\mathsf{S}_{A_1A_2}\otimes\mathsf{S}_{B_1B_2}\) is the swap of the two
complete bipartite copies, \(P_{\rm sym}\) projects onto their symmetric
subspace.  Its implementation above nevertheless uses only Alice's local swap,
Bob's local swap, and comparison of their classical outcomes.
In the channel and Choi notation used in the following appendix, this retained
output is denoted by \(A'B'\).

Now let \(\psi=\ket{\psi}\!\bra{\psi}\) be the unknown target state.  For the
two-copy input \(X=\mathcal N^\gamma(\psi)^{\otimes2}\), the two tensor factors
are identical, so \(\mathsf{S}_{A_1A_2}\otimes \mathsf{S}_{B_1B_2}\) commutes with \(X\).
Hence
\[
P_{\rm sym}
\mathcal N^\gamma(\psi)^{\otimes2}
P_{\rm sym}
=
\frac{
\mathcal N^\gamma(\psi)^{\otimes2}
+
\left(\mathsf{S}_{A_1A_2}\otimes \mathsf{S}_{B_1B_2}\right)
\mathcal N^\gamma(\psi)^{\otimes2}
}{2}.
\]
The two required partial traces are
\[
\operatorname{Tr}_{A_2B_2}
\left[\mathcal N^\gamma(\psi)^{\otimes2}\right]
=
\mathcal N^\gamma(\psi),
\qquad
\operatorname{Tr}_{A_2B_2}
\left[
\left(\mathsf{S}_{A_1A_2}\otimes \mathsf{S}_{B_1B_2}\right)
\mathcal N^\gamma(\psi)^{\otimes2}
\right]
=
\bigl(\mathcal N^\gamma(\psi)\bigr)^2.
\]
To see the second identity, choose local bases
\(\{\ket{a}\}_{a=1}^d\) and \(\{\ket{b}\}_{b=1}^d\), and write
\(\ket{a,b}:=\ket{a}\otimes\ket{b}\).  Expanding the two local swaps and
regrouping the labelled tensor factors according to
\(A_1B_1:A_2B_2\) gives
\[
\begin{aligned}
\mathsf S_{A_1A_2}\otimes\mathsf S_{B_1B_2}
&=
\sum_{a,a',b,b'=1}^d
\bigl(\ket{a}\!\bra{a'}\bigr)_{A_1}\otimes
\bigl(\ket{a'}\!\bra{a}\bigr)_{A_2}
\otimes
\bigl(\ket{b}\!\bra{b'}\bigr)_{B_1}\otimes
\bigl(\ket{b'}\!\bra{b}\bigr)_{B_2}\\
&=
\sum_{a,a',b,b'=1}^d
\bigl(\ket{a,b}\!\bra{a',b'}\bigr)_{A_1B_1}
\otimes
\bigl(\ket{a',b'}\!\bra{a,b}\bigr)_{A_2B_2}\\
&=
\sum_{i,j=1}^D
\bigl(\ket{i}\!\bra{j}\bigr)_{A_1B_1}\otimes
\bigl(\ket{j}\!\bra{i}\bigr)_{A_2B_2}.
\end{aligned}
\]
Here the last line uses the composite labels \(i=(a,b)\),
\(j=(a',b')\), and \(D=d^2\).  Thus the product of the two local swaps is
exactly the swap of the two complete bipartite copies.
Consequently, for operators \(Y\) and \(Z\) on one copy,
\[
\begin{aligned}
\operatorname{Tr}_{A_2B_2}\!\left[
\left(\mathsf S_{A_1A_2}\otimes\mathsf S_{B_1B_2}\right)
\left(Y_{A_1B_1}\otimes Z_{A_2B_2}\right)
\right]
&=
\sum_{i,j=1}^D
\ket{i}\!\bra{j}\,Y\,
\operatorname{Tr}\!\left[\ket{j}\!\bra{i}\,Z\right]\\
&=
\sum_{i,j=1}^D
\langle i|Z|j\rangle\ket{i}\!\bra{j}\,Y
=
ZY.
\end{aligned}
\]
Taking \(Y=Z=\mathcal N^\gamma(\psi)\) gives the stated square.
The unnormalized accepted output is
\begin{equation}\label{eq:one_epr_accepted_output}
\widehat{\sigma}_\psi
=
p_\psi\sigma_\psi
=
\frac{
\mathcal N^\gamma(\psi)
+\bigl(\mathcal N^\gamma(\psi)\bigr)^2
}{2}.
\end{equation}

The input fidelity and the purity of the noisy state are
\[
\mathrm F\!\left(\psi,\mathcal N^\gamma(\psi)\right)
=
\left\langle\psi\left|\mathcal N^\gamma(\psi)\right|\psi\right\rangle
=
1-\gamma+\frac{\gamma}{D},
\]
and
\[
\operatorname{Tr}\!\left[\bigl(\mathcal N^\gamma(\psi)\bigr)^2\right]
=
(1-\gamma)^2
+\frac{2\gamma(1-\gamma)}{D}
+\frac{\gamma^2}{D},
\]
respectively.  Taking the trace of
Eq.~\eqref{eq:one_epr_accepted_output}, and using
\(\operatorname{Tr}[\mathcal N^\gamma(\psi)]=1\), gives the accepted
probability:
\[
\begin{aligned}
p_\psi
&=
\operatorname{Tr}[\widehat{\sigma}_\psi]\\
&=
\frac{
1+\operatorname{Tr}\!\left[\bigl(\mathcal N^\gamma(\psi)\bigr)^2\right]
}{2}.
\end{aligned}
\]
To evaluate the target overlap, now use that \(\ket{\psi}\) is an eigenvector of
\(\mathcal N^\gamma(\psi)\) with eigenvalue
\(\mathrm F\!\left(\psi,\mathcal N^\gamma(\psi)\right)\).  It follows that
\[
\begin{aligned}
\langle\psi|\widehat{\sigma}_\psi|\psi\rangle
&=
\frac{
\langle\psi|\mathcal N^\gamma(\psi)|\psi\rangle
+
\langle\psi|\bigl(\mathcal N^\gamma(\psi)\bigr)^2|\psi\rangle
}{2}\\
&=
\frac{
\mathrm F\!\left(\psi,\mathcal N^\gamma(\psi)\right)
+\mathrm F\!\left(\psi,\mathcal N^\gamma(\psi)\right)^2
}{2}\\
&=
p_\psi\mathrm F(\psi,\sigma_\psi).
\end{aligned}
\]
The fidelity of the normalized accepted output is therefore
\begin{equation}\label{eq:one_epr_conditional_fidelity}
\begin{aligned}
\mathrm F(\psi,\sigma_\psi)
&=
\frac{
\mathrm F\!\left(\psi,\mathcal N^\gamma(\psi)\right)
+\mathrm F\!\left(\psi,\mathcal N^\gamma(\psi)\right)^2
}{
1+\operatorname{Tr}\!\left[\bigl(\mathcal N^\gamma(\psi)\bigr)^2\right]
}\\
&=
\frac{
\left(1-\gamma+\frac{\gamma}{D}\right)
\left(2-\gamma+\frac{\gamma}{D}\right)
}{
1+(1-\gamma)^2
+\frac{2\gamma(1-\gamma)}{D}
+\frac{\gamma^2}{D}
}.
\end{aligned}
\end{equation}

The success-probability-weighted fidelity improvement appearing in the
postselected benchmark is
\[
p_\psi
\left[
\mathrm F(\psi,\sigma_\psi)
-\mathrm F\!\left(\psi,\mathcal N^\gamma(\psi)\right)
\right]
=
\frac{\mathrm F\!\left(\psi,\mathcal N^\gamma(\psi)\right)}{2}
\left(
\mathrm F\!\left(\psi,\mathcal N^\gamma(\psi)\right)
-\operatorname{Tr}\!\left[\bigl(\mathcal N^\gamma(\psi)\bigr)^2\right]
\right).
\]
Here
\[
\mathrm F\!\left(\psi,\mathcal N^\gamma(\psi)\right)
-\operatorname{Tr}\!\left[\bigl(\mathcal N^\gamma(\psi)\bigr)^2\right]
=
\gamma(1-\gamma)
\left(1-\frac1D\right),
\]
and hence
\begin{equation}\label{eq:one_epr_direct_gain}
p_\psi
\left[
\mathrm F(\psi,\sigma_\psi)
-\mathrm F\!\left(\psi,\mathcal N^\gamma(\psi)\right)
\right]
=
\frac12
\left(1-\gamma+\frac{\gamma}{D}\right)
\gamma(1-\gamma)
\left(1-\frac1D\right).
\end{equation}
The right-hand side is strictly positive for \(d\ge2\) and
\(0<\gamma<1\), so the conditional fidelity is strictly improved.

For comparison with the global SDP below, we finally record the Choi operator
of the accepted map.  In the three-leg ket \(|a,b,c\rangle\) below, the first
two entries belong to the Choi-input coordinate copies
\(\mathbb C^D_{A_1B_1}\) and \(\mathbb C^D_{A_2B_2}\), while the third belongs
to the output \(\mathbb C^D_{A'B'}\).  Define the symmetric contraction
isometry
\begin{equation}\label{eq:gd_Wplus_definition}
W_+|a\rangle
=
\frac{1}{\sqrt{2(D+1)}}\sum_i
\left(|a,i,i\rangle+|i,a,i\rangle\right),
\qquad
P_{W_+}:=W_+W_+^\dagger .
\end{equation}
With \(A_1B_1\) relabeled as the output \(A'B'\), a Kraus family for
\(X\mapsto\operatorname{Tr}_{A_2B_2}[P_{\rm sym}XP_{\rm sym}]\) is
\[
K_i
:=
\left(I_{A_1B_1}\otimes\langle i|_{A_2B_2}\right)P_{\rm sym},
\qquad i=1,\ldots,D.
\]
Using input-first vectorization and the three-leg convention just stated,
\[
\begin{aligned}
|K_i\rangle\!\rangle
&:=
\sum_{a,b}|a,b\rangle\otimes K_i|a,b\rangle\\
&=
\frac12\sum_a
\left(|a,i,a\rangle+|i,a,a\rangle\right)\\
&=
\sqrt{\frac{D+1}{2}}\,W_+|i\rangle.
\end{aligned}
\]
The third entry in each three-leg ket is on \(A'B'\).  Hence the Choi operator
is
\begin{equation}\label{eq:one_epr_accepted_choi}
\sum_{i=1}^D|K_i\rangle\!\rangle\langle\!\langle K_i|
=
\frac{D+1}{2}P_{W_+}.
\end{equation}

\section{Global Optimizer and PPT Resource Constraints}
\label{app:high_rank_entropy_lower_bound}

The two noisy input copies are \(A_1B_1\) and \(A_2B_2\), the output copy is
\(A'B'\), and the preshared resource state is held in
\(A_{\mathrm r}B_{\mathrm r}\).  Write \(R:=A_{\mathrm r}B_{\mathrm r}\) for
the joint resource register.  Each local register has dimension \(d\), so each
bipartite copy has dimension \(D=d^2\).  We first determine the unrestricted
global CPTN benchmark and then formulate the PPT relaxation of LOCC instruments
with a preshared resource state.  Because every LOCC instrument belongs to the
larger PPT class, any necessary resource bound established for PPT instruments
also applies to LOCC protocols.

\subsection{Optimal global CPTN benchmark}

For a map \(\mathcal E:\mathcal B(X)\to\mathcal B(Y)\), fix an orthonormal
basis \(\{|i\rangle_X\}\) and use the input-first Choi convention
\[
J_{\mathcal E}
:=
\sum_{i,j}
|i\rangle\langle j|_X
\otimes
\mathcal E\!\left(|i\rangle\langle j|_X\right),
\]
where the first factor is the Choi-input coordinate copy of \(X\).  Then
\[
\mathcal E(Z)
=
\operatorname{Tr}_X\!\left[
\left(Z^T_X\otimes I_Y\right)J_{\mathcal E}
\right].
\]
Thus the CPTN constraints are \(J_{\mathcal E}\succeq0\) and
\(\operatorname{Tr}_YJ_{\mathcal E}\preceq I_X\).  All transposes and partial
transposes below use the tensor-product bases induced by these choices.

Define the global three-leg Choi space
\[
\mathcal H
:=
\mathbb C^D_{A_1B_1}
\otimes
\mathbb C^D_{A_2B_2}
\otimes
\mathbb C^D_{A'B'}.
\]
In a three-leg ket \(|a,b,c\rangle\), the first two entries are composite
\(AB\) input indices of dimension \(D\), and the third is the output index.
For \(|u\rangle=\sum_i u_i|i\rangle\), write
\(|u^*\rangle:=\sum_i\overline{u_i}|i\rangle\).
Define
\begin{equation}\label{eq:gd_noise_channel}
\mathcal N^\gamma(\rho)=(1-\gamma)\rho+\frac{\gamma}{D}I_D.
\end{equation}
For pure \(\psi\),
\(\mathrm F(\psi,\mathcal N^\gamma(\psi))=1-\gamma+\gamma/D\).  If
\(J\in\mathcal B(\mathcal H)\) is the Choi operator of a CPTN success map
\(\mathcal E\), then
\[
\begin{aligned}
g_\psi(\mathcal E)
&=
\langle\psi|\widehat{\sigma}_\psi|\psi\rangle
-p_\psi\,\mathrm F\!\left(\psi,\mathcal N^\gamma(\psi)\right)\\
&=
\operatorname{Tr}\!\left[
\left(
\psi-\mathrm F\!\left(\psi,\mathcal N^\gamma(\psi)\right)I_D
\right)
\mathcal E\!\left(\mathcal N^\gamma(\psi)^{\otimes2}\right)
\right].
\end{aligned}
\]
Averaging this identity over \(\psi\) and using the Choi trace identity gives
the objective \(\operatorname{Tr}[MJ]\).  Hence the optimal global CPTN
benchmark is given by the SDP
\begin{equation}\label{eq:gd_global_cptn_sdp}
\begin{aligned}
\max_J\quad & \operatorname{Tr}[MJ]\\
\text{\rm s.t.}\quad & J\succeq0,\\
& \operatorname{Tr}_{A'B'}J\preceq I_{A_1B_1A_2B_2},
\end{aligned}
\end{equation}
Using \(\int d\psi\) to denote integration with respect to \(\mu_P\), the
performance operator is
\begin{equation}\label{eq:gd_global_performance_operator}
M=
\int d\psi\,
\left(\mathcal N^\gamma(\psi)^{\otimes2}\right)^T
\otimes
\left[
\psi-\left(1-\gamma+\frac{\gamma}{D}\right)I_D
\right].
\end{equation}
\begin{lemma}[Unique optimizer of the global CPTN benchmark]
\label{lem:general_d_global_optimum}
For \(0<\gamma<1\), the SDP \eqref{eq:gd_global_cptn_sdp} has value
\[
G_\star(D,\gamma)
=
\frac12
\left(1-\gamma+\frac{\gamma}{D}\right)
\gamma(1-\gamma)
\left(1-\frac1D\right),
\]
and its unique optimizer is
\[
J_\star=\frac{D+1}{2}P_{W_+}.
\]
\end{lemma}

\begin{proof}
Under the input-first Choi convention, a physical transformation by \(U\) is
represented by \(\overline U\) on each Choi-input coordinate factor.  The
relevant global representation is therefore
\[
\pi_D(U)
:=
\overline U_{A_1B_1}\otimes\overline U_{A_2B_2}\otimes U_{A'B'},
\qquad U\in U(D).
\]
Here \(U\) is an unrestricted unitary on a complete \(AB\) copy and need not
factor across Alice and Bob.  Haar invariance gives
\[
[M,\pi_D(U)]=0,
\qquad U\in U(D).
\]

Mixed Schur--Weyl duality
\cite{benkart1994tensor,grinko2023gelfand,nguyen2023mixed} for
\(\mathbb C^D_{A'B'}\otimes\mathbb C^D_{A_1B_1}
\otimes\mathbb C^D_{A_2B_2}\), with the first factor carrying \(U\) and the
last two carrying \(\overline U\), gives
\[
\mathcal H
\simeq
\bigoplus_{\lambda\in\operatorname{Irr}(\mathcal A^D_{1,2})}
V_\lambda^{\mathcal A^D_{1,2}}\otimes V_\lambda^{U(D)},
\qquad
\pi_D(U)
\simeq
\bigoplus_{\lambda}
I_{d_\lambda}\otimes\phi_\lambda(U),
\]
where
\[
d_\lambda
:=
\dim V_\lambda^{\mathcal A^D_{1,2}}
=
|\operatorname{Paths}(\lambda)|
\]
is the multiplicity of the irreducible \(U(D)\)-module
\(V_\lambda^{U(D)}\).  For \((n,m)=(1,2)\), the path rule gives
\[
\begin{array}{rcl}
(\varnothing,\varnothing)
&\to&((1),\varnothing)\to(\varnothing,\varnothing)
       \to(\varnothing,(1)),\\
(\varnothing,\varnothing)
&\to&((1),\varnothing)\to((1),(1))
       \to(\varnothing,(1)),\\
(\varnothing,\varnothing)
&\to&((1),\varnothing)\to((1),(1))
       \to((1),(2)),\\
(\varnothing,\varnothing)
&\to&((1),\varnothing)\to((1),(1))
       \to((1),(1,1)).
\end{array}
\]
Set
\[
\lambda_0:=(\varnothing,(1)),
\qquad
\lambda_+:=((1),(2)),
\qquad
\lambda_-:=((1),(1,1)).
\]
Their path multiplicities are
\[
d_{\lambda_0}=2,
\qquad
d_{\lambda_+}=d_{\lambda_-}=1.
\]

For a leaf \(\lambda=(\lambda_l,\lambda_r)\), pad both diagrams by zeros and
associate the dominant weight
\[
\Lambda_i:=\lambda_{l,i}-\lambda_{r,D+1-i},
\qquad 1\le i\le D.
\]
The Weyl dimension formula
\cite[Sec.~2]{grinko2023efficient}
\[
\dim V_\Lambda
=
\prod_{1\le i<j\le D}
\frac{\Lambda_i-\Lambda_j+j-i}{j-i}
\]
gives
\[
\begin{array}{c|c|c|c}
\lambda&\Lambda&d_\lambda&\dim V_\lambda^{U(D)}\\ \hline
\lambda_0&(0^{D-1},-1)&2&D\\
\lambda_+&(1,0^{D-2},-2)&1&\dfrac{D(D-1)(D+2)}2\\
\lambda_-&(1,0^{D-3},-1,-1)&1&\dfrac{D(D+1)(D-2)}2.
\end{array}
\]
Here \(0^r\) denotes a string of \(r\) zeros.  The admissibility condition is
\(\ell(\lambda_l)+\ell(\lambda_r)\le D\); since \(D=d^2\ge4\), all three
leaves are admissible.  The dimension identity
\[
2D+\frac{D(D-1)(D+2)}2+\frac{D(D+1)(D-2)}2
=D^3=\dim\mathcal H
\]
checks that these abstract summands exhaust the three-leg space.

We next identify the multiplicity-two summand concretely.  Alongside the
\(W_+\) in \eqref{eq:gd_Wplus_definition}, define
\[
W_-|a\rangle
=
\frac{1}{\sqrt{2(D-1)}}\sum_i
\left(|a,i,i\rangle-|i,a,i\rangle\right),
\qquad
P_{W_-}:=W_-W_-^\dagger.
\]
For every \(U\in U(D)\), unitarity gives
\[
\sum_i \overline U|i\rangle\otimes U|i\rangle
=
\sum_{j,k}(\overline U U^T)_{jk}|j,k\rangle
=
\sum_j|j,j\rangle.
\]
Applying this identity to the two terms in each \(W_\pm\), and contracting
their explicit definitions, gives
\begin{equation}\label{eq:W_intertwining_identities}
W_\pm^\dagger W_\pm=I_D,
\qquad
W_+^\dagger W_-=0,
\qquad
\pi_D(U)W_\pm=W_\pm\overline U.
\end{equation}
The weight \((0^{D-1},-1)\) is the dual defining representation, so
\eqref{eq:W_intertwining_identities} exhibits two orthogonal concrete copies of
the \(\lambda_0\)-type.  Their total dimension \(2D\) matches the path
multiplicity and Weyl dimension above.  Hence
\[
\operatorname{Ran}W_+\oplus\operatorname{Ran}W_-
\]
is the full \(\lambda_0\)-isotypic component, on which the \(U(D)\)-commutant
is \(\mathcal B(\mathbb C^2)\otimes I_D\).

The two identical input copies provide an additional commuting symmetry.  On
the full Choi space, their exchange acts as
\[
\bigl(\mathsf{S}_{A_1A_2}\otimes\mathsf{S}_{B_1B_2}\otimes I_{A'B'}\bigr)
|x,y,z\rangle
=
|y,x,z\rangle.
\]
It commutes with every \(\pi_D(U)\).  Moreover, each integrand in
\eqref{eq:gd_global_performance_operator} is invariant under this exchange,
because its two input factors are identical.  Therefore
\[
\left[
M,
\mathsf{S}_{A_1A_2}\otimes\mathsf{S}_{B_1B_2}\otimes I_{A'B'}
\right]
=0.
\]
Conjugation by either this swap or \(\pi_D(U)\) preserves the SDP constraints,
so an optimizer may be chosen in their joint commutant.

Directly from the definitions of \(W_+\) and \(W_-\),
\begin{equation}\label{eq:W_input_swap_parity}
\bigl(
\mathsf{S}_{A_1A_2}\otimes\mathsf{S}_{B_1B_2}\otimes I_{A'B'}
\bigr)W_\pm
=\pm W_\pm.
\end{equation}
Indeed, the input-copy swap exchanges the two summands in each definition of
\(W_\pm\), preserving their sum and negating their difference.
Thus the swap selects a parity basis in the two-dimensional multiplicity
space.  Explicitly,
\[
|+\rangle\otimes|a\rangle\longmapsto W_+|a\rangle,
\qquad
|-\rangle\otimes|a\rangle\longmapsto W_-|a\rangle,
\]
under which
\[
\left.\pi_D(U)\right|_{\operatorname{Ran}W_+\oplus\operatorname{Ran}W_-}
\simeq I_2\otimes\overline U,
\qquad
\left.
\bigl(\mathsf{S}_{A_1A_2}\otimes\mathsf{S}_{B_1B_2}\otimes I_{A'B'}\bigr)
\right|_{\operatorname{Ran}W_+\oplus\operatorname{Ran}W_-}
\simeq
\begin{pmatrix}1&0\\0&-1\end{pmatrix}\otimes I_D.
\]
In particular, the swap symmetry removes the off-diagonal multiplicity block
of \(M\) between \(\operatorname{Ran}W_+\) and \(\operatorname{Ran}W_-\).

Recall the symmetric input projector from the direct protocol calculation and
define its antisymmetric counterpart:
\[
P_{\rm sym}
:=
\frac{
I_{A_1B_1A_2B_2}+\mathsf{S}_{A_1A_2}\otimes\mathsf{S}_{B_1B_2}
}{2},
\qquad
P_{\rm asym}
:=
\frac{
I_{A_1B_1A_2B_2}-\mathsf{S}_{A_1A_2}\otimes\mathsf{S}_{B_1B_2}
}{2}.
\]
On the full Choi space, these projectors are extended trivially on \(A'B'\).
Recall that \(P_{W_\pm}=W_\pm W_\pm^\dagger\) are the orthogonal projectors
onto \(\operatorname{Ran}W_\pm\).  Equation~\eqref{eq:W_input_swap_parity}
gives
\[
\left(P_{\rm sym}\otimes I_{A'B'}\right)W_+=W_+,
\qquad
\left(P_{\rm asym}\otimes I_{A'B'}\right)W_-=W_-.
\]
Hence
\[
\operatorname{Ran}W_+
\subseteq
\operatorname{Ran}\left(P_{\rm sym}\otimes I_{A'B'}\right),
\qquad
\operatorname{Ran}W_-
\subseteq
\operatorname{Ran}\left(P_{\rm asym}\otimes I_{A'B'}\right),
\]
or, equivalently,
\[
P_{W_+}\preceq P_{\rm sym}\otimes I_{A'B'},
\qquad
P_{W_-}\preceq P_{\rm asym}\otimes I_{A'B'}.
\]
Define the residual parity subspaces by
\[
\begin{aligned}
\mathcal L_+
&:=
\operatorname{Ran}\left(P_{\rm sym}\otimes I_{A'B'}\right)
\cap\left(\operatorname{Ran}W_+\right)^\perp,\\
\mathcal L_-
&:=
\operatorname{Ran}\left(P_{\rm asym}\otimes I_{A'B'}\right)
\cap\left(\operatorname{Ran}W_-\right)^\perp.
\end{aligned}
\]
Their orthogonal projectors are
\[
P_{\mathcal L_+}
:=
P_{\rm sym}\otimes I_{A'B'}-P_{W_+},
\qquad
P_{\mathcal L_-}
:=
P_{\rm asym}\otimes I_{A'B'}-P_{W_-}.
\]
The intertwining identities imply that \(P_{W_\pm}\) commute with
\(\pi_D(U)\); the input-parity projectors do as well.  Hence
\(\mathcal L_\pm\) are invariant under \(\pi_D(U)\), and the input-copy swap
acts on them with parity \(\pm1\).
Their dimensions are
\[
\begin{aligned}
\dim\mathcal L_+
&=\frac{D^2(D+1)}2-D
=\frac{D(D-1)(D+2)}2,\\
\dim\mathcal L_-
&=\frac{D^2(D-1)}2-D
=\frac{D(D+1)(D-2)}2.
\end{aligned}
\]
These dimensions match the Weyl dimensions of the two remaining
multiplicity-one types.  Their one-dimensional
\(\mathcal A^D_{1,2}\)-modules carry input-swap characters \(+1\) for
\(\lambda_+\) and \(-1\) for \(\lambda_-\).  Hence
\(\mathcal L_+\) and \(\mathcal L_-\) realize the \(\lambda_+\)- and
\(\lambda_-\)-types, respectively.  The resulting identification is
\[
\begin{array}{c|c|c|c}
\text{mixed-Schur label}&\text{input-swap parity}&
\text{concrete subspace}&\text{dimension}\\ \hline
(\varnothing,(1))&+1&\operatorname{Ran}W_+&D\\
(\varnothing,(1))&-1&\operatorname{Ran}W_-&D\\
((1),(2))&+1&\mathcal L_+&
\frac{D(D-1)(D+2)}2\\
((1),(1,1))&-1&\mathcal L_-&
\frac{D(D+1)(D-2)}2.
\end{array}
\]
Thus
\[
\mathcal H
=
\operatorname{Ran}W_+\oplus\operatorname{Ran}W_-
\oplus\mathcal L_+\oplus\mathcal L_-.
\]
In the physical ordering \(A_1B_1,A_2B_2,A'B'\), let
\(\mathsf{S}_{A_1A'}\otimes \mathsf{S}_{B_1B'}\) and
\(\mathsf{S}_{A_2A'}\otimes \mathsf{S}_{B_2B'}\) be the two algebraic
input--output permutations, acting trivially on the remaining input copy.  Thus
\[
\left(\mathsf{S}_{A_1A'}\otimes \mathsf{S}_{B_1B'}\right)
|x,y,z\rangle=|z,y,x\rangle,
\qquad
\left(\mathsf{S}_{A_2A'}\otimes \mathsf{S}_{B_2B'}\right)
|x,y,z\rangle=|x,z,y\rangle.
\]
The projector \(P_{\rm sym}^{(3)}\) onto the fully symmetric subspace of
\((\mathbb C^D)^{\otimes3}\) is then
\[
\begin{aligned}
6P_{\rm sym}^{(3)}
&={}I_{\mathcal H}
+\mathsf{S}_{A_1A_2}\otimes \mathsf{S}_{B_1B_2}\otimes I_{A'B'}
\\
&+\mathsf{S}_{A_1A'}\otimes \mathsf{S}_{B_1B'}
+\mathsf{S}_{A_2A'}\otimes \mathsf{S}_{B_2B'}\\
&+\left(\mathsf{S}_{A_2A'}\otimes \mathsf{S}_{B_2B'}\right)
 \left(\mathsf{S}_{A_1A'}\otimes \mathsf{S}_{B_1B'}\right)\\
&+\left(\mathsf{S}_{A_1A'}\otimes \mathsf{S}_{B_1B'}\right)
 \left(\mathsf{S}_{A_2A'}\otimes \mathsf{S}_{B_2B'}\right).
\end{aligned}
\]
For the objective calculation, use the Haar moments
\[
\int d\psi\,\psi=\frac{I_D}{D},
\qquad
\int d\psi\,\psi^{\otimes2}=\frac{P_{\rm sym}}{\binom{D+1}{2}},
\qquad
\int d\psi\,\psi^{\otimes3}=\frac{P_{\rm sym}^{(3)}}{\binom{D+2}{3}}.
\]
An operator subscript will record the physical composite Choi leg on which it
acts.  Thus the three copies of \(\psi\) are written
\(\psi_{A_1B_1}\), \(\psi_{A_2B_2}\), and \(\psi_{A'B'}\).  Since partial
transpose over the two input copies is linear,
\[
\int d\psi\,
\left(
\psi_{A_1B_1}\otimes\psi_{A_2B_2}\otimes\psi_{A'B'}
\right)^{T_{A_1B_1A_2B_2}}
=
\left(
\int d\psi\,
\psi_{A_1B_1}\otimes\psi_{A_2B_2}\otimes\psi_{A'B'}
\right)^{T_{A_1B_1A_2B_2}},
\]
and similarly for the two-factor moments.  Therefore, after expanding
\[
\mathcal N^\gamma(\psi)^{\otimes2}
=(1-\gamma)^2\psi^{\otimes2}
+\frac{\gamma(1-\gamma)}{D}(\psi\otimes I_D+I_D\otimes\psi)
+\frac{\gamma^2}{D^2}I_D\otimes I_D,
\]
the performance operator is computed as
\[
\begin{aligned}
M={}&(1-\gamma)^2
\left(
\int d\psi\,
\psi_{A_1B_1}\otimes\psi_{A_2B_2}\otimes\psi_{A'B'}
\right)^{T_{A_1B_1A_2B_2}}\\
&-\left(1-\gamma+\frac{\gamma}{D}\right)(1-\gamma)^2
\left(
\int d\psi\,
\psi_{A_1B_1}\otimes\psi_{A_2B_2}\otimes I_{A'B'}
\right)^{T_{A_1B_1A_2B_2}}\\
&+\frac{\gamma(1-\gamma)}{D}
\left(
\int d\psi\,
\psi_{A_1B_1}\otimes I_{A_2B_2}\otimes\psi_{A'B'}
\right)^{T_{A_1B_1A_2B_2}}\\
&+\frac{\gamma(1-\gamma)}{D}
\left(
\int d\psi\,
I_{A_1B_1}\otimes\psi_{A_2B_2}\otimes\psi_{A'B'}
\right)^{T_{A_1B_1A_2B_2}}\\
&-\frac{2\gamma(1-\gamma)}{D^2}
\left(1-\gamma+\frac{\gamma}{D}\right)I_{\mathcal H}\\
&-\frac{\gamma^2(1-\gamma)}{D^2}
\left(1-\frac1D\right)I_{\mathcal H}.
\end{aligned}
\]
Introduce the contraction maps
\[
\Gamma_1|a\rangle:=\sum_i|i,a,i\rangle,
\qquad
\Gamma_2|a\rangle:=\sum_i|a,i,i\rangle,
\]
and set
\[
C_{\rm diag}:=\Gamma_1\Gamma_1^\dagger+\Gamma_2\Gamma_2^\dagger,
\qquad
C_{\rm cross}:=\Gamma_1\Gamma_2^\dagger+\Gamma_2\Gamma_1^\dagger.
\]
The required matrix-unit expansions of these permutations after partial
transpose are
\[
\begin{aligned}
\left(\mathsf{S}_{A_1A'}\otimes \mathsf{S}_{B_1B'}\right)^{T_{A_1B_1A_2B_2}}
&=\Gamma_1\Gamma_1^\dagger\\
&=\sum_{i,j,a}|i,a,i\rangle\langle j,a,j|,\\[2pt]
\left(\mathsf{S}_{A_2A'}\otimes \mathsf{S}_{B_2B'}\right)^{T_{A_1B_1A_2B_2}}
&=\Gamma_2\Gamma_2^\dagger\\
&=\sum_{i,j,a}|a,i,i\rangle\langle a,j,j|,\\[2pt]
\left[
\left(\mathsf{S}_{A_2A'}\otimes \mathsf{S}_{B_2B'}\right)
\left(\mathsf{S}_{A_1A'}\otimes \mathsf{S}_{B_1B'}\right)
\right]^{T_{A_1B_1A_2B_2}}
&=\Gamma_2\Gamma_1^\dagger\\
&=\sum_{i,j,a}|a,i,i\rangle\langle j,a,j|,\\[2pt]
\left[
\left(\mathsf{S}_{A_1A'}\otimes \mathsf{S}_{B_1B'}\right)
\left(\mathsf{S}_{A_2A'}\otimes \mathsf{S}_{B_2B'}\right)
\right]^{T_{A_1B_1A_2B_2}}
&=\Gamma_1\Gamma_2^\dagger\\
&=\sum_{i,j,a}|i,a,i\rangle\langle a,j,j|.
\end{aligned}
\]
For example,
\[
\begin{aligned}
\mathsf{S}_{A_1A'}\otimes \mathsf{S}_{B_1B'}
&=\sum_{i,j,a}|j,a,i\rangle\langle i,a,j|,\\
\left(\mathsf{S}_{A_1A'}\otimes \mathsf{S}_{B_1B'}\right)^{T_{A_1B_1A_2B_2}}
&=\sum_{i,j,a}|i,a,i\rangle\langle j,a,j|;
\end{aligned}
\]
and similarly for the other three lines.  Also the identity and the full
input-copy swap are unchanged by \(T_{A_1B_1A_2B_2}\).
Consequently the three moment terms are
\[
\begin{aligned}
&\left(\int d\psi\,
\psi_{A_1B_1}\otimes\psi_{A_2B_2}\otimes\psi_{A'B'}
\right)^{T_{A_1B_1A_2B_2}}\\
&\qquad=
\frac{
I_{\mathcal H}
+\mathsf{S}_{A_1A_2}\otimes \mathsf{S}_{B_1B_2}\otimes I_{A'B'}
+C_{\rm diag}+C_{\rm cross}
}
{D(D+1)(D+2)},\\[3pt]
&\left(\int d\psi\,
\psi_{A_1B_1}\otimes\psi_{A_2B_2}\otimes I_{A'B'}
\right)^{T_{A_1B_1A_2B_2}}\\
&\qquad=
\frac{
I_{\mathcal H}
+\mathsf{S}_{A_1A_2}\otimes \mathsf{S}_{B_1B_2}\otimes I_{A'B'}
}{D(D+1)},\\[3pt]
&\left(\int d\psi\,
\psi_{A_1B_1}\otimes I_{A_2B_2}\otimes\psi_{A'B'}
\right)^{T_{A_1B_1A_2B_2}}\\
&\quad+
\left(\int d\psi\,
I_{A_1B_1}\otimes\psi_{A_2B_2}\otimes\psi_{A'B'}
\right)^{T_{A_1B_1A_2B_2}}\\
&\qquad=
\frac{2I_{\mathcal H}+C_{\rm diag}}{D(D+1)}.
\end{aligned}
\]
The definitions of \(W_\pm\) give
\[
\begin{aligned}
\Gamma_1&=\sqrt{\frac{D+1}{2}}\,W_+
     -\sqrt{\frac{D-1}{2}}\,W_-,\\
\Gamma_2&=\sqrt{\frac{D+1}{2}}\,W_+
     +\sqrt{\frac{D-1}{2}}\,W_-.
\end{aligned}
\]
Substituting the above expressions for \(\Gamma_1\) and \(\Gamma_2\)
into the definitions of \(C_{\rm diag}\) and \(C_{\rm cross}\),
we find that the mixed terms cancel in both sums:
\[
\begin{aligned}
C_{\rm diag}
&=2\frac{D+1}{2}P_{W_+}
  +2\frac{D-1}{2}P_{W_-}
=(D+1)P_{W_+}+(D-1)P_{W_-},\\
C_{\rm cross}
&=2\frac{D+1}{2}P_{W_+}
  -2\frac{D-1}{2}P_{W_-}
=(D+1)P_{W_+}-(D-1)P_{W_-}.
\end{aligned}
\]
In particular, both operators vanish on
\(\mathcal L_+\oplus\mathcal L_-\), and their scalar values on the four
joint symmetry sectors are
\[
\begin{array}{c|cccc}
&W_+&W_-&\mathcal L_+&\mathcal L_-\\ \hline
I_{\mathcal H}&1&1&1&1\\
\text{input-copy swap}&1&-1&1&-1\\
C_{\rm diag}&D+1&D-1&0&0\\
C_{\rm cross}&D+1&-(D-1)&0&0.
\end{array}
\]
Every operator in the displayed formula for \(M\) is scalar on each table
column.  Hence, for \(\chi\ne\chi'\),
\[
P_\chi MP_{\chi'}=0,
\qquad
P_\chi MP_\chi=\lambda_\chi P_\chi,
\]
so the table directly yields the spectral decomposition
\(M=\sum_\chi\lambda_\chi P_\chi\).
For a sector \(\chi\), denote the three entries in the input-copy-swap,
\(C_{\rm diag}\), and \(C_{\rm cross}\) rows by
\(s_\chi,c_\chi,r_\chi\).  Substitution into the displayed formula for \(M\)
first gives the single scalar formula
\[
\begin{aligned}
\lambda_\chi
={}&(1-\gamma)^2
\left[
\frac{1+s_\chi+c_\chi+r_\chi}{D(D+1)(D+2)}
-\left(1-\gamma+\frac{\gamma}{D}\right)
\frac{1+s_\chi}{D(D+1)}
\right]\\
&+\frac{\gamma(1-\gamma)}D
\left[
\frac{2+c_\chi}{D(D+1)}
-\frac{2}{D}\left(1-\gamma+\frac{\gamma}{D}\right)
\right]
-\frac{\gamma^2(1-\gamma)}{D^2}
\left(1-\frac1D\right).
\end{aligned}
\]
Simplifying gives
\[
\begin{aligned}
\lambda_{W_+}
&=\frac{(1-\gamma)\gamma(D-1)}{D^2(D+1)}
\left(1-\gamma+\frac{\gamma}{D}\right),\\
\lambda_{W_-}
&=-\frac{(1-\gamma)\gamma}{D^2}
\left(1-\gamma+\frac{\gamma}{D}\right),\\
\lambda_{\mathcal L_-}
&=-\frac{(1-\gamma)\gamma}{D(D+1)}
\left(2-\frac{D^2-1}{D^2}\gamma\right),\\
\lambda_{\mathcal L_+}
&=-\frac{1-\gamma}{D^3(D+1)(D+2)}\Bigl[
D(D+1)^2+2+2(D-2)(1-\gamma)\\
&\hspace{35mm}+(D-1)^2(D+2)(1-\gamma)^2
\Bigr].
\end{aligned}
\]
For our \(D\ge4\) and \(0<\gamma<1\), the last three eigenvalues are strictly
negative, while \(\lambda_{W_+}>0\).  Thus the only positive objective sector is
\(\operatorname{Ran}W_+\).

Since
\[
P_{W_+}
=\frac{(\Gamma_1+\Gamma_2)(\Gamma_1+\Gamma_2)^\dagger}{2(D+1)},
\]
we compute its output trace explicitly.  Write each composite input label as
\(x=(a,b)\) and \(y=(c,d)\).  On input basis vectors
\(|x,y\rangle,|x',y'\rangle\), contraction of \(A'B'\) gives
\[
\begin{aligned}
\langle x,y|\operatorname{Tr}_{A'B'}(\Gamma_1\Gamma_1^\dagger)|x',y'\rangle
&=\delta_{x,x'}\delta_{y,y'},\\
\langle x,y|\operatorname{Tr}_{A'B'}(\Gamma_2\Gamma_2^\dagger)|x',y'\rangle
&=\delta_{x,x'}\delta_{y,y'},\\
\langle x,y|\operatorname{Tr}_{A'B'}(\Gamma_1\Gamma_2^\dagger)|x',y'\rangle
&=\delta_{x,y'}\delta_{y,x'},\\
\langle x,y|\operatorname{Tr}_{A'B'}(\Gamma_2\Gamma_1^\dagger)|x',y'\rangle
&=\delta_{x,y'}\delta_{y,x'}.
\end{aligned}
\]
Therefore
\begin{equation}\label{eq:contraction_output_trace_identities}
\begin{aligned}
\operatorname{Tr}_{A'B'}(\Gamma_1\Gamma_1^\dagger)
&=\operatorname{Tr}_{A'B'}(\Gamma_2\Gamma_2^\dagger)
=I_{A_1B_1A_2B_2},\\
\operatorname{Tr}_{A'B'}(\Gamma_1\Gamma_2^\dagger)
&=\operatorname{Tr}_{A'B'}(\Gamma_2\Gamma_1^\dagger)
=\mathsf{S}_{A_1A_2}\otimes \mathsf{S}_{B_1B_2}.
\end{aligned}
\end{equation}
This calculation uses only the three-leg contraction pattern and hence holds
in any single-leg dimension.  It now follows that
\[
\operatorname{Tr}_{A'B'}P_{W_+}
=
\frac{
2I_{A_1B_1A_2B_2}
+2\mathsf{S}_{A_1A_2}\otimes \mathsf{S}_{B_1B_2}
}{2(D+1)}
=
\frac{2}{D+1}P_{\rm sym}.
\]
Hence
\[
J_\star=\frac{D+1}{2}P_{W_+},
\qquad
\operatorname{Tr}_{A'B'}J_\star
=P_{\rm sym}\preceq I_{A_1B_1A_2B_2},
\]
so \(J_\star\) is feasible and
\[
\operatorname{Tr}[MJ_\star]
=
\lambda_{W_+}\frac{D(D+1)}2
=
G_\star(D,\gamma).
\]

For an arbitrary feasible \(J\), put \(P_\perp=I_{\mathcal H}-P_{W_+}\).  The sector
decomposition has no cross blocks between \(W_+\) and \(P_\perp\), and hence
\[
M=\lambda_{W_+}P_{W_+}+P_\perp MP_\perp,
\]
where \(P_\perp MP_\perp\) is negative definite on
\(\operatorname{Ran}P_\perp\).
Since \(P_\perp JP_\perp\succeq0\), this gives
\[
\operatorname{Tr}[MJ]
=
\lambda_{W_+}\operatorname{Tr}[P_{W_+}J]
+
\operatorname{Tr}[P_\perp MP_\perp\,P_\perp JP_\perp]
\le
\lambda_{W_+}\operatorname{Tr}[P_{W_+}J].
\]
The range of \(W_+\) lies in the input-symmetric subspace, so
\[
P_{W_+}\preceq P_{\rm sym}\otimes I_{A'B'}.
\]
Therefore
\[
\operatorname{Tr}[P_{W_+}J]
\le
\operatorname{Tr}[P_{\rm sym}\operatorname{Tr}_{A'B'}J]
\le
\operatorname{Tr}P_{\rm sym}
=\frac{D(D+1)}2.
\]
Thus \(\operatorname{Tr}[MJ]\le G_\star(D,\gamma)\), and the feasible
\(J_\star\) above is optimal.

The same upper bound has an explicit dual certificate.  The dual of
\eqref{eq:gd_global_cptn_sdp} is
\[
\begin{aligned}
\min_Y\quad &\operatorname{Tr}Y\\
\text{\rm s.t.}\quad &Y\succeq0,\\
&Y\otimes I_{A'B'}\succeq M.
\end{aligned}
\]
Indeed,
\(M\preceq\lambda_{W_+}P_{W_+}\preceq
\lambda_{W_+}P_{\rm sym}\otimes I_{A'B'}\), so the choice
\(Y=\lambda_{W_+}P_{\rm sym}\) is dual feasible and has value
\(\operatorname{Tr}Y=G_\star(D,\gamma)\).

It remains to prove uniqueness.  Let \(J\) be any feasible operator attaining
\(G_\star(D,\gamma)\).  Since \(\lambda_{W_+}>0\), equality must hold throughout
\begin{equation}\label{eq:gd_optimizer_equality_chain}
G_\star(D,\gamma)
=\operatorname{Tr}[MJ]
\le
\lambda_{W_+}\operatorname{Tr}[P_{W_+}J]
\le
\lambda_{W_+}\operatorname{Tr}P_{\rm sym}
=G_\star(D,\gamma).
\end{equation}
Since
\(\lambda_{W_-},\lambda_{\mathcal L_+},\lambda_{\mathcal L_-}<0\), the restriction of \(M\)
to \(\operatorname{Ran}P_\perp\) is strictly negative.  Hence there is
\(\mu>0\) such that
\[
-P_\perp MP_\perp\succeq\mu P_\perp.
\]
Together with the decomposition of \(\operatorname{Tr}[MJ]\) above,
\eqref{eq:gd_optimizer_equality_chain} forces the complementary-sector
contribution to vanish:
\[
0=
\operatorname{Tr}\!\left[
(-P_\perp MP_\perp)P_\perp JP_\perp
\right]
\ge
\mu\operatorname{Tr}[P_\perp JP_\perp]
\ge0.
\]
Thus \(P_\perp JP_\perp=0\), and positivity gives
\(P_\perp J=JP_\perp=0\): indeed,
\(P_\perp JP_\perp=(J^{1/2}P_\perp)^\dagger
(J^{1/2}P_\perp)=0\), so \(J^{1/2}P_\perp=0\).
Consequently, \(J=P_{W_+}JP_{W_+}\).  With
\(A:=W_+^\dagger J W_+\) acting on \(\mathbb C^D\),
\[
A\succeq0,\qquad
J
=W_+W_+^\dagger J W_+W_+^\dagger
=W_+AW_+^\dagger.
\]
Moreover, \eqref{eq:gd_optimizer_equality_chain} gives
\[
\operatorname{Tr}A
=\operatorname{Tr}[P_{W_+}J]
=\operatorname{Tr}P_{\rm sym}
=\frac{D(D+1)}2.
\]
For every unit vector \(|u\rangle\in\mathbb C^D\), let \(|u,u\rangle\) denote the
vector having coordinates \(u\) on each of the two Choi-input coordinate copies.
Then
\[
\begin{aligned}
(\langle u,u|\otimes I_D)W_+
&=\frac{2}{\sqrt{2(D+1)}}
\left(\sum_i\langle u|i\rangle|i\rangle\right)\langle u|\\
&=
\sqrt{\frac{2}{D+1}}\,|u^*\rangle\langle u|.
\end{aligned}
\]
Consequently, using \(J=W_+AW_+^\dagger\) and tracing the remaining output
operator,
\[
\langle u,u|\operatorname{Tr}_{A'B'}J|u,u\rangle
=
\operatorname{Tr}\!\left[
\bigl((\langle u,u|\otimes I_D)W_+\bigr)
A
\bigl((\langle u,u|\otimes I_D)W_+\bigr)^\dagger
\right]
=
\frac{2}{D+1}\langle u|A|u\rangle,
\]
and the feasibility condition
\(\operatorname{Tr}_{A'B'}J\preceq I_{A_1B_1A_2B_2}\) therefore gives
\(A\preceq(D+1)I_D/2\).  Its trace already equals the trace of this upper
bound; therefore
\[
A=\frac{D+1}{2}I_D,
\qquad
J=\frac{D+1}{2}P_{W_+}=J_\star.
\]
\end{proof}

Comparing \eqref{eq:one_epr_accepted_choi} with the unique optimizer above
shows that the accepted map of Algorithm~\ref{alg:one_epr_purification} has
Choi operator \(J_\star\).

\subsection{PPT relaxation}
\label{app:ppt_relaxation}

Before introducing local coordinates, recall the instrument description of an
LOCC protocol.  A finite-round LOCC protocol is a sequence of local quantum
instruments in which the instrument selected at a later round may depend on
the classical outcomes obtained in earlier rounds.  Along each fixed classical
history, the corresponding CP maps are composed; for example, two successive
outcomes \(j\) and \(k\) produce the branch
\(\mathcal B_{k|j}\circ\mathcal A_j\). By contrast, coarse-graining adds the
CP maps associated with histories assigned the same retained outcome
\cite{chitambar2014everything}.

Let \((\mathcal E_\lambda)_{\lambda\in\Theta}\) be the terminal instrument of
the uninserted LOCC protocol, where \(\lambda\) records its complete classical
history.  Partition the terminal histories into the disjoint accepted and
rejected sets \(\Theta_{\mathrm S}\) and \(\Theta_{\mathrm F}\), so that
\(\Theta=\Theta_{\mathrm S}\cup\Theta_{\mathrm F}\) and
\(\Theta_{\mathrm S}\cap\Theta_{\mathrm F}=\varnothing\).
The CPTN map \(\mathcal E\) used in the main text is the coarse-graining of the
accepted histories.
\[
\mathcal E
=
\sum_{\lambda\in\Theta_{\mathrm S}}\mathcal E_\lambda.
\]
We take all terminal maps to have output space \(A'B'\); rejected histories can
be completed by a fixed local preparation.  Retaining an explicit classical
outcome register then gives the quantum--classical channel
\[
X
\longmapsto
\left(
\sum_{\lambda\in\Theta_{\mathrm S}}\mathcal E_\lambda(X)
\right)\otimes|\mathrm S\rangle\langle\mathrm S|
+
\left(
\sum_{\lambda\in\Theta_{\mathrm F}}\mathcal E_\lambda(X)
\right)\otimes|\mathrm F\rangle\langle\mathrm F|.
\]
The Choi operator of the accepted coarse-graining is
\[
S
:=
J_{\mathcal E}
=
\sum_{\lambda\in\Theta_{\mathrm S}}J_{\mathcal E_\lambda}.
\]
The Choi operator of the rejected coarse-graining is
\[
F
:=
\sum_{\lambda\in\Theta_{\mathrm F}}J_{\mathcal E_\lambda}.
\]
They act on \(\mathcal H\otimes R\) before the resource state is inserted.

Under the input-first Choi convention, inserting a normalized resource
\(\omega_R\) into a branch with Choi operator \(J_{\mathcal E}\) gives
\begin{equation}\label{eq:resource_insertion_convention}
J_{\mathcal E(\,\cdot\,\otimes\omega_R)}
=
\operatorname{Tr}_R\!\left[
\left(I_{\mathcal H}\otimes\omega_R^T\right)J_{\mathcal E}
\right].
\end{equation}
In particular, for the accepted branch \(J_{\mathcal E}=S\), the objective is
\[
\operatorname{Tr}\!\left[
M J_{\mathcal E(\,\cdot\,\otimes\omega_R)}
\right]
=
\operatorname{Tr}\!\left[(M\otimes\omega_R^T)S\right].
\]
Thus the full transpose in the objective comes from resource insertion; it is
distinct from the Bob-side partial transpose in the PPT constraints.

Fix a normalized resource \(\omega_R\).  Since the inclusions from LOCC to SEP
and then to PPT hold branchwise, replacing LOCC by branchwise PPT constraints
gives the relaxation
\begin{equation}\label{eq:fixed_resource_ppt_sdp}
\begin{aligned}
\max_{S,F}\quad
&\operatorname{Tr}\!\left[(M\otimes\omega_R^T)S\right]\\
\text{\rm s.t.}\quad
&S,F\succeq0,\\
&S^{T_{B_1B_2B'B_{\mathrm r}}}\succeq0,\qquad
F^{T_{B_1B_2B'B_{\mathrm r}}}\succeq0,\\
&\operatorname{Tr}_{A'B'}(S+F)
=I_{A_1B_1A_2B_2}\otimes I_R,
\end{aligned}
\end{equation}
where \(S,F\in\mathcal B(\mathcal H\otimes R)\).  The final constraint is
imposed on the full instrument before resource insertion.  Inserting the
normalized resource contracts \(I_R\) to \(1\), leaving the trace-preservation
identity on \(A_1B_1A_2B_2\).  The PPT conditions apply only to the uninserted
instrument; after resource insertion, the effective branches need only satisfy
the global CPTN constraints and need not remain PPT across Alice and Bob.

\subsection{Local mixed-Schur reduction}

We next exploit the local symmetry of the fixed-resource relaxation to express
its branches in independent Alice and Bob coordinates.
Let the three global Choi legs in dimension \(D\) be
\[
(A_1B_1),(A_2B_2),(A'B'),
\]
and reorder the corresponding local factors as
\[
A_1,A_2,A',B_1,B_2,B'.
\]
When the total dimension needs emphasis, write \(W_+^{(D)}\) for the isometry
in \eqref{eq:gd_Wplus_definition}.
For each side \(X\in\{A,B\}\), set
\[
H_X:=
\mathbb C^d_{X_1}
\otimes\mathbb C^d_{X_2}
\otimes\mathbb C^d_{X'} .
\]
The corresponding local representation is
\[
\pi_X(U):=\overline U_{X_1}\otimes\overline U_{X_2}\otimes U_{X'},
\qquad U\in U(d).
\]
Under the fixed Alice--Bob reorder above,
\[
\mathcal H\simeq H_A\otimes H_B .
\]
\begin{lemma}[Local-twirl reduction]
\label{lem:local_twirl_reduction}
For every fixed normalized resource \(\omega_R\), any feasible pair \(S,F\)
in \eqref{eq:fixed_resource_ppt_sdp} can be replaced, without changing the
objective value, by a feasible pair invariant under conjugation by
\[
\pi_A(U)\otimes\pi_B(V)\otimes I_R,
\qquad U,V\in U(d).
\]
\end{lemma}

\begin{proof}
For each \(K\in\{S,F\}\), replace \(K\) by the independent local Haar average
\[
\mathbb E_{U,V}\!\left[
\bigl(\pi_A(U)\otimes\pi_B(V)\otimes I_R\bigr)K
\bigl(\pi_A(U)\otimes\pi_B(V)\otimes I_R\bigr)^\dagger
\right].
\]
By left invariance of Haar measure, the displayed average is invariant under
the conjugations in the statement.  If \(K\succeq0\), every integrand is
positive, so the average is positive.

For the PPT constraint, linearity of partial transpose and its reversal of
products on Bob's tensor factors give the explicit identity
\[
\begin{aligned}
&\left\{
\mathbb E_{U,V}\!\left[
\bigl(\pi_A(U)\otimes\pi_B(V)\otimes I_R\bigr)K
\bigl(\pi_A(U)\otimes\pi_B(V)\otimes I_R\bigr)^\dagger
\right]
\right\}^{T_{B_1B_2B'B_{\mathrm r}}}\\
&\quad=
\mathbb E_{U,V}\!\left[
\bigl(\pi_A(U)\otimes\overline{\pi_B(V)}\otimes I_R\bigr)
K^{T_{B_1B_2B'B_{\mathrm r}}}
\bigl(\pi_A(U)\otimes\overline{\pi_B(V)}\otimes I_R\bigr)^\dagger
\right]
\succeq0.
\end{aligned}
\]
Here
\[
\overline{\pi_B(V)}
=V_{B_1}\otimes V_{B_2}\otimes\overline V_{B'}
\]
is unitary.  Hence each integrand on the right is a unitary conjugate of
\(K^{T_{B_1B_2B'B_{\mathrm r}}}\succeq0\), which proves the displayed
inequality.

The output trace removes the output conjugations and leaves only the input
ones.  Applying this to the averaged sum of the two branches gives
\[
\begin{aligned}
&\operatorname{Tr}_{A'B'}
\mathbb E_{U,V}\!\left[
\bigl(\pi_A(U)\otimes\pi_B(V)\otimes I_R\bigr)(S+F)
\bigl(\pi_A(U)\otimes\pi_B(V)\otimes I_R\bigr)^\dagger
\right]\\
&\quad=
\mathbb E_{U,V}\!\left[
\bigl(
\overline U_{A_1}\otimes\overline U_{A_2}
\otimes\overline V_{B_1}\otimes\overline V_{B_2}\otimes I_R
\bigr)
\operatorname{Tr}_{A'B'}(S+F)
\bigl(
\overline U_{A_1}\otimes\overline U_{A_2}
\otimes\overline V_{B_1}\otimes\overline V_{B_2}\otimes I_R
\bigr)^\dagger
\right]\\
&\quad=
I_{A_1B_1A_2B_2}\otimes I_R,
\end{aligned}
\]
where the last equality uses the trace identity in
\eqref{eq:fixed_resource_ppt_sdp}.  Thus that identity is preserved.

Under the fixed Alice--Bob reorder,
\(\pi_A(U)\otimes\pi_B(V)=\pi_D(U\otimes V)\).  Since
\(U\otimes V\in U(D)\), the definition
\eqref{eq:gd_global_performance_operator}, covariance of
\(\mathcal N^\gamma\), and the Haar change of variable
\(\psi\mapsto(U\otimes V)\psi(U\otimes V)^\dagger\) give
\[
\bigl(\pi_A(U)\otimes\pi_B(V)\bigr)^\dagger
M
\bigl(\pi_A(U)\otimes\pi_B(V)\bigr)
=M.
\]
Because the twirl acts as the identity on \(R\), cyclicity of trace now gives
\[
\begin{aligned}
&\operatorname{Tr}\!\left[
(M\otimes\omega_R^T)
\mathbb E_{U,V}\!\left[
\bigl(\pi_A(U)\otimes\pi_B(V)\otimes I_R\bigr)S
\bigl(\pi_A(U)\otimes\pi_B(V)\otimes I_R\bigr)^\dagger
\right]
\right]\\
&\quad=
\mathbb E_{U,V}\operatorname{Tr}\!\left[
\bigl(\pi_A(U)\otimes\pi_B(V)\otimes I_R\bigr)^\dagger
(M\otimes\omega_R^T)
\bigl(\pi_A(U)\otimes\pi_B(V)\otimes I_R\bigr)S
\right]\\
&\quad=
\operatorname{Tr}\!\left[(M\otimes\omega_R^T)S\right].
\end{aligned}
\]
Therefore the objective value is unchanged.
\end{proof}

Define the normalized local parity isometries
\(W_\pm^X:\mathbb C^d\to H_X\) directly by
\[
W_+^X|a\rangle
=
\frac{1}{\sqrt{2(d+1)}}\sum_i
\left(
|a,i,i\rangle_{H_X}
+
|i,a,i\rangle_{H_X}
\right),
\qquad
W_-^X|a\rangle
=
\frac{1}{\sqrt{2(d-1)}}\sum_i
\left(
|a,i,i\rangle_{H_X}
-
|i,a,i\rangle_{H_X}
\right).
\]
The swap on the two local input legs is \(\mathsf{S}_{X_1X_2}\).  Its extension to
\(H_X\), acting trivially on the output, is
\(\mathsf{S}_{X_1X_2}\otimes I_{X'}\).  The local maps satisfy
\begin{equation}\label{eq:local_W_structural_identities}
(W_\pm^X)^\dagger W_\pm^X=I_d,
\qquad (W_+^X)^\dagger W_-^X=0,
\qquad
\pi_X(U)W_\pm^X=W_\pm^X\overline U,
\qquad
(\mathsf{S}_{X_1X_2}\otimes I_{X'})W_\pm^X=\pm W_\pm^X.
\end{equation}

The mixed-Schur analysis in the proof of
Lemma~\ref{lem:general_d_global_optimum} applies to the three-leg
representation \(\overline U^{\otimes2}\otimes U\) in an arbitrary single-leg
dimension.  Since \(H_X\) carries precisely this representation through
\(\pi_X\), the same labels and path multiplicities apply with \(D\) replaced by
\(d\).  For \(d\ge3\), all three leaves
\(\lambda_0,\lambda_+,\lambda_-\) are admissible.  For \(d=2\), the condition
\(\ell(\lambda_l)+\ell(\lambda_r)\le d\) removes \(\lambda_-\), while
\(\lambda_0\) still has two paths and \(\lambda_+\) has one.  The identities
above exhibit two orthogonal copies of the \(\lambda_0\)-type, so
\[
\operatorname{Ran}W_+^X\oplus\operatorname{Ran}W_-^X
\]
is the full local contraction isotypic component.  When \(d=2\), these two
ranges have total dimension four, the remaining \(+1\)-parity component has
dimension four, and the remaining \(-1\)-parity component is zero; the total is
\(8=\dim H_X\).

Define the local input-parity projectors by
\[
P_{\rm sym}^X:=\frac{I_{X_1X_2}+\mathsf{S}_{X_1X_2}}{2},
\qquad
P_{\rm asym}^X:=\frac{I_{X_1X_2}-\mathsf{S}_{X_1X_2}}{2},
\]
where they are extended trivially on \(X'\) when acting on \(H_X\).
Equation~\eqref{eq:local_W_structural_identities} gives
\[
\left(P_{\rm sym}^X\otimes I_{X'}\right)W_+^X=W_+^X,
\qquad
\left(P_{\rm asym}^X\otimes I_{X'}\right)W_-^X=W_-^X.
\]
Define the residual local parity subspaces by
\[
\begin{aligned}
\mathcal L_+^X
&:=
\operatorname{Ran}\left(P_{\rm sym}^X\otimes I_{X'}\right)
\cap\left(\operatorname{Ran}W_+^X\right)^\perp,\\
\mathcal L_-^X
&:=
\operatorname{Ran}\left(P_{\rm asym}^X\otimes I_{X'}\right)
\cap\left(\operatorname{Ran}W_-^X\right)^\perp.
\end{aligned}
\]
Their orthogonal projectors are
\[
P_{\mathcal L_+^X}
:=
P_{\rm sym}^X\otimes I_{X'}-W_+^X(W_+^X)^\dagger,
\qquad
P_{\mathcal L_-^X}
:=
P_{\rm asym}^X\otimes I_{X'}-W_-^X(W_-^X)^\dagger.
\]
Expanding the explicit definition of \(P_{\rm sym}\) in the two local parity
projectors gives
\begin{equation}\label{eq:global_input_parity_local_decomposition}
\begin{aligned}
P_{\rm sym}
&=P_{\rm sym}^A\otimes P_{\rm sym}^B
+P_{\rm asym}^A\otimes P_{\rm asym}^B,\\
P_{\rm asym}
&=P_{\rm sym}^A\otimes P_{\rm asym}^B
+P_{\rm asym}^A\otimes P_{\rm sym}^B.
\end{aligned}
\end{equation}
The local mixed-Schur decomposition is therefore
\begin{equation}\label{eq:local_mixed_schur_direct_sum}
H_X
=
\operatorname{Ran}W_+^X
\oplus
\operatorname{Ran}W_-^X
\oplus
\mathcal L_+^X
\oplus
\mathcal L_-^X.
\end{equation}
Their dimensions are
\[
\dim\mathcal L_+^X=\frac{d(d-1)(d+2)}2,
\qquad
\dim\mathcal L_-^X=\frac{d(d+1)(d-2)}2.
\]
The subspaces \(\mathcal L_\pm^X\) are invariant under \(\pi_X(U)\) and are the
residual mixed-Schur components complementary to the local contraction
isotypic component.
Thus \(\mathcal L_-^X=\{0\}\) when \(d=2\), consistently with the
admissibility argument above.

For each \(X\in\{A,B\}\), let \(\mathcal M_X\) be an abstract
two-dimensional multiplicity coordinate space with fixed real orthonormal
basis
\[
\{|+\rangle_{\mathcal M_X},|-\rangle_{\mathcal M_X}\}.
\]
For later use, abbreviate the joint multiplicity space by
\[
\mathcal M_{AB}:=\mathcal M_A\otimes\mathcal M_B.
\]

Define the contraction-sector coordinate isometry
\(G_X:\mathcal M_X\otimes\mathbb C^d\to H_X\)
by
\[
G_X\bigl(|+\rangle_{\mathcal M_X}\otimes|a\rangle\bigr)
=
W_+^X|a\rangle,
\qquad
G_X\bigl(|-\rangle_{\mathcal M_X}\otimes|a\rangle\bigr)
=
W_-^X|a\rangle.
\]
The definition also fixes the adjoint explicitly:
\[
\begin{aligned}
G_X
&=
\sum_{s\in\{+,-\}}
W_s^X\circ\bigl(\langle s|_{\mathcal M_X}\otimes I_d\bigr),\\
G_X^\dagger
&=
\sum_{s\in\{+,-\}}
\bigl(|s\rangle_{\mathcal M_X}\otimes I_d\bigr)\circ(W_s^X)^\dagger,
\qquad
G_X^\dagger W_t^X=|t\rangle_{\mathcal M_X}\otimes I_d.
\end{aligned}
\]
The last identity uses the orthogonality relations in
\eqref{eq:local_W_structural_identities}.
By the preceding identification, \(G_X\) is unitary from
\(\mathcal M_X\otimes\mathbb C^d\) onto the full contraction isotypic
component.  Its group-intertwining property follows directly on a coordinate
basis:
\[
\begin{aligned}
&\pi_X(U)G_X\bigl(|s\rangle_{\mathcal M_X}\otimes|a\rangle\bigr)\\
&\quad=
\pi_X(U)W_s^X|a\rangle
=W_s^X\overline U|a\rangle\\
&\quad=
G_X\bigl(I_{\mathcal M_X}\otimes\overline U\bigr)
\bigl(|s\rangle_{\mathcal M_X}\otimes|a\rangle\bigr).
\end{aligned}
\]
Since these vectors span the coordinate space, and the swap relation follows
from \eqref{eq:local_W_structural_identities},
\begin{equation}\label{eq:local_contraction_representation_coordinates}
\begin{aligned}
\pi_X(U)G_X
&=G_X\bigl(I_{\mathcal M_X}\otimes\overline U\bigr),\\
G_X^\dagger(\mathsf{S}_{X_1X_2}\otimes I_{X'})G_X
&=
\left(
|+\rangle\langle+|_{\mathcal M_X}
-|-\rangle\langle-|_{\mathcal M_X}
\right)\otimes I_d.
\end{aligned}
\end{equation}
Thus \(G_X\) puts the local contraction component in Schur form: the
group action is trivial on \(\mathcal M_X\) and irreducible on the
\(\mathbb C^d\) factor, while the local input swap is diagonal on the fixed
multiplicity basis.

After reordering the tensor factors, the tensor-product isometry
\[
G_A\otimes G_B:
\mathcal M_A\otimes\mathcal M_B
\otimes\mathbb C_A^d\otimes\mathbb C_B^d
\longrightarrow H_A\otimes H_B
\]
has range equal to the tensor product of the two local contraction components,
and
\[
H_A\otimes H_B
=
\operatorname{Ran}(G_A\otimes G_B)
\oplus\operatorname{Ran}(G_A\otimes G_B)^\perp.
\]
The pullback to these local contraction coordinates is the map
\[
\begin{gathered}
(G_A\otimes G_B)^\dagger(\,\cdot\,)(G_A\otimes G_B):
\mathcal B(H_A\otimes H_B)
\longrightarrow
\mathcal B\left(
\mathcal M_A\otimes\mathcal M_B
\otimes\mathbb C_A^d\otimes\mathbb C_B^d
\right).
\\[-1mm]
L\longmapsto
(G_A\otimes G_B)^\dagger L(G_A\otimes G_B).
\end{gathered}
\]
It records the matrix elements of a physical operator within the local
contraction component in the chosen Schur coordinates.
The reverse operation embeds a coordinate-space operator back into the physical
space:
\[
\begin{gathered}
(G_A\otimes G_B)(\,\cdot\,)(G_A\otimes G_B)^\dagger:
\mathcal B\left(
\mathcal M_A\otimes\mathcal M_B
\otimes\mathbb C_A^d\otimes\mathbb C_B^d
\right)
\longrightarrow
\mathcal B(H_A\otimes H_B).
\\[-1mm]
Q\longmapsto
(G_A\otimes G_B)Q(G_A\otimes G_B)^\dagger.
\end{gathered}
\]
This embedded operator is supported on
\(\operatorname{Ran}(G_A\otimes G_B)\).

Set
\[
\alpha:=\frac{d+1}{2},
\qquad
\beta:=\frac{d-1}{2},
\qquad
|v_d\rangle:=\alpha|++\rangle+\beta|--\rangle.
\]

\begin{lemma}[Local contraction block of the global optimizer]
\label{lem:global_optimizer_local_block}
In the local Schur coordinates of the contraction component, the global
optimizer satisfies
\[
(G_A\otimes G_B)^\dagger J_\star (G_A\otimes G_B)
=
|v_d\rangle\langle v_d|\otimes I_D.
\]
\end{lemma}

\begin{proof}
Write every \(D\)-dimensional basis label as an Alice--Bob pair.  Under the
fixed reordering, the definition of \(W_+^{(D)}\) gives
\[
\begin{aligned}
W_+^{(D)}|a,b\rangle
&=
\frac{1}{\sqrt{2(D+1)}}\sum_{i,j}
\Bigl(
|a,i,i\rangle_{H_A}\otimes|b,j,j\rangle_{H_B}
+|i,a,i\rangle_{H_A}\otimes|j,b,j\rangle_{H_B}
\Bigr)\\
&=
\frac{d+1}{\sqrt{2(D+1)}}
W_+^A|a\rangle\otimes W_+^B|b\rangle
+
\frac{d-1}{\sqrt{2(D+1)}}
W_-^A|a\rangle\otimes W_-^B|b\rangle .
\end{aligned}
\]
The second equality follows by expanding the definitions of
\(W_\pm^A,W_\pm^B\); the two mixed-sign terms cancel.
Applying \((G_A\otimes G_B)^\dagger\) and using the definition of
\(|v_d\rangle\) gives the isometry-level branching identity
\[
(G_A\otimes G_B)^\dagger W_+^{(D)}
=
\sqrt{\frac{2}{D+1}}\,|v_d\rangle\otimes I_D.
\]
Since \(J_\star=(D+1)P_{W_+^{(D)}}/2\) and
\(P_{W_+^{(D)}}=W_+^{(D)}W_+^{(D)\dagger}\), it follows that
\[
\begin{aligned}
(G_A\otimes G_B)^\dagger J_\star (G_A\otimes G_B)
&=
\frac{D+1}{2}
\bigl((G_A\otimes G_B)^\dagger W_+^{(D)}\bigr)
\bigl((G_A\otimes G_B)^\dagger W_+^{(D)}\bigr)^\dagger\\
&=
|v_d\rangle\langle v_d|\otimes I_D,
\end{aligned}
\]
as claimed.
\end{proof}
\subsection{Compression of the assisted branches}

With the resource register included, order the contraction coordinates as
\[
\mathcal M_{AB}\otimes R
\otimes\mathbb C_A^d\otimes\mathbb C_B^d
\]
and define the resource-extended contraction isometry
\[
\mathcal G:
\mathcal M_{AB}\otimes R
\otimes\mathbb C_A^d\otimes\mathbb C_B^d
\longrightarrow H_A\otimes H_B\otimes R
\]
on the coordinate basis by
\[
\begin{aligned}
&\mathcal G\bigl(
|s_A\rangle_{\mathcal M_A}\otimes|s_B\rangle_{\mathcal M_B}
\otimes|\zeta\rangle_R\otimes|a\rangle_A\otimes|b\rangle_B
\bigr)\\
&\quad :=
G_A\bigl(|s_A\rangle_{\mathcal M_A}\otimes|a\rangle_A\bigr)
\otimes
G_B\bigl(|s_B\rangle_{\mathcal M_B}\otimes|b\rangle_B\bigr)
\otimes|\zeta\rangle_R,
\end{aligned}
\]
where \(s_A,s_B\in\{+,-\}\), and \(|\zeta\rangle,|a\rangle,|b\rangle\) range over
fixed bases of \(R,\mathbb C_A^d,\mathbb C_B^d\), respectively; extend linearly.
Thus \(\mathcal G\) is the resource extension of \(G_A\otimes G_B\), with tensor
factors in the displayed order.  It is a real isometry with
\[
\operatorname{Ran}\mathcal G
=
\Bigl[
(\operatorname{Ran}W_+^A\oplus\operatorname{Ran}W_-^A)
\otimes
(\operatorname{Ran}W_+^B\oplus\operatorname{Ran}W_-^B)
\otimes R
\Bigr].
\]
Its adjoint is determined on this range by
\[
\begin{aligned}
&\mathcal G^\dagger\left(
W_{s_A}^A|a\rangle_A\otimes W_{s_B}^B|b\rangle_B
\otimes|\zeta\rangle_R
\right)\\
&\quad=
|s_A\rangle_{\mathcal M_A}\otimes|s_B\rangle_{\mathcal M_B}
\otimes|\zeta\rangle_R\otimes|a\rangle_A\otimes|b\rangle_B,
\end{aligned}
\]
and it vanishes on \((\operatorname{Ran}\mathcal G)^\perp\).
The corresponding resource-extended pullback is
\[
\begin{gathered}
\mathcal G^\dagger(\,\cdot\,)\mathcal G:
\mathcal B(H_A\otimes H_B\otimes R)
\longrightarrow
\mathcal B\left(
\mathcal M_{AB}\otimes R
\otimes\mathbb C_A^d\otimes\mathbb C_B^d
\right).
\\[-1mm]
K\longmapsto
\mathcal G^\dagger K\mathcal G.
\end{gathered}
\]
The reverse conjugation is the physical embedding of a coordinate-space
operator:
\[
\begin{gathered}
\mathcal G(\,\cdot\,)\mathcal G^\dagger:
\mathcal B\left(
\mathcal M_{AB}\otimes R
\otimes\mathbb C_A^d\otimes\mathbb C_B^d
\right)
\longrightarrow
\mathcal B(H_A\otimes H_B\otimes R).
\\[-1mm]
Q\longmapsto
\mathcal G Q\mathcal G^\dagger.
\end{gathered}
\]
Since \(\mathcal G\) is an isometry, \(\mathcal G\mathcal G^\dagger\) is the
orthogonal projector onto \(\operatorname{Ran}\mathcal G\).
Composing the pullback and embedding gives the physical contraction-block map
\[
\begin{gathered}
(\mathcal G\mathcal G^\dagger)(\,\cdot\,)
(\mathcal G\mathcal G^\dagger):
\mathcal B(H_A\otimes H_B\otimes R)
\longrightarrow
\mathcal B(\operatorname{Ran}\mathcal G),
\\[-1mm]
K\longmapsto
(\mathcal G\mathcal G^\dagger)K
(\mathcal G\mathcal G^\dagger)
=
\mathcal G(\mathcal G^\dagger K\mathcal G)\mathcal G^\dagger.
\end{gathered}
\]
On the right, an operator on \(\operatorname{Ran}\mathcal G\) is identified
with its extension by zero on \((\operatorname{Ran}\mathcal G)^\perp\).
For \(s_A,s_B,t_A,t_B\in\{+,-\}\) and \(Y_R\in\mathcal B(R)\), the preceding
adjoint formula and \eqref{eq:local_W_structural_identities} give
\begin{equation}\label{eq:resource_contraction_element_correspondence}
\begin{aligned}
&\mathcal G^\dagger\left[
W_{s_A}^A(W_{t_A}^A)^\dagger
\otimes W_{s_B}^B(W_{t_B}^B)^\dagger
\otimes Y_R
\right]\mathcal G \\
&\qquad=
\left(
|s_A\rangle\langle t_A|_{\mathcal M_A}
\otimes|s_B\rangle\langle t_B|_{\mathcal M_B}
\otimes Y_R
\right)\otimes I_D,\\[1mm]
&\mathcal G\left[
\left(
|s_A\rangle\langle t_A|_{\mathcal M_A}
\otimes|s_B\rangle\langle t_B|_{\mathcal M_B}
\otimes Y_R
\right)\otimes I_D
\right]\mathcal G^\dagger \\
&\qquad=
W_{s_A}^A(W_{t_A}^A)^\dagger
\otimes W_{s_B}^B(W_{t_B}^B)^\dagger
\otimes Y_R.
\end{aligned}
\end{equation}
The second identity follows by conjugating the first back with \(\mathcal G\):
the physical operator on its right is already supported on
\(\operatorname{Ran}\mathcal G\), so the inserted projectors
\(\mathcal G\mathcal G^\dagger\) act trivially.
Applying the first identity in
\eqref{eq:local_contraction_representation_coordinates} to \(X=A\) and \(X=B\),
and including the trivial action on \(R\), gives the joint intertwining identity
\begin{equation}\label{eq:joint_contraction_intertwining}
\bigl(\pi_A(U)\otimes\pi_B(V)\otimes I_R\bigr)\mathcal G
=
\mathcal G
\bigl(I_{\mathcal M_{AB}R}\otimes\overline U\otimes\overline V\bigr).
\end{equation}
Together with invariance of \(K\) under the independent local twirls, this
implies
\[
\left[
\mathcal G^\dagger K\mathcal G,\,
I_{\mathcal M_{AB}R}\otimes
\overline U\otimes\overline V
\right]=0
\qquad(U,V\in U(d)).
\]
The identity on \(\mathcal M_{AB}\otimes R\) means that the group acts
trivially there, so the commutation relation leaves arbitrary operator freedom
on that factor.  Since
\(\overline U\otimes\overline V\) is irreducible for \(U(d)\times U(d)\), the
commutant on the displayed block is
\[
\mathcal B(\mathcal M_{AB}\otimes R)\otimes I_D.
\]
Thus there is a unique operator
\[
K_{\mathcal M_{AB}R}
\in\mathcal B(\mathcal M_{AB}\otimes R)
\]
such that
\[
\mathcal G^\dagger K\mathcal G
=
K_{\mathcal M_{AB}R}\otimes I_D.
\]
Here \(I_D=I_{\mathbb C_A^d}\otimes I_{\mathbb C_B^d}\).
Introduce the sector-label set
\[
\mathcal Z:=\{(+,+),(+,-),(-,+),(-,-)\}.
\]
For \(z=(s_A,s_B)\in\mathcal Z\), write
\[
|z\rangle:=|s_A\rangle_{\mathcal M_A}\otimes
|s_B\rangle_{\mathcal M_B}.
\]
The first and second signs always refer to Alice and Bob, respectively.

\subsection{Output-trace decomposition}

For either side \(X\in\{A,B\}\), expand the direct definitions of
\(W_\pm^X\).  Taking the trace over the third leg uses
\(\operatorname{Tr}(\ket i\!\bra j)=\delta_{ij}\): the two diagonal products
give \(I_{X_1X_2}\), while the two cross products give
\(\mathsf S_{X_1X_2}\).  Hence
\[
\begin{aligned}
\operatorname{Tr}_{X^{\prime}}[W_+^X(W_+^X)^\dagger]
&=
\frac{I_{X_1X_2}+I_{X_1X_2}+\mathsf{S}_{X_1X_2}+\mathsf{S}_{X_1X_2}}{2(d+1)}
=
\frac{2}{d+1}P_{\rm sym}^X,\\
\operatorname{Tr}_{X^{\prime}}[W_-^X(W_-^X)^\dagger]
&=
\frac{I_{X_1X_2}+I_{X_1X_2}-\mathsf{S}_{X_1X_2}-\mathsf{S}_{X_1X_2}}{2(d-1)}
=
\frac{2}{d-1}P_{\rm asym}^X,\\
\operatorname{Tr}_{X^{\prime}}[W_+^X(W_-^X)^\dagger]
&=
\frac{\mathsf{S}_{X_1X_2}-I_{X_1X_2}+I_{X_1X_2}-\mathsf{S}_{X_1X_2}}
{2\sqrt{(d+1)(d-1)}}
=0.
\end{aligned}
\]
The remaining off-diagonal identity is the adjoint of the last line.  The signs
\(+\) and \(-\) correspond to the symmetric and antisymmetric input sectors,
respectively.  Write the elements of \(\mathcal Z\) as \(++,+-,-+,--\) in
subscripts.

For \(K\in\{S,F\}\) and \(z,w\in\mathcal Z\), define the resource-operator
blocks of the parent contraction operator \(K_{\mathcal M_{AB}R}\) by
\[
[K_{\mathcal M_{AB}R}]_{z,w}
:=
(\langle z|\otimes I_R)
K_{\mathcal M_{AB}R}
(|w\rangle\otimes I_R).
\]
For \(z=(s_A,s_B)\) and \(w=(t_A,t_B)\), the corresponding multiplicity
matrix unit is explicitly
\[
|z\rangle\langle w|
=
|s_A\rangle\langle t_A|_{\mathcal M_A}
\otimes|s_B\rangle\langle t_B|_{\mathcal M_B}.
\]
Because \(|z\rangle\) is an orthonormal basis of \(\mathcal M_{AB}\), these
matrix elements reconstruct the
full resource-valued contraction block:
\[
K_{\mathcal M_{AB}R}
=
\sum_{\substack{z=(s_A,s_B)\in\mathcal Z\\w=(t_A,t_B)\in\mathcal Z}}
\left(
|s_A\rangle\langle t_A|_{\mathcal M_A}
\otimes|s_B\rangle\langle t_B|_{\mathcal M_B}
\right)\otimes[K_{\mathcal M_{AB}R}]_{z,w},
\qquad [K_{\mathcal M_{AB}R}]_{z,w}\in\mathcal B(R).
\]
Applying \eqref{eq:resource_contraction_element_correspondence} to each
summand in the expansion of \(K_{\mathcal M_{AB}R}\), with
\(Y_R=[K_{\mathcal M_{AB}R}]_{z,w}\), gives
\[
\begin{aligned}
&\mathcal G\left[
\left(
|s_A\rangle\langle t_A|_{\mathcal M_A}
\otimes|s_B\rangle\langle t_B|_{\mathcal M_B}
\otimes[K_{\mathcal M_{AB}R}]_{z,w}
\right)\otimes I_D
\right]\mathcal G^\dagger \\
&\quad =
W_{s_A}^A(W_{t_A}^A)^\dagger
\otimes
W_{s_B}^B(W_{t_B}^B)^\dagger
\otimes[K_{\mathcal M_{AB}R}]_{z,w},
\end{aligned}
\]
This equality holds on the contraction summand of
\(H_A\otimes H_B\otimes R\).
For example, taking \(z=w=(+,-)\) in this formula gives
\[
\begin{aligned}
&\mathcal G\left[
\left(
|+\rangle\langle+|_{\mathcal M_A}
\otimes|-\rangle\langle-|_{\mathcal M_B}
\otimes[K_{\mathcal M_{AB}R}]_{+-,+-}
\right)\otimes I_D
\right]\mathcal G^\dagger\\
&\quad=
W_+^A(W_+^A)^\dagger
\otimes W_-^B(W_-^B)^\dagger
\otimes[K_{\mathcal M_{AB}R}]_{+-,+-}.
\end{aligned}
\]
Taking the output trace of the generic term gives
\[
\begin{aligned}
&\operatorname{Tr}_{A'B'}\!\left[
W_{s_A}^A(W_{t_A}^A)^\dagger
\otimes
W_{s_B}^B(W_{t_B}^B)^\dagger
\otimes[K_{\mathcal M_{AB}R}]_{z,w}
\right] \\
&\quad =
\operatorname{Tr}_{A'}[W_{s_A}^A(W_{t_A}^A)^\dagger]
\otimes
\operatorname{Tr}_{B'}[W_{s_B}^B(W_{t_B}^B)^\dagger]
\otimes[K_{\mathcal M_{AB}R}]_{z,w}.
\end{aligned}
\]
Using \(\alpha=(d+1)/2\) and \(\beta=(d-1)/2\), the preceding identities are
\(\operatorname{Tr}_{X'}[W_+^X(W_+^X)^\dagger]
=\alpha^{-1}P_{\rm sym}^X\),
\(\operatorname{Tr}_{X'}[W_-^X(W_-^X)^\dagger]
=\beta^{-1}P_{\rm asym}^X\), and the two cross terms are zero.
Therefore the full contraction summand contributes, in the input--resource
order,
\begin{equation}\label{eq:local_contraction_output_trace_sector_expansion}
\begin{aligned}
&\operatorname{Tr}_{A'B'}\!\left[\mathcal G\left(K_{\mathcal M_{AB}R}\otimes I_D\right)\mathcal G^\dagger\right]\\
&\quad=
\alpha^{-2}P_{\rm sym}^A\otimes P_{\rm sym}^B
\otimes[K_{\mathcal M_{AB}R}]_{++,++}\\
&\qquad
+(\alpha\beta)^{-1}P_{\rm sym}^A\otimes P_{\rm asym}^B
\otimes[K_{\mathcal M_{AB}R}]_{+-,+-}\\
&\qquad
+(\alpha\beta)^{-1}P_{\rm asym}^A\otimes P_{\rm sym}^B
\otimes[K_{\mathcal M_{AB}R}]_{-+,-+}\\
&\qquad
+\beta^{-2}P_{\rm asym}^A\otimes P_{\rm asym}^B
\otimes[K_{\mathcal M_{AB}R}]_{--,--}.
\end{aligned}
\end{equation}
Every off-diagonal block \([K_{\mathcal M_{AB}R}]_{z,w}\) with \(z\ne w\)
disappears because at least one side contains
\(\operatorname{Tr}_{X^{\prime}}[W_+^X(W_-^X)^\dagger]=0\) or its adjoint.
Partial trace is equivariant under unitary conjugation on the discarded
outputs.  Hence local twirling gives the shorter commutant statement
\begin{equation}\label{eq:output_trace_input_commutant}
\left[
\operatorname{Tr}_{A'B'}K,\,
\bigl(\overline U^{\otimes2}\bigr)_A
\otimes\bigl(\overline V^{\otimes2}\bigr)_B\otimes I_R
\right]
=0.
\end{equation}
Here the subscripts \(A\) and \(B\) refer to the input pairs \(A_1A_2\) and
\(B_1B_2\), respectively.
After tracing out \(A'B'\), the remaining non-resource input space is
\[
(\mathbb C^d_{A_1}\otimes\mathbb C^d_{A_2})
\otimes
(\mathbb C^d_{B_1}\otimes\mathbb C^d_{B_2}).
\]
On each side, the ordinary two-copy Schur decomposition is the orthogonal
direct sum
\[
\mathbb C^d_{X_1}\otimes\mathbb C^d_{X_2}
=
\operatorname{Ran}P_{\rm sym}^X
\oplus
\operatorname{Ran}P_{\rm asym}^X,
\qquad X\in\{A,B\}.
\]
Taking the tensor product of the two decompositions gives
\[
\begin{aligned}
&
(\mathbb C^d_{A_1}\otimes\mathbb C^d_{A_2})
\otimes
(\mathbb C^d_{B_1}\otimes\mathbb C^d_{B_2})\\
&\quad=
(\operatorname{Ran}P_{\rm sym}^A\otimes
 \operatorname{Ran}P_{\rm sym}^B)
\oplus
(\operatorname{Ran}P_{\rm sym}^A\otimes
 \operatorname{Ran}P_{\rm asym}^B)\\
&\qquad\quad\oplus
(\operatorname{Ran}P_{\rm asym}^A\otimes
 \operatorname{Ran}P_{\rm sym}^B)
\oplus
(\operatorname{Ran}P_{\rm asym}^A\otimes
 \operatorname{Ran}P_{\rm asym}^B).
\end{aligned}
\]
In the displayed order, these are the \(++,+-,-+,--\) input sectors.  They are
pairwise inequivalent irreducible representations of \(U(d)\times U(d)\).
Let \(P,Q\) range over the four product projectors displayed above, and let
\(\operatorname{Hom}_G\) denote the space of intertwiners for a group \(G\).
The identities
\[
\mathsf S_{A_1A_2}\overline U^{\otimes2}
=\overline U^{\otimes2}\mathsf S_{A_1A_2},
\qquad
\mathsf S_{B_1B_2}\overline V^{\otimes2}
=\overline V^{\otimes2}\mathsf S_{B_1B_2}
\]
show that both \((I+\mathsf S_{X_1X_2})/2\) and
\((I-\mathsf S_{X_1X_2})/2\) commute with the corresponding two-copy action.
Hence \(P,Q\) commute with the input action.  Together with
\eqref{eq:output_trace_input_commutant}, this gives
\[
\left[
(\overline U^{\otimes2})_A\otimes(\overline V^{\otimes2})_B\otimes I_R,
(P\otimes I_R)\operatorname{Tr}_{A'B'}K(Q\otimes I_R)
\right]=0.
\]
Restricted to \(\operatorname{Ran}Q\otimes R\), the second operator maps into
\(\operatorname{Ran}P\otimes R\), and the displayed equality is precisely the
intertwining relation between the two sector representations.  Since the group
acts trivially on \(R\), Schur's lemma gives
\begin{equation}\label{eq:output_trace_sector_schur_relations}
\operatorname{Hom}_{U(d)\times U(d)}
(\operatorname{Ran}Q\otimes R,\operatorname{Ran}P\otimes R)
=
\begin{cases}
\{0\},&P\ne Q,\\
I_{\operatorname{Ran}P}\otimes\mathcal B(R),&P=Q.
\end{cases}
\end{equation}
For \(P\ne Q\), the restricted block is zero, and hence
\((P\otimes I_R)\operatorname{Tr}_{A'B'}K(Q\otimes I_R)=0\).  For \(P=Q\),
the second line makes the diagonal block \(P\) tensored with a unique operator
on \(R\).  Therefore there are unique operators
\(\left[\operatorname{Tr}_{A'B'}K\right]_\tau\in\mathcal B(R)\),
\(\tau\in\mathcal Z\), such that
\begin{equation}\label{eq:output_trace_sector_coefficient_definition}
\begin{aligned}
\operatorname{Tr}_{A'B'}K
={}&
P_{\rm sym}^A\otimes P_{\rm sym}^B
\otimes\left[\operatorname{Tr}_{A'B'}K\right]_{++}\\
&+P_{\rm sym}^A\otimes P_{\rm asym}^B
\otimes\left[\operatorname{Tr}_{A'B'}K\right]_{+-}\\
&+P_{\rm asym}^A\otimes P_{\rm sym}^B
\otimes\left[\operatorname{Tr}_{A'B'}K\right]_{-+}\\
&+P_{\rm asym}^A\otimes P_{\rm asym}^B
\otimes\left[\operatorname{Tr}_{A'B'}K\right]_{--}.
\end{aligned}
\end{equation}
Thus the two bracket conventions keep their parent operators visible:
\([K_{\mathcal M_{AB}R}]_{z,w}\) extracts a matrix block from the contraction
operator, whereas \(\left[\operatorname{Tr}_{A'B'}K\right]_\tau\) extracts a
sector coefficient after the output trace.

For every \(\tau\in\mathcal Z\), the standing positivity and trace identity give
\[
\begin{aligned}
0&\preceq
\left[\operatorname{Tr}_{A'B'}S\right]_\tau
\preceq I_R,\\
0&\preceq
\left[\operatorname{Tr}_{A'B'}F\right]_\tau
\preceq I_R,\\
\left[\operatorname{Tr}_{A'B'}S\right]_\tau
&+
\left[\operatorname{Tr}_{A'B'}F\right]_\tau
 =I_R.
\end{aligned}
\]

\subsection{Residual-sector trace bounds}

\begin{lemma}[Contraction-sector separation and output-trace bounds]
\label{lem:contraction_sector_trace_bounds}
Let \(K\in\{S,F\}\) be locally twirled.  In the four-summand tensor-product
decomposition obtained from \eqref{eq:local_mixed_schur_direct_sum}, the
contraction summand is reducing for \(K\): both off-diagonal blocks between it
and its orthogonal complement vanish.  Consequently, the residual block is
positive, and the output-trace sector coefficients satisfy
\begin{equation}\label{eq:branch_sector_trace_lower_bounds}
\begin{aligned}
\left[\operatorname{Tr}_{A'B'}K\right]_{++}
&\succeq\alpha^{-2}[K_{\mathcal M_{AB}R}]_{++,++},&
\left[\operatorname{Tr}_{A'B'}K\right]_{+-}
&\succeq(\alpha\beta)^{-1}[K_{\mathcal M_{AB}R}]_{+-,+-},\\
\left[\operatorname{Tr}_{A'B'}K\right]_{-+}
&\succeq(\alpha\beta)^{-1}[K_{\mathcal M_{AB}R}]_{-+,-+},&
\left[\operatorname{Tr}_{A'B'}K\right]_{--}
&\succeq\beta^{-2}[K_{\mathcal M_{AB}R}]_{--,--}.
\end{aligned}
\end{equation}
\end{lemma}

\begin{proof}
We now start from the local decomposition already established above:
\[
H_X
=
\operatorname{Ran}W_+^X
\oplus
\operatorname{Ran}W_-^X
\oplus
\mathcal L_+^X
\oplus
\mathcal L_-^X,
\qquad X\in\{A,B\}.
\]
We apply this decomposition once for Alice and once for Bob.  Taking their
tensor product and distributing the direct sums gives the following
decomposition of the full physical space:
\[
\begin{aligned}
H_A\otimes H_B\otimes R
={}&
\Bigl[
(\operatorname{Ran}W_+^A\oplus\operatorname{Ran}W_-^A)
\otimes
(\operatorname{Ran}W_+^B\oplus\operatorname{Ran}W_-^B)
\otimes R
\Bigr]\\
&\oplus
\Bigl[
(\operatorname{Ran}W_+^A\oplus\operatorname{Ran}W_-^A)
\otimes
(\mathcal L_+^B\oplus\mathcal L_-^B)
\otimes R
\Bigr]\\
&\oplus
\Bigl[
(\mathcal L_+^A\oplus\mathcal L_-^A)
\otimes
(\operatorname{Ran}W_+^B\oplus\operatorname{Ran}W_-^B)
\otimes R
\Bigr]\\
&\oplus
\Bigl[
(\mathcal L_+^A\oplus\mathcal L_-^A)
\otimes
(\mathcal L_+^B\oplus\mathcal L_-^B)
\otimes R
\Bigr].
\end{aligned}
\]
The first summand is the sector in which both local factors belong to the
contraction isotypic components.  By the definition of \(\mathcal G\), it is
exactly the range of \(\mathcal G\).  The remaining three summands are exactly
the vectors with at least one residual local factor, and therefore form the
orthogonal complement of that range:
\begin{equation}\label{eq:contraction_range_orthogonal_complement}
\begin{aligned}
(\operatorname{Ran}\mathcal G)^\perp
={}&
\Bigl[
(\operatorname{Ran}W_+^A\oplus\operatorname{Ran}W_-^A)
\otimes
(\mathcal L_+^B\oplus\mathcal L_-^B)
\otimes R
\Bigr]\\
&\oplus
\Bigl[
(\mathcal L_+^A\oplus\mathcal L_-^A)
\otimes
(\operatorname{Ran}W_+^B\oplus\operatorname{Ran}W_-^B)
\otimes R
\Bigr]\\
&\oplus
\Bigl[
(\mathcal L_+^A\oplus\mathcal L_-^A)
\otimes
(\mathcal L_+^B\oplus\mathcal L_-^B)
\otimes R
\Bigr].
\end{aligned}
\end{equation}
In this four-summand decomposition, the first summand is the contraction
sector and the remaining three summands contain at least one residual local
factor.  The off-diagonal blocks of \(K\), connecting the first summand with
the sum of the remaining three, are
\[
\begin{aligned}
&(\mathcal G\mathcal G^\dagger)K
(I-\mathcal G\mathcal G^\dagger),\\
&(I-\mathcal G\mathcal G^\dagger)K
(\mathcal G\mathcal G^\dagger).
\end{aligned}
\]
With the same intertwiner notation, the local mixed-Schur decomposition
\eqref{eq:local_mixed_schur_direct_sum} gives
\begin{equation}\label{eq:local_residual_contraction_hom_zero}
\operatorname{Hom}_{U(d)}
\left(
\mathcal L_+^X\oplus\mathcal L_-^X,
\operatorname{Ran}W_+^X\oplus\operatorname{Ran}W_-^X
\right)
=\{0\},
\qquad X\in\{A,B\},
\end{equation}
because the residual summands have representation types inequivalent to the
contraction type carried by
\(\operatorname{Ran}W_+^X\oplus\operatorname{Ran}W_-^X\).  When \(d=2\), the
summand \(\mathcal L_-^X\) in this formula is simply the zero space.

Equation~\eqref{eq:joint_contraction_intertwining} gives
\[
\begin{aligned}
&\bigl(\pi_A(U)\otimes\pi_B(V)\otimes I_R\bigr)
\mathcal G\mathcal G^\dagger
\bigl(\pi_A(U)\otimes\pi_B(V)\otimes I_R\bigr)^\dagger\\
&\quad=
\mathcal G
\bigl(I_{\mathcal M_{AB}R}\otimes\overline U\otimes\overline V\bigr)
\bigl(I_{\mathcal M_{AB}R}\otimes\overline U\otimes\overline V\bigr)^\dagger
\mathcal G^\dagger\\
&\quad=\mathcal G\mathcal G^\dagger.
\end{aligned}
\]
Thus \(\mathcal G\mathcal G^\dagger\) commutes with the physical
representation.  Local twirling gives the same commutation relation for
\(K\).  Therefore the first off-diagonal block satisfies
\[
\begin{aligned}
&\left(\pi_A(U)\otimes \pi_B(V)\otimes I_R\right)
\left[(\mathcal G\mathcal G^\dagger)K
(I-\mathcal G\mathcal G^\dagger)\right]\\
&\quad=
\left[(\mathcal G\mathcal G^\dagger)K
(I-\mathcal G\mathcal G^\dagger)\right]
\left(\pi_A(U)\otimes \pi_B(V)\otimes I_R\right).
\end{aligned}
\]
Therefore
\[
(\mathcal G\mathcal G^\dagger)K
(I-\mathcal G\mathcal G^\dagger)
\in
\operatorname{Hom}_{U(d)\times U(d)}
\left((\operatorname{Ran}\mathcal G)^\perp,\operatorname{Ran}\mathcal G\right).
\]
The three summands in
\eqref{eq:contraction_range_orthogonal_complement} show why this global Hom
space is zero.  In the first, Bob's factor is
\(\mathcal L_+^B\oplus\mathcal L_-^B\), so
\eqref{eq:local_residual_contraction_hom_zero} with \(X=B\) rules out an
intertwiner into Bob's contraction factor.  The second is ruled out by the same
equation with \(X=A\), and the third by either choice.  Hence
\[
\operatorname{Hom}_{U(d)\times U(d)}
\left((\operatorname{Ran}\mathcal G)^\perp,\operatorname{Ran}\mathcal G\right)
=\{0\}.
\]
Consequently Schur's lemma gives
\[
\begin{aligned}
(\mathcal G\mathcal G^\dagger)K
(I-\mathcal G\mathcal G^\dagger)&=0,\\
(I-\mathcal G\mathcal G^\dagger)K
(\mathcal G\mathcal G^\dagger)&=0.
\end{aligned}
\]
The physical contraction-block map consequently extracts exactly the
corresponding diagonal block of \(K\):
\[
\begin{aligned}
(\mathcal G\mathcal G^\dagger)K
(\mathcal G\mathcal G^\dagger)
&=
\mathcal G(\mathcal G^\dagger K\mathcal G)\mathcal G^\dagger\\
&=
\mathcal G\left(
K_{\mathcal M_{AB}R}\otimes I_D
\right)\mathcal G^\dagger.
\end{aligned}
\]
Consequently
\begin{equation}\label{eq:branch_contraction_residual_decomposition}
\begin{aligned}
K
&=
\mathcal G\left(
K_{\mathcal M_{AB}R}\otimes I_D
\right)\mathcal G^\dagger+K_\perp,\\
K_\perp
&:=
(I-\mathcal G\mathcal G^\dagger)K
(I-\mathcal G\mathcal G^\dagger)\succeq0.
\end{aligned}
\end{equation}
The last inequality follows because \(K_\perp\) is a projection compression of
the positive operator \(K\).  It also inherits local-twirl invariance.  We
therefore define
\(\left[\operatorname{Tr}_{A'B'}K_\perp\right]_\tau\) by the sector expansion
\eqref{eq:output_trace_sector_coefficient_definition} with \(K\) replaced by
\(K_\perp\).  Applying
\eqref{eq:local_contraction_output_trace_sector_expansion} to the first term
in \eqref{eq:branch_contraction_residual_decomposition}, and taking the
positive output trace of \(K_\perp\) sector by sector, gives
\begin{equation}\label{eq:branch_sector_trace_decomposition}
\begin{aligned}
\left[\operatorname{Tr}_{A'B'}K\right]_{++}
&=\alpha^{-2}[K_{\mathcal M_{AB}R}]_{++,++}
+\left[\operatorname{Tr}_{A'B'}K_\perp\right]_{++},\\
\left[\operatorname{Tr}_{A'B'}K\right]_{+-}
&=(\alpha\beta)^{-1}[K_{\mathcal M_{AB}R}]_{+-,+-}
+\left[\operatorname{Tr}_{A'B'}K_\perp\right]_{+-},\\
\left[\operatorname{Tr}_{A'B'}K\right]_{-+}
&=(\alpha\beta)^{-1}[K_{\mathcal M_{AB}R}]_{-+,-+}
+\left[\operatorname{Tr}_{A'B'}K_\perp\right]_{-+},\\
\left[\operatorname{Tr}_{A'B'}K\right]_{--}
&=\beta^{-2}[K_{\mathcal M_{AB}R}]_{--,--}
+\left[\operatorname{Tr}_{A'B'}K_\perp\right]_{--}.
\end{aligned}
\end{equation}
Moreover,
\begin{equation}\label{eq:residual_output_trace_sector_positive}
\left[\operatorname{Tr}_{A'B'}K_\perp\right]_\tau\succeq0,
\qquad \tau\in\mathcal Z.
\end{equation}
Thus \eqref{eq:branch_sector_trace_decomposition} proves the four inequalities
displayed in the statement.

\end{proof}

\subsection{PPT constraints after compression}

Let \(T_{\mathcal M_B}\) denote transpose in the fixed real orthonormal basis
\(\{|+\rangle_{\mathcal M_B},|-\rangle_{\mathcal M_B}\}\).  On multiplicity
matrix units it acts as
\[
\left[
|s_A\rangle\langle t_A|_{\mathcal M_A}
\otimes|s_B\rangle\langle t_B|_{\mathcal M_B}
\right]^{T_{\mathcal M_B}}
=
|s_A\rangle\langle t_A|_{\mathcal M_A}
\otimes|t_B\rangle\langle s_B|_{\mathcal M_B},
\qquad s_A,s_B,t_A,t_B\in\{+,-\}.
\]
\begin{lemma}[PPT constraints in contraction coordinates]
\label{lem:mixed_schur_ppt_interface}
Let \(K\in\{S,F\}\) be locally twirled.  Then
\begin{equation}\label{eq:mixed_schur_ppt_contraction_block}
\mathcal G^\dagger
K^{T_{B_1B_2B'B_{\mathrm r}}}
\mathcal G
=
\bigl(K_{\mathcal M_{AB}R}\bigr)^{T_{\mathcal M_B}T_{B_{\mathrm r}}}
\otimes I_D,
\end{equation}
and every output-trace sector coefficient satisfies
\begin{equation}\label{eq:mixed_schur_trace_ppt_order}
\left(
\left[\operatorname{Tr}_{A'B'}K\right]_\tau
\right)^{T_{B_{\mathrm r}}}
\succeq0,
\qquad \tau\in\mathcal Z.
\end{equation}
\end{lemma}

\begin{proof}
For coordinate basis vectors write
\[
|s_A,s_B;j,k;a,b\rangle
=
|s_A\rangle_{\mathcal M_A}\otimes|s_B\rangle_{\mathcal M_B}
\otimes|j,k\rangle_R\otimes|a\rangle_A\otimes|b\rangle_B.
\]
For product vectors, fixed-basis partial transpose obeys
\[
\langle x_A,y_B|L^{T_B}|x_A',y_B'\rangle
=
\langle x_A,\overline{y_B'}|L|x_A',\overline{y_B}\rangle.
\]
Here the vectors \(W_s^X|a\rangle\) and all coordinate basis vectors are real,
so the conjugates in this rule disappear.  Expanding \(\mathcal G\) and applying
the rule on the physical and coordinate spaces gives
\[
\begin{aligned}
&\langle s_A,s_B;j,k;a,b|
\mathcal G^\dagger K^{T_{B_1B_2B'B_{\mathrm r}}}\mathcal G
|t_A,t_B;\ell,m;c,e\rangle\\
&\quad=
\left(
\langle a|(W_{s_A}^A)^\dagger
\otimes\langle b|(W_{s_B}^B)^\dagger
\otimes\langle j,k|
\right)
K^{T_{B_1B_2B'B_{\mathrm r}}}
\left(
W_{t_A}^A|c\rangle\otimes W_{t_B}^B|e\rangle
\otimes|\ell,m\rangle
\right)\\
&\quad=
\left(
\langle a|(W_{s_A}^A)^\dagger
\otimes\langle e|(W_{t_B}^B)^\dagger
\otimes\langle j,m|
\right)
K
\left(
W_{t_A}^A|c\rangle\otimes W_{s_B}^B|b\rangle
\otimes|\ell,k\rangle
\right)\\
&\quad=
\langle s_A,t_B;j,m;a,e|
\mathcal G^\dagger K\mathcal G
|t_A,s_B;\ell,k;c,b\rangle\\
&\quad=
\langle s_A,s_B;j,k;a,b|
\bigl(\mathcal G^\dagger K\mathcal G\bigr)^{
T_{\mathcal M_B}T_{\mathbb C_B^d}T_{B_{\mathrm r}}}
|t_A,t_B;\ell,m;c,e\rangle.
\end{aligned}
\]
Since this holds for all coordinate basis vectors,
\[
\mathcal G^\dagger K^{T_{B_1B_2B'B_{\mathrm r}}}\mathcal G
=
\bigl(\mathcal G^\dagger K\mathcal G\bigr)^{
T_{\mathcal M_B}T_{\mathbb C_B^d}T_{B_{\mathrm r}}}.
\]
Local twirling gives
\[
\mathcal G^\dagger K\mathcal G
=
K_{\mathcal M_{AB}R}\otimes I_D.
\]
Since \(I_{\mathbb C_B^d}^{T}=I_{\mathbb C_B^d}\), this gives
\eqref{eq:mixed_schur_ppt_contraction_block}.

The transpose on \(B_1B_2B_{\mathrm r}\) commutes with the output trace, while
the transpose on the traced register \(B'\) leaves that trace unchanged.  Thus
\[
\begin{aligned}
\operatorname{Tr}_{A'B'}\!\left[K^{T_{B_1B_2B'B_{\mathrm r}}}\right]
&=
\left(\operatorname{Tr}_{A'B'}\!\left[K^{T_{B'}}\right]\right)^{
T_{B_1B_2B_{\mathrm r}}}\\
&=
\left(\operatorname{Tr}_{A'B'}K\right)^{T_{B_1B_2B_{\mathrm r}}}
\succeq0.
\end{aligned}
\]
The four input-sector projectors are real, mutually orthogonal, and nonzero.
For each \(\tau\in\mathcal Z\), choose any unit vector \(u_\tau\) in the
corresponding sector.  Compressing
\(\left(\operatorname{Tr}_{A'B'}K\right)^{T_{B_1B_2B_{\mathrm r}}}\succeq0\)
to
\(\operatorname{span}\{|u_\tau\rangle\}\otimes R\) gives, for every
\(|\xi\rangle\in R\),
\[
\begin{aligned}
0
&\le
\left\langle u_\tau,\xi\left|
\left(\operatorname{Tr}_{A'B'}K\right)^{T_{B_1B_2B_{\mathrm r}}}
\right|u_\tau,\xi\right\rangle\\
&=
\left\langle\xi\left|
\left(
\left[\operatorname{Tr}_{A'B'}K\right]_\tau
\right)^{T_{B_{\mathrm r}}}
\right|\xi\right\rangle.
\end{aligned}
\]
Since \(|\xi\rangle\) is arbitrary, this proves
\eqref{eq:mixed_schur_trace_ppt_order}.
\end{proof}
\subsection{Resource insertion at the global optimum}

The optimal global Choi operator \(J_\star\) acts on
\(H_A\otimes H_B\), while the uninserted success and failure branches \(S,F\)
act on \(H_A\otimes H_B\otimes R\).
By Lemma~\ref{lem:global_optimizer_local_block}, the pullback of \(J_\star\)
to the local contraction coordinates has multiplicity block
\(|v_d\rangle\langle v_d|\).
The convention \eqref{eq:resource_insertion_convention} applies equally to
operator blocks carrying the resource input.  For a pure resource, choose
Schmidt coordinates in which \(\ket\eta\) is real.
Then \((\ket\eta\!\bra\eta)^T=\ket\eta\!\bra\eta\), and we write
\[
\begin{aligned}
X^{\eta,\eta}
&:=
\operatorname{Tr}_R\!\left[
\left(I\otimes(\ket\eta\!\bra\eta)^T\right)X
\right]\\
&=
(I\otimes\langle\eta|_R)X(I\otimes|\eta\rangle_R).
\end{aligned}
\]
Here \(I\) denotes the identity on all non-resource tensor factors of \(X\).
For \(K\in\{S,F\}\), retain \(\mathcal M_{AB}R\) in the subscript to record the
parent contraction space, and use the superscript to indicate that the resource
register has been inserted.  Thus
\[
\begin{aligned}
K_{\mathcal M_{AB}R}^{\eta,\eta}
&:=
\operatorname{Tr}_R\!\left[
\left(I_{\mathcal M_{AB}}\otimes(\ket\eta\!\bra\eta)^T\right)
K_{\mathcal M_{AB}R}
\right]\\
&=
(I_{\mathcal M_{AB}}\otimes\langle\eta|_R)
K_{\mathcal M_{AB}R}
(I_{\mathcal M_{AB}}\otimes|\eta\rangle_R)
\in\mathcal B(\mathcal M_{AB}).
\end{aligned}
\]
Taking a multiplicity matrix element commutes with resource insertion:
\begin{equation}\label{eq:block_insertion_compatibility}
\left[K_{\mathcal M_{AB}R}^{\eta,\eta}\right]_{z,w}
=
\left\langle\eta\left|
[K_{\mathcal M_{AB}R}]_{z,w}
\right|\eta\right\rangle,
\qquad z,w\in\mathcal Z.
\end{equation}

\begin{lemma}[Optimality fixes the inserted success branch]
\label{lem:gd_exposed_standard_representative}
Assume \(0<\gamma<1\).  If a locally twirled PPT success/failure pair assisted
by \(\ket\eta\) attains the optimal global CPTN value
\(G_\star(D,\gamma)\), then its inserted success branch is
\begin{equation}\label{eq:gd_inserted_global_optimizer}
S^{\eta,\eta}
=
J_\star
=
\frac{D+1}{2}P_{W_+^{(D)}}.
\end{equation}
In particular, inserting \(\ket\eta\) into its contraction block gives an
operator on \(\mathcal M_{AB}\) satisfying
\begin{equation}\label{eq:gd_standard_block}
S_{\mathcal M_{AB}R}^{\eta,\eta}
=
\ket{v_d}\bra{v_d}.
\end{equation}
\end{lemma}

\begin{proof}
Use the inserted branches \(S^{\eta,\eta}\) and \(F^{\eta,\eta}\) defined
above.  For every vector \(|y\rangle\in H_A\otimes H_B\), positivity is
preserved because
\[
\langle y|K^{\eta,\eta}|y\rangle
=(\langle y|\otimes\langle\eta|)K(|y\rangle\otimes|\eta\rangle)
\ge0,
\qquad K\in\{S,F\}.
\]
Resource insertion also commutes with the output trace.  Using
\(\langle\eta|\eta\rangle=1\) and the standing trace identity gives
\[
\begin{aligned}
\operatorname{Tr}_{A'B'}
\left(S^{\eta,\eta}+F^{\eta,\eta}\right)
&=
(I_{A_1B_1A_2B_2}\otimes\langle\eta|_R)
\operatorname{Tr}_{A'B'}(S+F)
(I_{A_1B_1A_2B_2}\otimes|\eta\rangle_R)\\
&=
(I_{A_1B_1A_2B_2}\otimes\langle\eta|_R)
(I_{A_1B_1A_2B_2}\otimes I_R)
(I_{A_1B_1A_2B_2}\otimes|\eta\rangle_R)\\
&=
I_{A_1B_1A_2B_2}.
\end{aligned}
\]
Since \(\operatorname{Tr}_{A'B'}F^{\eta,\eta}\succeq0\), this also yields
\[
0\preceq\operatorname{Tr}_{A'B'}S^{\eta,\eta}
\preceq I_{A_1B_1A_2B_2}.
\]
Thus \(S^{\eta,\eta}\) is feasible for the global CPTN SDP.
By the attainment assumption,
\[
\operatorname{Tr}[M S^{\eta,\eta}]
=G_\star(D,\gamma).
\]
Uniqueness in Lemma~\ref{lem:general_d_global_optimum} forces
\[
S^{\eta,\eta}
=J_\star
=\frac{D+1}{2}P_{W_+^{(D)}}.
\]
Lemma~\ref{lem:global_optimizer_local_block} gives
\eqref{eq:gd_standard_block}.
\end{proof}
\subsection{Resource-effect certificate in Schmidt coordinates}

Let \(\ket\eta\) be a finite-Schmidt-rank pure resource of Schmidt rank \(r\).
Choose Schmidt bases on its local supports.  The standard Schmidt form is
\[
|\eta\rangle
=
\sum_{j=0}^{r-1}\sqrt{p_j}|jj\rangle_R,
\qquad
p_0\ge\cdots\ge p_{r-1}>0,
\qquad
\sum_{j=0}^{r-1}p_j=1.
\]
Set \(p_j=0\) for \(j\ge r\).  We may now flip the phases of Bob's Schmidt
vectors with labels \(j\ge1\) and conjugate \(S,F\) by the same local unitary.
This preserves positivity, PPT, the trace identities, and the inserted-block
matrix elements.  We therefore use the phase convention
\[
|\eta\rangle
=
\sqrt{p_0}|00\rangle_R-\sum_{j\ge1}\sqrt{p_j}|jj\rangle_R,
\qquad
\sum_{j\ge0}p_j=1,
\]
Set
\[
\Pi_\eta^\perp:=I_R-|\eta\rangle\langle\eta|.
\]

For the purification problem, exact optimality supplies
\eqref{eq:gd_standard_block} through
Lemma~\ref{lem:gd_exposed_standard_representative}; the two additional
zero-sector conditions are then verified directly.
\begin{lemma}[PPT-induced resource-effect certificate]
\label{lem:ppt_effect_certificate}
Fix the pure resource \(|\eta\rangle\) in the Schmidt form and phase convention
above.  Let a locally twirled pair feasible for
\eqref{eq:fixed_resource_ppt_sdp} with
\(\omega_R=|\eta\rangle\langle\eta|\) also obey
\eqref{eq:gd_standard_block} and
\begin{equation}\label{eq:gd_zero_input_sector_success}
\left\langle\eta\left|
\left[\operatorname{Tr}_{A'B'}S\right]_{+-}
\right|\eta\right\rangle
=
\left\langle\eta\left|
\left[\operatorname{Tr}_{A'B'}S\right]_{-+}
\right|\eta\right\rangle
=0 .
\end{equation}
Then \(r\ge2\), and there exists an effect \(N\in\mathcal B(R)\) such that
\begin{equation}\label{eq:ppt_effect_certificate}
0\preceq N\preceq I_R,
\qquad
\frac13
\le
\langle00|\Pi_\eta^\perp N\Pi_\eta^\perp|00\rangle
\le1-p_0,
\qquad
\operatorname{Re}
\langle00|\Pi_\eta^\perp N\Pi_\eta^\perp|11\rangle
\le\frac13.
\end{equation}
\end{lemma}

\begin{proof}
\paragraph{Block identities and trace bounds.}
The pullback preserves positivity because
\(\langle x|\mathcal G^\dagger S\mathcal G|x\rangle
=\langle\mathcal Gx|S|\mathcal Gx\rangle\ge0\) for every \(|x\rangle\).
Together with the pullback identity, this gives
\begin{equation}\label{eq:success_contraction_block_positive}
\mathcal G^\dagger S\mathcal G
=
S_{\mathcal M_{AB}R}\otimes I_D
\succeq0,
\end{equation}
and hence \(S_{\mathcal M_{AB}R}\succeq0\).  Likewise,
\eqref{eq:mixed_schur_ppt_contraction_block} and the physical PPT constraint
give
\[
\mathcal G^\dagger S^{T_{B_1B_2B'B_{\mathrm r}}}\mathcal G
=
\bigl(S_{\mathcal M_{AB}R}\bigr)^{T_{\mathcal M_B}T_{B_{\mathrm r}}}
\otimes I_D
\succeq0,
\]
so
\begin{equation}\label{eq:success_contraction_block_ppt_positive}
\bigl(S_{\mathcal M_{AB}R}\bigr)^{T_{\mathcal M_B}T_{B_{\mathrm r}}}
\succeq0.
\end{equation}
Lemma~\ref{lem:contraction_sector_trace_bounds}, applied to \(K=S\), gives
the four contraction-sector lower bounds.
On the other hand, the trace identity and \(F\succeq0\) give
\begin{equation}\label{eq:success_failure_sector_complement}
I_R-\left[\operatorname{Tr}_{A'B'}S\right]_\tau
=
\left[\operatorname{Tr}_{A'B'}F\right]_\tau
\succeq0,
\qquad \tau\in\mathcal Z.
\end{equation}
Combining the two sides yields
\begin{equation}\label{eq:success_output_trace_sector_bounds}
\begin{aligned}
0&\preceq \alpha^{-2}[S_{\mathcal M_{AB}R}]_{++,++}
\preceq\left[\operatorname{Tr}_{A'B'}S\right]_{++}\preceq I_R,\\
0&\preceq (\alpha\beta)^{-1}[S_{\mathcal M_{AB}R}]_{+-,+-}
\preceq\left[\operatorname{Tr}_{A'B'}S\right]_{+-}\preceq I_R,\\
0&\preceq (\alpha\beta)^{-1}[S_{\mathcal M_{AB}R}]_{-+,-+}
\preceq\left[\operatorname{Tr}_{A'B'}S\right]_{-+}\preceq I_R,\\
0&\preceq \beta^{-2}[S_{\mathcal M_{AB}R}]_{--,--}
\preceq\left[\operatorname{Tr}_{A'B'}S\right]_{--}\preceq I_R.
\end{aligned}
\end{equation}
The inserted block identity
\[
S_{\mathcal M_{AB}R}^{\eta,\eta}
=
|v_d\rangle\langle v_d|,
\qquad
|v_d\rangle=\alpha|++\rangle+\beta|--\rangle,
\]
gives, for every \(z\in\mathcal Z\),
\[
\left\langle\eta\left|
[S_{\mathcal M_{AB}R}]_{z,z}
\right|\eta\right\rangle
=
\langle z|S_{\mathcal M_{AB}R}^{\eta,\eta}|z\rangle
=
|\langle z|v_d\rangle|^2.
\]
Thus
\begin{equation}\label{eq:inserted_diagonal_block_expectations}
\begin{gathered}
\left\langle\eta\left|[S_{\mathcal M_{AB}R}]_{+-,+-}\right|\eta\right\rangle
=
\left\langle\eta\left|[S_{\mathcal M_{AB}R}]_{-+,-+}\right|\eta\right\rangle=0,\\
\left\langle\eta\left|[S_{\mathcal M_{AB}R}]_{++,++}\right|\eta\right\rangle=\alpha^2,
\qquad
\left\langle\eta\left|[S_{\mathcal M_{AB}R}]_{--,--}\right|\eta\right\rangle=\beta^2.
\end{gathered}
\end{equation}
For either cross-sector diagonal block, positivity gives, for example,
\[
0=
\left\langle\eta\left|[S_{\mathcal M_{AB}R}]_{+-,+-}\right|\eta\right\rangle
=
\left\|
\bigl([S_{\mathcal M_{AB}R}]_{+-,+-}\bigr)^{1/2}|\eta\rangle
\right\|^2,
\]
so \([S_{\mathcal M_{AB}R}]_{+-,+-}|\eta\rangle=0\), and the same
calculation applies to \([S_{\mathcal M_{AB}R}]_{-+,-+}\).  For the \(++\)
block, \eqref{eq:inserted_diagonal_block_expectations} and
\(\langle\eta|I_R|\eta\rangle=1\) give
\[
\begin{aligned}
0
&=
\left\langle\eta\left|
I_R-\alpha^{-2}[S_{\mathcal M_{AB}R}]_{++,++}
\right|\eta\right\rangle\\
&=
\left\|
\left(
I_R-\alpha^{-2}[S_{\mathcal M_{AB}R}]_{++,++}
\right)^{1/2}|\eta\rangle
\right\|^2,
\end{aligned}
\]
and hence
\([S_{\mathcal M_{AB}R}]_{++,++}|\eta\rangle=\alpha^2|\eta\rangle\).
The same calculation with
\(I_R-\beta^{-2}[S_{\mathcal M_{AB}R}]_{--,--}\) gives the \(--\) identity.
We have therefore shown
\begin{equation}\label{eq:forced_diagonal_block_actions}
\begin{aligned}
[S_{\mathcal M_{AB}R}]_{+-,+-}|\eta\rangle
&=[S_{\mathcal M_{AB}R}]_{-+,-+}|\eta\rangle=0,\\
[S_{\mathcal M_{AB}R}]_{++,++}|\eta\rangle
&=\alpha^2|\eta\rangle,
\qquad
[S_{\mathcal M_{AB}R}]_{--,--}|\eta\rangle=\beta^2|\eta\rangle.
\end{aligned}
\end{equation}
The multiplicity vector
\[
\beta|++\rangle-\alpha|--\rangle
\]
is orthogonal to \(|v_d\rangle\).  Hence
\[
\begin{aligned}
&\langle\eta|
(\beta\langle++|-\alpha\langle--|)
S_{\mathcal M_{AB}R}
(\beta|++\rangle-\alpha|--\rangle)
|\eta\rangle\\
&=
(\beta\langle++|-\alpha\langle--|)
|v_d\rangle\langle v_d|
(\beta|++\rangle-\alpha|--\rangle)\\
&=
|\alpha\beta-\beta\alpha|^2
=0.
\end{aligned}
\]
Since \(S_{\mathcal M_{AB}R}\succeq0\), the vector
\((\beta|++\rangle-\alpha|--\rangle)\otimes|\eta\rangle\)
lies in its kernel, so
\[
S_{\mathcal M_{AB}R}
\bigl((\beta|++\rangle-\alpha|--\rangle)\otimes|\eta\rangle\bigr)
=0.
\]
Projecting this equation onto the \(++\) and \(--\) coordinates gives
\begin{equation}\label{eq:kernel_projected_block_relations}
\begin{aligned}
\beta[S_{\mathcal M_{AB}R}]_{++,++}|\eta\rangle
&-\alpha[S_{\mathcal M_{AB}R}]_{++,--}|\eta\rangle=0,\\
\beta[S_{\mathcal M_{AB}R}]_{--,++}|\eta\rangle
&-\alpha[S_{\mathcal M_{AB}R}]_{--,--}|\eta\rangle=0.
\end{aligned}
\end{equation}
Because \(S_{\mathcal M_{AB}R}\) is Hermitian,
\begin{equation}\label{eq:offdiagonal_block_adjoint}
\begin{aligned}
\bigl([S_{\mathcal M_{AB}R}]_{++,--}\bigr)^\dagger
&=
\left[(\langle++|\otimes I_R)S_{\mathcal M_{AB}R}
(|--\rangle\otimes I_R)\right]^\dagger\\
&=(\langle--|\otimes I_R)S_{\mathcal M_{AB}R}
(|++\rangle\otimes I_R)\\
&=[S_{\mathcal M_{AB}R}]_{--,++}.
\end{aligned}
\end{equation}
Substituting \eqref{eq:forced_diagonal_block_actions} into
\eqref{eq:kernel_projected_block_relations} gives
\begin{equation}\label{eq:forced_offdiagonal_block_actions}
[S_{\mathcal M_{AB}R}]_{++,--}|\eta\rangle
=[S_{\mathcal M_{AB}R}]_{--,++}|\eta\rangle
=\alpha\beta|\eta\rangle.
\end{equation}
Expanding \(I_R=|\eta\rangle\langle\eta|+\Pi_\eta^\perp\) on both sides of
the off-diagonal block and its adjoint gives
\begin{equation}\label{eq:forced_offdiagonal_block_decompositions}
\begin{aligned}
[S_{\mathcal M_{AB}R}]_{++,--}
&=\alpha\beta|\eta\rangle\langle\eta|
+\Pi_\eta^\perp[S_{\mathcal M_{AB}R}]_{++,--}\Pi_\eta^\perp,\\
[S_{\mathcal M_{AB}R}]_{--,++}
&=\alpha\beta|\eta\rangle\langle\eta|
+\Pi_\eta^\perp[S_{\mathcal M_{AB}R}]_{--,++}\Pi_\eta^\perp.
\end{aligned}
\end{equation}
Specialize the residual operator \(K_\perp\) in
\eqref{eq:branch_contraction_residual_decomposition} to \(K=S\), and write it
as \(S_\perp\).  In particular,
\[
S_\perp=(I-\mathcal G\mathcal G^\dagger)S
(I-\mathcal G\mathcal G^\dagger)\succeq0.
\]
Define the cross-sector success budget by
\begin{equation}\label{eq:cross_sector_budget_definition}
Q:=
\left[\operatorname{Tr}_{A'B'}S\right]_{+-}
+
\left[\operatorname{Tr}_{A'B'}S\right]_{-+}.
\end{equation}
The two corresponding identities in
\eqref{eq:branch_sector_trace_decomposition} give
\begin{equation}\label{eq:cross_sector_budget_decomposition}
\begin{aligned}
Q
={}&
\frac{
[S_{\mathcal M_{AB}R}]_{+-,+-}
+
[S_{\mathcal M_{AB}R}]_{-+,-+}
}{\alpha\beta}\\
&+
\left[\operatorname{Tr}_{A'B'}S_\perp\right]_{+-}
+
\left[\operatorname{Tr}_{A'B'}S_\perp\right]_{-+}\\
\succeq{}&
\frac{
[S_{\mathcal M_{AB}R}]_{+-,+-}
+
[S_{\mathcal M_{AB}R}]_{-+,-+}
}{\alpha\beta}.
\end{aligned}
\end{equation}
The last inequality follows from
\eqref{eq:residual_output_trace_sector_positive}.  Positivity of
\(S_{\mathcal M_{AB}R}\), which follows from
\eqref{eq:success_contraction_block_positive}, gives the first inequality
below through its \(++/--\) principal block; the \(++\) and \(--\) bounds in
\eqref{eq:success_output_trace_sector_bounds} give the second:
\begin{equation}\label{eq:offdiagonal_block_matrix_element_bound}
\left|
\left\langle x\left|[S_{\mathcal M_{AB}R}]_{++,--}\right|y\right\rangle
\right|^2
\le
\left\langle x\left|[S_{\mathcal M_{AB}R}]_{++,++}\right|x\right\rangle
\left\langle y\left|[S_{\mathcal M_{AB}R}]_{--,--}\right|y\right\rangle
\le
\alpha^2\beta^2\|x\|^2\|y\|^2,
\end{equation}
so \(\|[S_{\mathcal M_{AB}R}]_{++,--}\|\le\alpha\beta\).  Define
\begin{equation}\label{eq:auxiliary_effect_definition}
N:=
\frac{
Q-\frac1{\alpha\beta}\Pi_\eta^\perp
\left(
[S_{\mathcal M_{AB}R}]_{++,--}
+
[S_{\mathcal M_{AB}R}]_{--,++}
\right)\Pi_\eta^\perp
+2I_R
}{6}.
\end{equation}
Indeed,
\begin{equation}\label{eq:auxiliary_operator_bounds}
-2I_R\preceq
\frac1{\alpha\beta}\Pi_\eta^\perp
\left(
[S_{\mathcal M_{AB}R}]_{++,--}
+
[S_{\mathcal M_{AB}R}]_{--,++}
\right)\Pi_\eta^\perp
\preceq2I_R,
\qquad
0\preceq Q\preceq2I_R,
\end{equation}
where the bound on \(Q\) follows by adding the \(+-\) and \(-+\) inequalities
in \eqref{eq:success_output_trace_sector_bounds}.  Hence
\begin{equation}\label{eq:auxiliary_effect_bounds}
0\preceq N\preceq I_R.
\end{equation}

The two vanishing assumptions in
\eqref{eq:gd_zero_input_sector_success} and
\(0\preceq[\operatorname{Tr}_{A'B'}S]_\tau\preceq I_R\) give
\begin{equation}\label{eq:mixed_output_sector_annihilation}
\left[\operatorname{Tr}_{A'B'}S\right]_{+-}|\eta\rangle
=
\left[\operatorname{Tr}_{A'B'}S\right]_{-+}|\eta\rangle
=
Q|\eta\rangle=0.
\end{equation}
Since \(Q\) is Hermitian, \(Q|\eta\rangle=0\) also gives
\(\langle\eta|Q=0\).  Expanding
\(I_R=|\eta\rangle\langle\eta|+\Pi_\eta^\perp\) on both sides of \(Q\)
therefore yields
\begin{equation}\label{eq:cross_sector_budget_support}
Q=\Pi_\eta^\perp Q\Pi_\eta^\perp.
\end{equation}

Table~\ref{tab:ppt_test_relations} summarizes the previously
established relations used in the two PPT tests.
\begin{table}[htbp]
\centering
\caption{Previously established relations used in the two PPT tests.}
\label{tab:ppt_test_relations}
\footnotesize
\setlength{\jot}{4pt}
\renewcommand{\arraystretch}{1.3}
\begin{tabularx}{\linewidth}{
@{}>{\raggedright\arraybackslash}X
>{\centering\arraybackslash}p{0.11\linewidth}@{}}
\toprule
Formulae & Source\\
\midrule
\(
S_\perp=(I-\mathcal G\mathcal G^\dagger)S
(I-\mathcal G\mathcal G^\dagger)\succeq0
\)
& \((\ref*{eq:branch_contraction_residual_decomposition})\)
\\
\addlinespace[4pt]
\(\begin{aligned}
\left[\operatorname{Tr}_{A'B'}S\right]_{++}
&=\alpha^{-2}[S_{\mathcal M_{AB}R}]_{++,++}\\
&\quad+\left[\operatorname{Tr}_{A'B'}S_\perp\right]_{++},\\
\left[\operatorname{Tr}_{A'B'}S\right]_{+-}
&=(\alpha\beta)^{-1}[S_{\mathcal M_{AB}R}]_{+-,+-}\\
&\quad+\left[\operatorname{Tr}_{A'B'}S_\perp\right]_{+-},\\
\left[\operatorname{Tr}_{A'B'}S\right]_{-+}
&=(\alpha\beta)^{-1}[S_{\mathcal M_{AB}R}]_{-+,-+}\\
&\quad+\left[\operatorname{Tr}_{A'B'}S_\perp\right]_{-+},\\
\left[\operatorname{Tr}_{A'B'}S\right]_{--}
&=\beta^{-2}[S_{\mathcal M_{AB}R}]_{--,--}\\
&\quad+\left[\operatorname{Tr}_{A'B'}S_\perp\right]_{--}
\end{aligned}\)
& \((\ref*{eq:branch_sector_trace_decomposition})\)
\\
\addlinespace[4pt]
\(
\left[\operatorname{Tr}_{A'B'}S_\perp\right]_\tau
\succeq0,\qquad \tau\in\mathcal Z
\)
& \((\ref*{eq:residual_output_trace_sector_positive})\)
\\
\midrule
\(\begin{aligned}
0&\preceq \alpha^{-2}[S_{\mathcal M_{AB}R}]_{++,++}
\preceq\left[\operatorname{Tr}_{A'B'}S\right]_{++}\preceq I_R,\\
0&\preceq (\alpha\beta)^{-1}[S_{\mathcal M_{AB}R}]_{+-,+-}
\preceq\left[\operatorname{Tr}_{A'B'}S\right]_{+-}\preceq I_R,\\
0&\preceq (\alpha\beta)^{-1}[S_{\mathcal M_{AB}R}]_{-+,-+}
\preceq\left[\operatorname{Tr}_{A'B'}S\right]_{-+}\preceq I_R,\\
0&\preceq \beta^{-2}[S_{\mathcal M_{AB}R}]_{--,--}
\preceq\left[\operatorname{Tr}_{A'B'}S\right]_{--}\preceq I_R
\end{aligned}\)
& \((\ref*{eq:success_output_trace_sector_bounds})\)
\\
\addlinespace[4pt]
\(\begin{aligned}
[S_{\mathcal M_{AB}R}]_{+-,+-}|\eta\rangle
&=[S_{\mathcal M_{AB}R}]_{-+,-+}|\eta\rangle=0,\\
[S_{\mathcal M_{AB}R}]_{++,++}|\eta\rangle
&=\alpha^2|\eta\rangle,\qquad
[S_{\mathcal M_{AB}R}]_{--,--}|\eta\rangle=\beta^2|\eta\rangle
\end{aligned}\)
&
\((\ref*{eq:forced_diagonal_block_actions})\)
\\
\addlinespace[4pt]
\midrule
\(
\bigl([S_{\mathcal M_{AB}R}]_{++,--}\bigr)^\dagger
=[S_{\mathcal M_{AB}R}]_{--,++}
\)
& \((\ref*{eq:offdiagonal_block_adjoint})\)
\\
\addlinespace[4pt]
\(\begin{aligned}
[S_{\mathcal M_{AB}R}]_{++,--}|\eta\rangle
&=\alpha\beta|\eta\rangle,\\
[S_{\mathcal M_{AB}R}]_{--,++}|\eta\rangle
&=\alpha\beta|\eta\rangle
\end{aligned}\)
& \((\ref*{eq:forced_offdiagonal_block_actions})\)
\\
\addlinespace[4pt]
\(\begin{aligned}
\left|\left\langle x\left|[S_{\mathcal M_{AB}R}]_{++,--}
\right|y\right\rangle\right|^2
&\le
\left\langle x\left|[S_{\mathcal M_{AB}R}]_{++,++}\right|x\right\rangle
\left\langle y\left|[S_{\mathcal M_{AB}R}]_{--,--}\right|y\right\rangle\\
&\le\alpha^2\beta^2\|x\|^2\|y\|^2,
\qquad
\|[S_{\mathcal M_{AB}R}]_{++,--}\|\le\alpha\beta
\end{aligned}\)
& \((\ref*{eq:offdiagonal_block_matrix_element_bound})\)
\\
\addlinespace[4pt]
\(\begin{aligned}
[S_{\mathcal M_{AB}R}]_{++,--}
&=\alpha\beta|\eta\rangle\langle\eta|\\
&\quad+\Pi_\eta^\perp
[S_{\mathcal M_{AB}R}]_{++,--}
\Pi_\eta^\perp,\\
[S_{\mathcal M_{AB}R}]_{--,++}
&=\alpha\beta|\eta\rangle\langle\eta|\\
&\quad+\Pi_\eta^\perp
[S_{\mathcal M_{AB}R}]_{--,++}
\Pi_\eta^\perp
\end{aligned}\)
& \((\ref*{eq:forced_offdiagonal_block_decompositions})\)
\\
\addlinespace[4pt]
\midrule
\(
\bigl(S_{\mathcal M_{AB}R}\bigr)^{T_{\mathcal M_B}T_{B_{\mathrm r}}}
\succeq0
\)
& \((\ref*{eq:success_contraction_block_ppt_positive})\)
\\
\addlinespace[4pt]
\(
\left(\left[\operatorname{Tr}_{A'B'}S\right]_\tau\right)^{T_{B_{\mathrm r}}}
\succeq0,\qquad \tau\in\mathcal Z
\)
& \((\ref*{eq:mixed_schur_trace_ppt_order})\)
\\
\addlinespace[4pt]
\(
\left[\operatorname{Tr}_{A'B'}S\right]_{+-}|\eta\rangle
=\left[\operatorname{Tr}_{A'B'}S\right]_{-+}|\eta\rangle
=Q|\eta\rangle=0
\)
& \((\ref*{eq:mixed_output_sector_annihilation})\)
\\
\addlinespace[4pt]
\(
\left[\operatorname{Tr}_{A'B'}F\right]_\tau
=I_R-\left[\operatorname{Tr}_{A'B'}S\right]_\tau\succeq0,
\qquad \tau\in\mathcal Z
\)
& \((\ref*{eq:success_failure_sector_complement})\)
\\
\midrule
\(
Q=\left[\operatorname{Tr}_{A'B'}S\right]_{+-}
+\left[\operatorname{Tr}_{A'B'}S\right]_{-+}
\)
& \((\ref*{eq:cross_sector_budget_definition})\)
\\
\addlinespace[4pt]
\(\begin{aligned}
Q={}&\frac{[S_{\mathcal M_{AB}R}]_{+-,+-}
+[S_{\mathcal M_{AB}R}]_{-+,-+}}{\alpha\beta}\\
&+\left[\operatorname{Tr}_{A'B'}S_\perp\right]_{+-}
+\left[\operatorname{Tr}_{A'B'}S_\perp\right]_{-+}\\
\succeq{}&\frac{[S_{\mathcal M_{AB}R}]_{+-,+-}
+[S_{\mathcal M_{AB}R}]_{-+,-+}}{\alpha\beta}
\end{aligned}\)
& \((\ref*{eq:cross_sector_budget_decomposition})\)
\\
\addlinespace[4pt]
\(\begin{aligned}
-2I_R&\preceq
\frac1{\alpha\beta}\Pi_\eta^\perp
\left(
[S_{\mathcal M_{AB}R}]_{++,--}
+[S_{\mathcal M_{AB}R}]_{--,++}
\right)
\Pi_\eta^\perp\preceq2I_R,\\
0&\preceq Q\preceq2I_R
\end{aligned}\)
& \((\ref*{eq:auxiliary_operator_bounds})\)
\\
\addlinespace[4pt]
\(
Q=\Pi_\eta^\perp Q\Pi_\eta^\perp
\)
& \((\ref*{eq:cross_sector_budget_support})\)
\\
\midrule
\(\displaystyle
N:=\frac16\left[Q-
\frac1{\alpha\beta}\Pi_\eta^\perp
\left(
[S_{\mathcal M_{AB}R}]_{++,--}
+[S_{\mathcal M_{AB}R}]_{--,++}
\right)
\Pi_\eta^\perp+2I_R\right]
\)
& \((\ref*{eq:auxiliary_effect_definition})\)
\\
\addlinespace[4pt]
\(
0\preceq N\preceq I_R
\)
& \((\ref*{eq:auxiliary_effect_bounds})\)
\\
\bottomrule
\end{tabularx}
\end{table}

\paragraph{Single-vector PPT test.}
We evaluate the required PPT quadratic form directly.  Applying
\(T_{\mathcal M_B}\) to the multiplicity matrix units and
\(T_{B_{\mathrm r}}\) to their resource-valued coefficients gives
\begin{equation}\label{eq:resource_block_partial_transpose_expansion}
\bigl(S_{\mathcal M_{AB}R}\bigr)^{T_{\mathcal M_B}T_{B_{\mathrm r}}}
=
\sum_{\substack{s_A,s_B\in\{+,-\}\\
t_A,t_B\in\{+,-\}}}
|s_A,t_B\rangle\langle t_A,s_B|
\otimes
\bigl([S_{\mathcal M_{AB}R}]_{(s_A,s_B),(t_A,t_B)}\bigr)^{T_{B_{\mathrm r}}}.
\end{equation}
Choose
\[
|a\rangle
:=
\frac{|+-\rangle-|-+\rangle}{\sqrt2}.
\]
For the multiplicity matrix unit in
\eqref{eq:resource_block_partial_transpose_expansion}, direct evaluation gives
\[
2\langle a|s_A,t_B\rangle\langle t_A,s_B|a\rangle
=
\begin{cases}
1,&(s_A,s_B,t_A,t_B)=(+,-,+,-)\text{ or }(-,+,-,+),\\
-1,&(s_A,s_B,t_A,t_B)=(+,+,-,-)\text{ or }(-,-,+,+),\\
0,&\text{otherwise}.
\end{cases}
\]
Substituting these four nonzero terms into
\eqref{eq:resource_block_partial_transpose_expansion} and using
\(\langle00|X^{T_{B_{\mathrm r}}}|00\rangle=\langle00|X|00\rangle\) gives
\[
\begin{aligned}
&\frac{2}{\alpha\beta}
\left\langle a,00\left|
\bigl(S_{\mathcal M_{AB}R}\bigr)^{T_{\mathcal M_B}T_{B_{\mathrm r}}}
\right|a,00\right\rangle\\
&\quad=
\frac1{\alpha\beta}
\left\langle00\left|
 [S_{\mathcal M_{AB}R}]_{+-,+-}
\right|00\right\rangle\\
&\qquad+
\frac1{\alpha\beta}
\left\langle00\left|[S_{\mathcal M_{AB}R}]_{-+,-+}\right|00\right\rangle\\
&\qquad-
\frac1{\alpha\beta}
\left\langle00\left|[S_{\mathcal M_{AB}R}]_{++,--}\right|00\right\rangle\\
&\qquad-
\frac1{\alpha\beta}
\left\langle00\left|[S_{\mathcal M_{AB}R}]_{--,++}\right|00\right\rangle.
\end{aligned}
\]
Moreover,
\[
|00\rangle
=
\sqrt{p_0}|\eta\rangle+\Pi_\eta^\perp|00\rangle.
\]
Using \eqref{eq:forced_diagonal_block_actions} and
\eqref{eq:forced_offdiagonal_block_decompositions}, we obtain
\[
\begin{aligned}
&\frac{2}{\alpha\beta}
\left\langle a,00\left|
\bigl(S_{\mathcal M_{AB}R}\bigr)^{T_{\mathcal M_B}T_{B_{\mathrm r}}}
\right|a,00\right\rangle\\
&\quad=
\frac1{\alpha\beta}
\left\langle00\left|
\Pi_\eta^\perp[S_{\mathcal M_{AB}R}]_{+-,+-}\Pi_\eta^\perp
\right|00\right\rangle\\
&\qquad+
\frac1{\alpha\beta}
\left\langle00\left|
\Pi_\eta^\perp[S_{\mathcal M_{AB}R}]_{-+,-+}\Pi_\eta^\perp
\right|00\right\rangle\\
&\qquad-
\frac1{\alpha\beta}
\left\langle00\left|
\Pi_\eta^\perp[S_{\mathcal M_{AB}R}]_{++,--}\Pi_\eta^\perp
\right|00\right\rangle\\
&\qquad-
\frac1{\alpha\beta}
\left\langle00\left|
\Pi_\eta^\perp[S_{\mathcal M_{AB}R}]_{--,++}\Pi_\eta^\perp
\right|00\right\rangle\\
&\qquad-2p_0.
\end{aligned}
\]
On the other hand, sandwiching the definition of \(N\) between
\(\langle00|\Pi_\eta^\perp\) and \(\Pi_\eta^\perp|00\rangle\) gives
\begin{equation}\label{eq:N_00_expansion}
\begin{aligned}
&6\langle00|\Pi_\eta^\perp N\Pi_\eta^\perp|00\rangle-2\\
&\quad=
\langle00|\Pi_\eta^\perp Q\Pi_\eta^\perp|00\rangle\\
&\qquad-
\frac1{\alpha\beta}
\left\langle00\left|
\Pi_\eta^\perp
\left(
[S_{\mathcal M_{AB}R}]_{++,--}
+
[S_{\mathcal M_{AB}R}]_{--,++}
\right)
\Pi_\eta^\perp
\right|00\right\rangle\\
&\qquad+
2\langle00|\Pi_\eta^\perp|00\rangle-2\\
&\quad=
\langle00|\Pi_\eta^\perp Q\Pi_\eta^\perp|00\rangle\\
&\qquad-
\frac1{\alpha\beta}
\left\langle00\left|
\Pi_\eta^\perp
\left(
[S_{\mathcal M_{AB}R}]_{++,--}
+
[S_{\mathcal M_{AB}R}]_{--,++}
\right)
\Pi_\eta^\perp
\right|00\right\rangle
-2p_0.
\end{aligned}
\end{equation}
The last equality uses
\(\langle00|\Pi_\eta^\perp|00\rangle=1-p_0\), and hence
\(2(1-p_0)-2=-2p_0\).  Comparing \eqref{eq:N_00_expansion} with the
preceding single-vector PPT quadratic form and using
\eqref{eq:cross_sector_budget_decomposition} gives
\begin{equation}\label{eq:N_00_ppt_lower_bound}
\begin{aligned}
6\langle00|\Pi_\eta^\perp N\Pi_\eta^\perp|00\rangle-2
&\ge
\frac{2}{\alpha\beta}
\left\langle a,00\left|
\bigl(S_{\mathcal M_{AB}R}\bigr)^{T_{\mathcal M_B}T_{B_{\mathrm r}}}
\right|a,00\right\rangle\\
&\ge0.
\end{aligned}
\end{equation}
The first inequality uses the contraction-sector lower bound on \(Q\), after
compression by \(\Pi_\eta^\perp\); the second uses
\(\bigl(S_{\mathcal M_{AB}R}\bigr)^{
T_{\mathcal M_B}T_{B_{\mathrm r}}}\succeq0\).
Since \(N\preceq I_R\), \eqref{eq:N_00_ppt_lower_bound} gives
\begin{equation}\label{eq:N_00_interval}
\frac13
\le
\langle00|\Pi_\eta^\perp N\Pi_\eta^\perp|00\rangle
\le
\langle00|\Pi_\eta^\perp|00\rangle
=1-p_0.
\end{equation}

If \(r=1\), then \(p_0=1\), contradicting
\eqref{eq:N_00_interval}.  Hence the hypotheses force \(r\ge2\), so the basis
vector \(\ket1\) and the Bell states used below are well-defined.

\paragraph{Pairwise PPT test.}
Introduce the standard Bell states
\[
\ket{\Psi^\pm}_R
:=
\frac{\ket{01}_R\pm\ket{10}_R}{\sqrt2},
\qquad
\ket{\phi_\pm}_{\mathcal M_{AB}R}
:=
\frac{\ket{+-}_{\mathcal M_{AB}}\pm
\ket{-+}_{\mathcal M_{AB}}}{\sqrt2}
\otimes\ket{\Psi^\pm}_R.
\]
The trace-complement relation needed below is
\eqref{eq:success_failure_sector_complement}.  The required resource
partial-transpose matrix elements are, for every \(X\in\mathcal B(R)\),
\begin{equation}\label{eq:resource_partial_transpose_matrix_elements}
\begin{aligned}
\langle01|X^{T_{B_{\mathrm r}}}|01\rangle&=\langle01|X|01\rangle,
&\langle10|X^{T_{B_{\mathrm r}}}|10\rangle&=\langle10|X|10\rangle,\\
\langle01|X^{T_{B_{\mathrm r}}}|10\rangle&=\langle00|X|11\rangle,
&\langle10|X^{T_{B_{\mathrm r}}}|01\rangle&=\langle11|X|00\rangle.
\end{aligned}
\end{equation}
Moreover, writing \(\sigma=\pm1\), with
\(\ket{\Psi^\sigma}_R=\ket{\Psi^+}_R\) for \(\sigma=+1\) and
\(\ket{\Psi^\sigma}_R=\ket{\Psi^-}_R\) for \(\sigma=-1\), direct
expansion gives
\begin{equation}\label{eq:bell_quadratic_form_expansion}
\begin{aligned}
\langle\Psi^\sigma|X|\Psi^\sigma\rangle
={}&\frac12\langle01|X|01\rangle
+\frac12\langle10|X|10\rangle\\
&+\frac{\sigma}{2}\langle01|X|10\rangle
+\frac{\sigma}{2}\langle10|X|01\rangle.
\end{aligned}
\end{equation}
Combining this expansion with
\eqref{eq:resource_partial_transpose_matrix_elements} gives
\begin{equation}\label{eq:bell_resource_partial_transpose_sign_sums}
\begin{aligned}
\sum_{\sigma=\pm1}
\langle\Psi^\sigma|X^{T_{B_{\mathrm r}}}|\Psi^\sigma\rangle
&=\langle01|X|01\rangle+\langle10|X|10\rangle,\\
\sum_{\sigma=\pm1}\sigma
\langle\Psi^\sigma|X^{T_{B_{\mathrm r}}}|\Psi^\sigma\rangle
&=\langle00|X|11\rangle+\langle11|X|00\rangle.
\end{aligned}
\end{equation}
By \eqref{eq:forced_offdiagonal_block_decompositions}, the cross-block sum
appearing below satisfies
\begin{equation}\label{eq:pairwise_cross_block_coherence_decomposition}
\begin{aligned}
&\left\langle00\left|
[S_{\mathcal M_{AB}R}]_{++,--}
+[S_{\mathcal M_{AB}R}]_{--,++}
\right|11\right\rangle\\
&\quad=-2\alpha\beta\sqrt{p_0p_1}
+\left\langle00\left|\Pi_\eta^\perp
\left(
[S_{\mathcal M_{AB}R}]_{++,--}
+[S_{\mathcal M_{AB}R}]_{--,++}
\right)
\Pi_\eta^\perp\right|11\right\rangle.
\end{aligned}
\end{equation}
Substituting \(\ket{\phi_\sigma}\) into
\eqref{eq:resource_block_partial_transpose_expansion} gives, for each
\(\sigma\in\{+1,-1\}\),
\begin{equation}\label{eq:pairwise_ppt_single_sign_quadratic_form}
\begin{aligned}
2\langle\phi_\sigma|
\bigl(S_{\mathcal M_{AB}R}\bigr)^{T_{\mathcal M_B}T_{B_{\mathrm r}}}
|\phi_\sigma\rangle
={}&
\langle\Psi^\sigma|
\left(
[S_{\mathcal M_{AB}R}]_{+-,+-}
+
[S_{\mathcal M_{AB}R}]_{-+,-+}
\right)^{T_{B_{\mathrm r}}}
|\Psi^\sigma\rangle\\
&+\sigma\langle\Psi^\sigma|
\left(
[S_{\mathcal M_{AB}R}]_{++,--}
+
[S_{\mathcal M_{AB}R}]_{--,++}
\right)^{T_{B_{\mathrm r}}}
|\Psi^\sigma\rangle.
\end{aligned}
\end{equation}
Sum \eqref{eq:pairwise_ppt_single_sign_quadratic_form} over \(\sigma\), and
apply the first and second identities in
\eqref{eq:bell_resource_partial_transpose_sign_sums} to the
\((+-,+-)/(-+,-+)\) and \((++,--)/(--,++)\) block sums, respectively.
The latter block sum is Hermitian by
\eqref{eq:offdiagonal_block_adjoint}; hence its two cross matrix elements
combine into twice the real part.  Substituting
\eqref{eq:pairwise_cross_block_coherence_decomposition} then gives
\begin{equation}\label{eq:pairwise_success_sum_expansion}
\begin{aligned}
&\frac{2}{\alpha\beta}\sum_{\sigma=\pm1}
\langle\phi_\sigma|
\bigl(S_{\mathcal M_{AB}R}\bigr)^{T_{\mathcal M_B}T_{B_{\mathrm r}}}
|\phi_\sigma\rangle\\
&\quad=
\frac1{\alpha\beta}
\left\langle01\left|
[S_{\mathcal M_{AB}R}]_{+-,+-}
+
[S_{\mathcal M_{AB}R}]_{-+,-+}
\right|01\right\rangle\\
&\qquad+
\frac1{\alpha\beta}
\left\langle10\left|
[S_{\mathcal M_{AB}R}]_{+-,+-}
+
[S_{\mathcal M_{AB}R}]_{-+,-+}
\right|10\right\rangle\\
&\qquad-4\sqrt{p_0p_1}
+\frac{2}{\alpha\beta}\operatorname{Re}
\left\langle00\left|\Pi_\eta^\perp
\left(
[S_{\mathcal M_{AB}R}]_{++,--}
+
[S_{\mathcal M_{AB}R}]_{--,++}
\right)
\Pi_\eta^\perp\right|11\right\rangle.
\end{aligned}
\end{equation}
For each \(\tau\in\{+-, -+\}\), take \(T_{B_{\mathrm r}}\) of
\eqref{eq:success_failure_sector_complement} and sandwich the result between
\(\ket{\Psi^+}_R\).  The identity term contributes \(2\), because the
quadratic form is multiplied by \(2\) and \(\|\ket{\Psi^+}_R\|=1\).  Applying
\eqref{eq:bell_quadratic_form_expansion} with \(\sigma=+1\), together with
\eqref{eq:resource_partial_transpose_matrix_elements} and the Hermiticity of
\(\left[\operatorname{Tr}_{A'B'}S\right]_\tau\), gives
\begin{equation}\label{eq:failure_sector_bell_quadratic_form}
\begin{aligned}
2\left\langle\Psi^+\left|
\left(
\left[\operatorname{Tr}_{A'B'}F\right]_\tau
\right)^{T_{B_{\mathrm r}}}
\right|\Psi^+\right\rangle
={}&2-
\left\langle01\left|
\left[\operatorname{Tr}_{A'B'}S\right]_\tau
\right|01\right\rangle\\
&-
\left\langle10\left|
\left[\operatorname{Tr}_{A'B'}S\right]_\tau
\right|10\right\rangle\\
&-2\operatorname{Re}
\left\langle00\left|
\left[\operatorname{Tr}_{A'B'}S\right]_\tau
\right|11\right\rangle.
\end{aligned}
\end{equation}
Summing \eqref{eq:failure_sector_bell_quadratic_form} over \(\tau\) and using
\eqref{eq:cross_sector_budget_support} gives
\begin{equation}\label{eq:pairwise_failure_sum_expansion}
\begin{aligned}
&2\sum_{\tau\in\{+-, -+\}}
\left\langle\Psi^+\left|
\left(
\left[\operatorname{Tr}_{A'B'}F\right]_\tau
\right)^{T_{B_{\mathrm r}}}
\right|\Psi^+\right\rangle\\
&\quad=4-\langle01|Q|01\rangle-\langle10|Q|10\rangle
-2\operatorname{Re}
\langle00|\Pi_\eta^\perp Q\Pi_\eta^\perp|11\rangle.
\end{aligned}
\end{equation}
The three quantities forming the nonnegative combination below are the success
quadratic-form sum, the residual-sector quadratic-form sum, and the failure
quadratic-form sum.  They are nonnegative by
\eqref{eq:success_contraction_block_ppt_positive},
\eqref{eq:residual_output_trace_sector_positive}, and
\eqref{eq:mixed_schur_trace_ppt_order} applied to \(K=F\), respectively.  Add
the expansions \eqref{eq:pairwise_success_sum_expansion} and
\eqref{eq:pairwise_failure_sum_expansion}, together with the residual-sector
sum displayed below.  For each \(x\in\{01,10\}\),
\eqref{eq:cross_sector_budget_decomposition} combines the contraction-sector
diagonal term from \eqref{eq:pairwise_success_sum_expansion} with this residual
term into \(\langle x|Q|x\rangle\).  It therefore cancels the corresponding
negative term in \eqref{eq:pairwise_failure_sum_expansion}.  Thus
\begin{equation}\label{eq:pairwise_nonnegative_combination}
\begin{aligned}
0
&\le
\frac{2}{\alpha\beta}\sum_{\sigma=\pm1}
\langle\phi_\sigma|
\bigl(S_{\mathcal M_{AB}R}\bigr)^{T_{\mathcal M_B}T_{B_{\mathrm r}}}
|\phi_\sigma\rangle\\
&\qquad+
\sum_{x\in\{01,10\}}
\left\langle x\left|
\left[\operatorname{Tr}_{A'B'}S_\perp\right]_{+-}
+
\left[\operatorname{Tr}_{A'B'}S_\perp\right]_{-+}
\right|x\right\rangle\\
&\qquad
+2\sum_{\tau\in\{+-, -+\}}
\left\langle\Psi^+\left|
\left(
\left[\operatorname{Tr}_{A'B'}F\right]_\tau
\right)^{T_{B_{\mathrm r}}}
\right|\Psi^+\right\rangle\\
&=4-4\sqrt{p_0p_1}
+2\operatorname{Re}
\langle00|\Pi_\eta^\perp
\Bigl(
\frac1{\alpha\beta}\Pi_\eta^\perp
\left(
[S_{\mathcal M_{AB}R}]_{++,--}
+
[S_{\mathcal M_{AB}R}]_{--,++}
\right)
\Pi_\eta^\perp-Q
\Bigr)
\Pi_\eta^\perp|11\rangle.
\end{aligned}
\end{equation}
Rearranging \eqref{eq:auxiliary_effect_definition} gives
\begin{equation}\label{eq:pairwise_auxiliary_effect_rearrangement}
\frac1{\alpha\beta}\Pi_\eta^\perp
\left(
[S_{\mathcal M_{AB}R}]_{++,--}
+[S_{\mathcal M_{AB}R}]_{--,++}
\right)
\Pi_\eta^\perp-Q
=-(6N-2I_R).
\end{equation}
Moreover, the chosen phase convention gives
\begin{equation}\label{eq:pairwise_projected_basis_overlap}
\langle00|\Pi_\eta^\perp|11\rangle=\sqrt{p_0p_1}.
\end{equation}
Substituting \eqref{eq:pairwise_auxiliary_effect_rearrangement} into
\eqref{eq:pairwise_nonnegative_combination} and then using
\eqref{eq:pairwise_projected_basis_overlap} yields
\[
\begin{aligned}
0
&\le4-4\sqrt{p_0p_1}
-2\operatorname{Re}
\langle00|\Pi_\eta^\perp(6N-2I_R)
\Pi_\eta^\perp|11\rangle\\
&=4-12\operatorname{Re}
\langle00|\Pi_\eta^\perp N\Pi_\eta^\perp|11\rangle.
\end{aligned}
\]
Hence
\begin{equation}\label{eq:N_pairwise_budget}
\operatorname{Re}
\langle00|\Pi_\eta^\perp N\Pi_\eta^\perp|11\rangle
\le\frac13.
\end{equation}
\end{proof}

\section{Effect Constraints and Schmidt-Spectrum Bounds}
\label{app:abstract_effect_schmidt}

This section converts the effect constraints derived above into restrictions on
the Schmidt spectrum of a pure resource state.

Let \(R:=A_{\mathrm r}B_{\mathrm r}\) be a bipartite resource space
with fixed orthonormal bases \(\{\ket j_{A_{\mathrm r}}\}\) and
\(\{\ket j_{B_{\mathrm r}}\}\), and abbreviate
\[
\ket{jj}:=\ket j_{A_{\mathrm r}}\otimes\ket j_{B_{\mathrm r}}.
\]
Let \(|\eta\rangle\in R\) be a finite-Schmidt-rank normalized pure state.
After local phase choices, write
\[
|\eta\rangle
=
\sqrt{p_0}|00\rangle-\sum_{j\ge1}\sqrt{p_j}|jj\rangle,
\qquad
p_0\ge p_1\ge p_2\ge\cdots,
\qquad
\sum_{j\ge0}p_j=1.
\]
Set
\[
\Pi_\eta^\perp:=I_R-|\eta\rangle\langle\eta|.
\]
For \(1/2\le x\le2/3\), define
\begin{equation}\label{eq:boundary_spectrum_definition}
z_\star(x):=
\frac{1-x-\sqrt{x(2-3x)}}2.
\end{equation}
On this interval, \(x(2-3x)\ge0\), and
\[
(1-x)^2-x(2-3x)=(2x-1)^2\ge0.
\]
Since \(1-x\ge0\), it follows that \(z_\star(x)\ge0\).  Moreover,
\((1-x-z_\star(x))-z_\star(x)=\sqrt{x(2-3x)}\ge0\), while
\(x-(1-x-z_\star(x))=2x+z_\star(x)-1\ge0\).  Hence
\(\bigl(x,1-x-z_\star(x),z_\star(x)\bigr)\) is a nonincreasing probability
vector.
For \(0\le x\le2/3\), define the comparison spectrum
\begin{equation}\label{eq:comparison_spectrum_definition}
\sigma_\star(x)
:=
\begin{cases}
(1/2,1/2,0,\ldots),&0\le x\le1/2,\\[1mm]
\bigl(x,1-x-z_\star(x),z_\star(x),0,\ldots\bigr),
&1/2<x\le2/3.
\end{cases}
\end{equation}
For two nonincreasing probability vectors \(a=(a_j)\) and \(b=(b_j)\),
padded by zeros to the same length, define
\[
a\prec b
\quad\Longleftrightarrow\quad
\sum_{j=0}^{k}a_j\le\sum_{j=0}^{k}b_j
\ \text{for every }k,
\qquad
\sum_j a_j=\sum_j b_j.
\]
For any probability vector \(a=(a_j)\), define its Shannon and collision
entropies by
\[
H(a):=-\sum_j a_j\log_2 a_j,
\qquad
H_2(a):=-\log_2\!\sum_j a_j^2.
\]
For a pure state \(\ket\eta\) with Schmidt vector \((p_j)_{j\ge0}\), this gives
\(E(\eta)=H((p_j)_{j\ge0})\); we abbreviate
\(H_2(\eta):=H_2((p_j)_{j\ge0})\).
\begin{lemma}[Effect constraints imply Schmidt majorization]
\label{lem:universal_schmidt_spectrum}
Suppose there exists an effect \(N\in\mathcal B(R)\) such that
\begin{equation}\label{eq:abstract_effect_hypotheses}
\begin{aligned}
0&\preceq N\preceq I_R,\\
\langle00|\Pi_\eta^\perp N\Pi_\eta^\perp|00\rangle
&\ge\frac13,\\
\operatorname{Re}
\langle00|\Pi_\eta^\perp N\Pi_\eta^\perp|11\rangle
&\le\frac13.
\end{aligned}
\end{equation}
Then
\begin{equation}\label{eq:p0_two_thirds}
p_0\le\frac23,
\end{equation}
and the ordered Schmidt vector obeys
\begin{equation}\label{eq:boundary_majorization}
(p_j)_{j\ge0}\prec\sigma_\star(p_0).
\end{equation}
If \(p_0>1/2\), then
\begin{equation}\label{eq:quadratic_tail_lower_bound}
\sum_{j\ge2}p_j
\ge z_\star(p_0)
>2\left(p_0-\frac12\right)^2.
\end{equation}
\end{lemma}

\begin{proof}
Directly from the phase convention and the definition of
\(\Pi_\eta^\perp\),
\begin{equation}\label{eq:abstract_projected_gram_data}
\begin{aligned}
\langle00|\Pi_\eta^\perp|00\rangle&=1-p_0,\\
\langle11|\Pi_\eta^\perp|11\rangle&=1-p_1,\\
\langle00|\Pi_\eta^\perp|11\rangle&=\sqrt{p_0p_1}.
\end{aligned}
\end{equation}
The effect assumptions therefore give
\[
\frac13
\le
\langle00|\Pi_\eta^\perp N\Pi_\eta^\perp|00\rangle
\le
\langle00|\Pi_\eta^\perp|00\rangle
=1-p_0.
\]
If \(p_0\le1/2\), sortedness gives
\[
(p_j)_{j\ge0}\prec(1/2,1/2,0,\ldots)=\sigma_\star(p_0).
\]
Assume \(p_0>1/2\).

The component of \(\Pi_\eta^\perp|11\rangle\) orthogonal to
\(\Pi_\eta^\perp|00\rangle\) has squared norm
\begin{equation}\label{eq:explicit_projected_residual_norm}
\begin{aligned}
&\left\|
\Pi_\eta^\perp|11\rangle
-
\frac{\sqrt{p_0p_1}}{1-p_0}\Pi_\eta^\perp|00\rangle
\right\|^2\\
&\qquad=
1-p_1-\frac{p_0p_1}{1-p_0}
=
\frac{\sum_{j\ge2}p_j}{1-p_0}.
\end{aligned}
\end{equation}
Set
\[
\mu
:=
\frac{
\langle00|\Pi_\eta^\perp N\Pi_\eta^\perp|00\rangle
}{1-p_0}.
\]
The effect assumptions and \(p_0>1/2\) give
\[
\frac{1}{3(1-p_0)}\le\mu\le1,
\qquad
\frac{1}{3(1-p_0)}\ge\frac23.
\]
Since \(N^2\preceq N\),
\begin{align}
&\left\|
\left(
N-\mu I_R
\right)
\Pi_\eta^\perp|00\rangle
\right\|^2\notag\\
&\quad=
\langle00|\Pi_\eta^\perp
(N-\mu I_R)^2
\Pi_\eta^\perp|00\rangle\notag\\
&\quad\le
\langle00|\Pi_\eta^\perp
(N-2\mu N+\mu^2 I_R)
\Pi_\eta^\perp|00\rangle\notag\\
&\quad=
\mu(1-\mu)(1-p_0).
\label{eq:explicit_effect_variance_bound}
\end{align}
The last equality uses
\begin{equation}\label{eq:explicit_effect_mean_identity}
\langle00|\Pi_\eta^\perp N\Pi_\eta^\perp|00\rangle
=\mu(1-p_0).
\end{equation}
Equation \eqref{eq:explicit_effect_mean_identity}, together with
\(\langle00|\Pi_\eta^\perp|00\rangle=1-p_0\), gives
\begin{equation}\label{eq:explicit_centered_effect_orthogonality}
\langle00|\Pi_\eta^\perp
(N-\mu I_R)\Pi_\eta^\perp|00\rangle=0.
\end{equation}
The corresponding Gram--Schmidt decomposition is
\begin{equation}\label{eq:explicit_projected_11_decomposition}
\begin{aligned}
\Pi_\eta^\perp|11\rangle
={}&\frac{\sqrt{p_0p_1}}{1-p_0}\Pi_\eta^\perp|00\rangle\\
&+\left(
\Pi_\eta^\perp|11\rangle
-\frac{\sqrt{p_0p_1}}{1-p_0}\Pi_\eta^\perp|00\rangle
\right).
\end{aligned}
\end{equation}
The parenthesized residual is orthogonal to
\(\Pi_\eta^\perp|00\rangle\), and its squared norm is given by
\eqref{eq:explicit_projected_residual_norm}.  Since \(N\) is Hermitian,
\eqref{eq:explicit_centered_effect_orthogonality} and
\eqref{eq:explicit_projected_11_decomposition} give
\[
\begin{aligned}
&\operatorname{Re}
\langle00|\Pi_\eta^\perp N\Pi_\eta^\perp|11\rangle\\
&\quad=
\mu\langle00|\Pi_\eta^\perp|11\rangle\\
&\qquad+
\operatorname{Re}\langle00|\Pi_\eta^\perp
(N-\mu I_R)\Pi_\eta^\perp|11\rangle\\
&\quad=
\mu\sqrt{p_0p_1}\\
&\qquad+
\operatorname{Re}\left\langle
(N-\mu I_R)\Pi_\eta^\perp|00\rangle,\,
\Pi_\eta^\perp|11\rangle
-\frac{\sqrt{p_0p_1}}{1-p_0}
\Pi_\eta^\perp|00\rangle
\right\rangle.
\end{aligned}
\]
Cauchy--Schwarz, followed by
\eqref{eq:explicit_effect_variance_bound} and
\eqref{eq:explicit_projected_residual_norm}, therefore gives
\begin{align}
&\operatorname{Re}
\langle00|\Pi_\eta^\perp N\Pi_\eta^\perp|11\rangle\notag\\
&\quad\ge \mu\sqrt{p_0p_1}\notag\\
&\qquad-
\sqrt{
\mu(1-\mu)
\sum_{j\ge2}p_j
}.
\label{eq:explicit_two_vector_effect_bound}
\end{align}
The right-hand side is increasing as \(\mu\) ranges from
\(1/[3(1-p_0)]\) to \(1\): the linear term increases, while
\(\mu(1-\mu)\) decreases throughout this interval because
\(\mu\ge1/[3(1-p_0)]\ge2/3\).  Combining
\eqref{eq:explicit_two_vector_effect_bound} with the assumed upper bound and
evaluating at the lower endpoint yields
\[
\frac13
\ge
\frac{\sqrt{p_0p_1}}{3(1-p_0)}
-
\frac{
\sqrt{(2-3p_0)\sum_{j\ge2}p_j}
}{3(1-p_0)}.
\]
Since \(p_1=1-p_0-\sum_{j\ge2}p_j\), multiplication by
\(3(1-p_0)>0\) gives
\begin{equation}\label{eq:sharp_tail_scalar_constraint}
\sqrt{
 p_0\left(1-p_0-\sum_{j\ge2}p_j\right)
}
\le
1-p_0+
\sqrt{(2-3p_0)\sum_{j\ge2}p_j}.
\end{equation}
Both sides are nonnegative.  Squaring and cancelling \(2(1-p_0)>0\) gives
\begin{align}
0
&\le
\sum_{j\ge2}p_j
+
\sqrt{(2-3p_0)\sum_{j\ge2}p_j}
-
\left(p_0-\frac12\right)\notag\\
&=
\left(
\sqrt{\sum_{j\ge2}p_j}
-
\frac{\sqrt{p_0}-\sqrt{2-3p_0}}2
\right)
\left(
\sqrt{\sum_{j\ge2}p_j}
+
\frac{\sqrt{p_0}+\sqrt{2-3p_0}}2
\right).
\label{eq:explicit_tail_factorization}
\end{align}
The second factor is positive, so
\[
\sum_{j\ge2}p_j
\ge
\left(
\frac{\sqrt{p_0}-\sqrt{2-3p_0}}2
\right)^2
=
z_\star(p_0).
\]
The boundary value satisfies
\begin{equation}\label{eq:zstar_boundary}
\begin{aligned}
z_\star(p_0)\bigl(1-p_0-z_\star(p_0)\bigr)
&=
\left(\frac{1-p_0-\sqrt{p_0(2-3p_0)}}2\right)
\left(\frac{1-p_0+\sqrt{p_0(2-3p_0)}}2\right)\\
&=
\frac{(1-p_0)^2-p_0(2-3p_0)}4\\
&=
\left(p_0-\frac12\right)^2.
\end{aligned}
\end{equation}
For \(p_0>1/2\), one has
\(1-p_0-z_\star(p_0)<1/2\); hence
\[
z_\star(p_0)
=
\frac{(p_0-1/2)^2}{1-p_0-z_\star(p_0)}
>
2\left(p_0-\frac12\right)^2.
\]
It remains to verify the majorization relation.  The first partial sums agree
because the first entry of \(\sigma_\star(p_0)\) is \(p_0\).  For the first two
entries,
\[
p_0+p_1
=
1-\sum_{j\ge2}p_j
\le
1-z_\star(p_0),
\]
which is the sum of the first two entries of \(\sigma_\star(p_0)\).  For
every subsequent partial sum, \((p_j)\) contributes at most its total mass
\(1\), while \(\sigma_\star(p_0)\) has already reached total mass \(1\) in its
first three entries.  This proves \eqref{eq:boundary_majorization}.
\end{proof}

\begin{corollary}[Power-sum and entropy consequences]
\label{cor:ppt_spectral_consequences}
Under the hypotheses of Lemma~\ref{lem:universal_schmidt_spectrum},
\begin{equation}\label{eq:schmidt_purity_half}
\sum_{j\ge0}p_j^2\le\frac12,
\end{equation}
and
\begin{equation}\label{eq:all_atom_amplitude_budget}
\sum_{j\ge1}\sqrt{p_j}\ge\sqrt{p_0}.
\end{equation}
Moreover,
\begin{equation}\label{eq:schmidt_entropy_comparison}
H((p_j)_{j\ge0})
\ge H\!\left(\sigma_\star(p_0)\right),
\end{equation}
with equality only when the two ordered probability vectors coincide.
\end{corollary}

\begin{proof}
For \(p_0\le1/2\), \eqref{eq:comparison_spectrum_definition} gives
\(\sum_j\sigma_{\star,j}(p_0)^2=2(1/2)^2=1/2\).  For \(p_0>1/2\), the last
two nonzero entries of \(\sigma_\star(p_0)\) sum to \(1-p_0\), while their
product is \((p_0-1/2)^2\) by \eqref{eq:zstar_boundary}.  Hence
\[
\begin{aligned}
p_0^2
+\bigl(1-p_0-z_\star(p_0)\bigr)^2
+z_\star(p_0)^2
&=p_0^2+(1-p_0)^2
-2z_\star(p_0)\bigl(1-p_0-z_\star(p_0)\bigr)\\
&=p_0^2+(1-p_0)^2
-2\left(p_0-\frac12\right)^2\\
&=\frac12.
\end{aligned}
\]
Thus
\begin{equation}\label{eq:boundary_purity_half}
\sum_j\sigma_{\star,j}(p_0)^2=\frac12.
\end{equation}
Since \(t\mapsto t^2\) is convex, the majorization relation
\eqref{eq:boundary_majorization} and Karamata's inequality give
\(\sum_jp_j^2\le\sum_j\sigma_{\star,j}(p_0)^2=1/2\), which is
\eqref{eq:schmidt_purity_half}.  Applying Karamata instead to the concave
function \(t\mapsto\sqrt t\) gives
\[
\sum_{j\ge1}\sqrt{p_j}
\ge
\begin{cases}
\sqrt2-\sqrt{p_0}\ge\sqrt{p_0},&p_0\le1/2,\\[2mm]
\sqrt{1-p_0-z_\star(p_0)}+\sqrt{z_\star(p_0)}=\sqrt{p_0},
&p_0>1/2,
\end{cases}
\]
To verify the last equality, square its left-hand side and use
\eqref{eq:zstar_boundary}.  Finally, Karamata's inequality applied to the
strictly concave function \(h(t)=-t\log_2t\) proves
\eqref{eq:schmidt_entropy_comparison}; its strict equality condition gives the
stated equality case.
\end{proof}

\section{Application to the Purification Benchmark}
\label{app:purification_application}

We now apply the preceding Schmidt-spectrum bounds to resource states that
attain the purification benchmark.

\begin{theorem}[Majorization constraint for exact attainment]
\label{thm:gd_lower_bound}
Let \(d\ge2\), \(0<\gamma<1\), and let \(|\eta\rangle\) be any
finite-Schmidt-rank pure resource state.  If a PPT success/failure
instrument assisted by \(|\eta\rangle\) attains the optimal global CPTN value
\(G_\star(D,\gamma)\), then its ordered Schmidt vector satisfies
\[
p_0\le\frac23,
\qquad
(p_j)_{j\ge0}\prec\sigma_\star(p_0).
\]
\end{theorem}

\begin{proof}
By Lemma~\ref{lem:local_twirl_reduction}, the pair may be assumed locally
twirled; retain the symbols \(S,F\) for the averaged branches.

Apply Lemma~\ref{lem:gd_exposed_standard_representative} to the twirled pair.
It gives \(S^{\eta,\eta}=J_\star\).  Compression commutes with the output trace
and with the input-sector Schur compression.  Therefore,
\[
\operatorname{Tr}_{A'B'}J_\star=P_{\rm sym}.
\]
By \eqref{eq:global_input_parity_local_decomposition} and orthogonality of the
local symmetric and antisymmetric projectors,
\[
\begin{aligned}
(P_{\rm sym}^A\otimes P_{\rm asym}^B)P_{\rm sym}&=0,\\
(P_{\rm asym}^A\otimes P_{\rm sym}^B)P_{\rm sym}&=0.
\end{aligned}
\]
Resource insertion commutes with the output trace.  Applying it to the sector
expansion \eqref{eq:output_trace_sector_coefficient_definition} gives
\[
\begin{aligned}
&(P_{\rm sym}^A\otimes P_{\rm asym}^B)
\operatorname{Tr}_{A'B'}S^{\eta,\eta}\\
&\quad=(P_{\rm sym}^A\otimes P_{\rm asym}^B)
\left\langle\eta\left|
\left[\operatorname{Tr}_{A'B'}S\right]_{+-}
\right|\eta\right\rangle,\\
&(P_{\rm asym}^A\otimes P_{\rm sym}^B)
\operatorname{Tr}_{A'B'}S^{\eta,\eta}\\
&\quad=(P_{\rm asym}^A\otimes P_{\rm sym}^B)
\left\langle\eta\left|
\left[\operatorname{Tr}_{A'B'}S\right]_{-+}
\right|\eta\right\rangle.
\end{aligned}
\]
Since \(S^{\eta,\eta}=J_\star\) and
\(\operatorname{Tr}_{A'B'}J_\star=P_{\rm sym}\), the preceding mixed-sector
identities make both left-hand sides zero.  The two projectors are nonzero, so
their scalar coefficients must vanish:
\[
\left\langle\eta\left|
\left[\operatorname{Tr}_{A'B'}S\right]_{+-}
\right|\eta\right\rangle
=
\left\langle\eta\left|
\left[\operatorname{Tr}_{A'B'}S\right]_{-+}
\right|\eta\right\rangle
=0.
\]
Lemma~\ref{lem:ppt_effect_certificate} supplies \(N\).  Its three displayed
effect inequalities are exactly the hypotheses of
Lemma~\ref{lem:universal_schmidt_spectrum}, which gives the claimed
largest-weight and majorization bounds.
\end{proof}

\begin{corollary}[One-ebit lower bound and equality case]
\label{cor:pure_entropy_obstruction}
Under the hypotheses of Theorem~\ref{thm:gd_lower_bound}, the collision entropy
\(H_2(\eta)\) and the Schmidt entropy \(E(\eta)\) are both at least one.  Equality
\(E(\eta)=1\) is possible only for the EPR Schmidt spectrum
\[
(p_0,p_1,p_2,\ldots)=\left(\frac12,\frac12,0,\ldots\right).
\]
\end{corollary}

\begin{proof}
The proof of Theorem~\ref{thm:gd_lower_bound} supplies the effect required by
Corollary~\ref{cor:ppt_spectral_consequences}, so
\(\sum_jp_j^2\le1/2\).
Hence \(H_2((p_j))\ge1\), and the Shannon--collision comparison gives
\(H((p_j))\ge H_2((p_j))\ge1\).  By definition, this proves
\(H_2(\eta)\ge1\) and \(E(\eta)\ge1\).

To characterize the case \(E(\eta)=1\),
\eqref{eq:schmidt_entropy_comparison} and
\eqref{eq:boundary_purity_half} give
\(H((p_j))\ge H(\sigma_\star(p_0))\ge1\).  If
\(p_0>1/2\), then \(\sigma_\star(p_0)\) is nonuniform, so its Shannon entropy
is strictly larger than its collision entropy.  Hence \(E(\eta)=1\) requires
\(p_0\le1/2\), for which
\(\sigma_\star(p_0)=(1/2,1/2,0,\ldots)\).  When \(E(\eta)=1\), equality
holds in \eqref{eq:schmidt_entropy_comparison}, whose equality condition
forces the ordered Schmidt vector to equal this spectrum.
\end{proof}

\begin{corollary}[Sufficiency when \(p_0\le1/2\)]\label{cor:p0_half_suffices}
Let \(\ket\eta\) be a finite-Schmidt-rank pure resource with ordered Schmidt
weights \((p_j)_{j\ge0}\).  If \(p_0\le1/2\), then \(\ket\eta\) suffices to
attain the global benchmark by LOCC.
\end{corollary}

\begin{proof}
By Nielsen's majorization theorem~\cite{nielsen1999conditions}, deterministic
LOCC conversion to one EPR pair is possible exactly when
\[
(p_0,p_1,p_2,\ldots)
\prec
\left(\frac12,\frac12,0,\ldots\right).
\]
The first majorization inequality follows from \(p_0\le1/2\), and all later
ones follow from normalization.  After this deterministic LOCC preprocessing,
Algorithm~\ref{alg:one_epr_purification} attains \(G_\star(D,\gamma)\).
\end{proof}

Thus \(p_0\le1/2\) is sufficient; this argument does not settle attainability
when \(p_0>1/2\).

\paragraph*{Quantitative refinement for \(p_0>1/2\).}
Under the hypotheses of Theorem~\ref{thm:gd_lower_bound},
\(1/2<p_0\le2/3\) implies
\[
\sum_{j\ge2}p_j
\ge z_\star(p_0)>2\left(p_0-\frac12\right)^2,
\qquad
E(\eta)\ge H\!\left(\sigma_\star(p_0)\right)>1.
\]
Equality in the \(p_0\)-dependent entropy bound is possible only for the
three-point spectrum \(\sigma_\star(p_0)\).  These statements follow from
\eqref{eq:quadratic_tail_lower_bound},
\eqref{eq:schmidt_entropy_comparison}, and
\eqref{eq:boundary_purity_half}.

\begin{corollary}[Schmidt entropy is not sufficient]
\label{cor:entropy_not_sufficient}
There exists a finite-Schmidt-rank pure resource with \(E(\eta)>1\) that cannot
attain the global benchmark under a PPT instrument.  It therefore cannot attain
the benchmark by LOCC, so \(E(\eta)\ge1\) is not sufficient.
\end{corollary}

\begin{proof}
Consider the pure resource with ordered Schmidt spectrum
\((p_0,p_1,p_2)=(3/5,39/100,1/100)\).
Its largest Schmidt weight lies in the region not covered by deterministic
conversion to an EPR pair,
\(1/2<p_0=3/5\le2/3\), and
\[
E(\eta)
=
-\frac35\log_2\frac35
-\frac{39}{100}\log_2\frac{39}{100}
-\frac1{100}\log_2\frac1{100}
\approx1.0384>1.
\]
However,
\[
\sum_jp_j^2
=
\frac{2561}{5000}
>
\frac12,
\]
which violates \eqref{eq:schmidt_purity_half}.  Therefore this resource cannot
attain \(G_\star(D,\gamma)\) even under the PPT relaxation.
\end{proof}

\begin{corollary}[Entanglement-of-formation bound for mixed resources]
\label{cor:mixed_resource_formation}
Let \(\omega_R\) be any finite-dimensional mixed resource state.  If
\(\omega_R\) enables a PPT success/failure instrument to attain
\(G_\star(D,\gamma)\), then
\begin{equation}
E_{\mathrm F}(\omega_R)\ge1.
\end{equation}
\end{corollary}

\begin{proof}
Fix a PPT success/failure instrument with uninserted Choi branches \(S,F\)
attaining \(G_\star(D,\gamma)\) with resource \(\omega_R\).  Since its trace
constraint is imposed before resource insertion, inserting any density
operator \(\rho\) on \(R\) yields a feasible branch for the global CPTN
benchmark.  Hence
\(\operatorname{Tr}[(M\otimes\rho^T)S]\le G_\star(D,\gamma)\).

For any pure-state ensemble
\(\omega_R=\sum_aq_a\ket{\xi_a}\bra{\xi_a}\) with \(q_a>0\), linearity gives
\[
\begin{aligned}
G_\star(D,\gamma)
&=\operatorname{Tr}\!\left[(M\otimes\omega_R^T)S\right]\\
&=\sum_aq_a\operatorname{Tr}\!\left[
\left(M\otimes\bigl(\ket{\xi_a}\!\bra{\xi_a}\bigr)^T\right)S
\right]\\
&\le\sum_aq_aG_\star(D,\gamma)=G_\star(D,\gamma).
\end{aligned}
\]
Thus every term with \(q_a>0\) attains the optimum, and
Corollary~\ref{cor:pure_entropy_obstruction} gives \(E(\xi_a)\ge1\).  Hence
\(\sum_aq_aE(\xi_a)\ge1\).  Since the ensemble was arbitrary, taking the
convex roof gives \(E_{\mathrm F}(\omega_R)\ge1\).
\end{proof}

\begin{corollary}[Two-qubit characterization of exact attainment]
\label{cor:two_qubit_resource_characterization}
Fix \(d\ge2\) and \(0<\gamma<1\), and consider the \(2\to1\) purification task
on \(\mathcal S_P\) under the depolarizing channel \(\mathcal N^\gamma\).
Let \(\omega_{A_{\mathrm r}B_{\mathrm r}}\) be an arbitrary two-qubit resource
state.  Then \(\omega\) enables finite-round
LOCC attainment of \(G_\star(D,\gamma)\) if and only if it is a pure
maximally entangled state, equivalently, an EPR pair up to local unitaries.
\end{corollary}

\begin{proof}
Suppose that a mixed two-qubit resource state \(\omega\) enables a
finite-round LOCC protocol to attain \(G_\star(D,\gamma)\).
Its support has dimension at least two.  Every two-dimensional subspace of
\(\mathbb C^2\otimes\mathbb C^2\) contains a pure product
state~\cite[Sec.~II~A and Appendix~A]{niu1999two}.
Consequently, the support of \(\omega\) contains a normalized pure product
state \(\ket{ab}=\ket a\otimes\ket b\).

Choose \(0<p<1\) smaller than the smallest positive eigenvalue of
\(\omega\).  Since \(\ket{ab}\) lies in its support, this ensures
\(\omega-p\ket{ab}\!\bra{ab}\succeq0\).  Thus
\[
\tau:=\frac{\omega-p\ket{ab}\!\bra{ab}}{1-p}\succeq0,
\qquad \operatorname{Tr}\tau=1,
\]
and \(\omega=p\ket{ab}\!\bra{ab}+(1-p)\tau\).
For any pure-state decomposition
\(\tau=\sum_a q_a\ket{\eta_a}\!\bra{\eta_a}\), the definition of
entanglement of formation gives
\[
E_{\mathrm F}(\tau)
\le\sum_a q_a S\!\left(\operatorname{Tr}_{B_{\mathrm r}}
\ket{\eta_a}\!\bra{\eta_a}\right)\le1,
\]
since each qubit reduced state has von Neumann entropy at most
\(\log_2 2=1\).  The pure product state has zero entanglement, so
convexity gives
\[
E_{\mathrm F}(\omega)
\le(1-p)E_{\mathrm F}(\tau)
\le1-p<1,
\]
contradicting Corollary~\ref{cor:mixed_resource_formation}.
If a pure two-qubit resource enables such attainment,
Corollary~\ref{cor:pure_entropy_obstruction} gives an entanglement entropy
of at least one, while its local dimension bounds this entropy by one.
Equality forces Schmidt weights \((1/2,1/2)\).  Conversely, local unitaries
convert such a state to the EPR pair used in
Algorithm~\ref{alg:one_epr_purification}.
\end{proof}

\end{document}